\newif\iffull

\fulltrue

\documentclass[11pt]{article}

\usepackage{setspace}
\usepackage[\iffull notes=true, \else notes=false, \fi later=false]{dtrt}
\usepackage{fullpage}
\usepackage{microtype}
\usepackage{amssymb,amsfonts,amsmath,amsthm}
\usepackage{fullpage}
\usepackage{graphicx}
\usepackage{verbatim}
\usepackage{array}
\usepackage{latexsym}
\usepackage{paralist}
\usepackage{enumitem}
  \setlist[itemize]{leftmargin=*}
  \setlist[enumerate]{leftmargin=*}
\usepackage{dsfont}
\usepackage{url}
\usepackage{braket}
\usepackage{mathrsfs}
\usepackage{color}
\usepackage{complexity}
\usepackage{multirow}
\usepackage{makecell}
\usepackage{mdframed}
\usepackage[margin=1cm,small]{caption}
\usepackage[normalem]{ulem}
\usepackage{bbm}
\usepackage[vlined,linesnumbered,ruled,resetcount]{algorithm2e}
\usepackage[capitalise]{cleveref}
\usepackage{float}

\usepackage{mathbbol}
\usepackage{tcolorbox}

\usepackage{mathtools}

\crefname{algocf}{Algorithm}{Algorithms}
\Crefname{algocf}{Algorithm}{Algorithms}

\newcommand{\nonl}{\renewcommand{\nl}{\let\oldnl}}

\newenvironment{protocol}[1][htb]
  {
   \begin{algorithm}[#1]%
  }{\end{algorithm}}

\newenvironment{transform}[1][htb]
  {
   \begin{algorithm}[#1]%
  }{\end{algorithm}}

\newtheorem{theorem}{Theorem}[section]

\newtheorem{corollary}{Corollary}[section]
\newtheorem{proposition}[corollary]{Proposition}
\newtheorem{definition}[corollary]{Definition}
\newtheorem{lemma}[corollary]{Lemma}
\newtheorem{claim}[corollary]{Claim}

\newtheorem{fact}[corollary]{Fact}

\newtheorem*{lemma*}{Lemma}
\theoremstyle{definition}

\newtheorem{remark}[corollary]{Remark}

\allowdisplaybreaks
\newcommand{\FormatAuthor}[3]{
\begin{tabular}{c}
#1 \\ {\small\texttt{#2}} \\ {\small #3}
\end{tabular}
}
\newcommand{\doclearpage}{%
\clearpage
}

\newcommand{\eps}{\epsilon}
\newcommand{\abs}[1]{\left|#1\right|}
\newcommand{\dist}{\mathsf{dist}}

\newcommand{\Hastad}{\mathsf{Has}}

\newcommand{\Bits}{\{0,1\}}

\newcommand{\N}{\mathbb{N}}
\newcommand{\Field}{\mathbb{F}}
\newcommand{\F}{\Field}
\renewcommand{\H}{\mathcal{H}}
\renewcommand{\D}{\mathcal{D}}
\newcommand{\fhat}{\hat{f}}
\newclass{\Junta}{J}
\newclass{\Fhat}{\hat{F}}

\newcommand{\VC}{\mathsf{VC}}
\newcommand{\err}{\mathsf{err}}
\newcommand{\UCsamples}{\mathsf{m}}
\newcommand{\opt}{\mathsf{opt}}
\newcommand{\proofof}[1]{\normalfont{\textbf{Proof of #1}}}
\newcommand{\rom}[1]{\uppercase\expandafter{\romannumeral #1\relax}}

\DeclareSymbolFont{extraup}{U}{zavm}{m}{n}
\DeclareMathSymbol{\varheart}{\mathalpha}{extraup}{86}
\DeclareMathSymbol{\vardiamond}{\mathalpha}{extraup}{87}

\newcommand{\ip}[1]{{\langle #1 \rangle}}

\DeclareMathOperator*{\argmax}{arg\,max}

\newcommand{\InteractiveParams}[6]{{\left[
{\small
\begin{array}{r  l}
\textsf{input access type:}&\enspace {#1} \\
\textsf{sample/query complexity:}&\enspace {#2} \\
\textsf{communication complexity:}&\enspace {#3} \\
\textsf{round complexity:}&\enspace {#4} \\
\textsf{verifier's running time:}&\enspace {#5}\\
\textsf{prover's running time:}&\enspace {#6}\\
\end{array}
}
\right]}}

\newcommand{\membership}{\textsf{membership queries}}
\newcommand{\random}{\textsf{random examples}}
\newcommand{\proofonly}{\textsf{proof only}}

\newcommand{\Chernoff}{q_\mathsf{CB}}
\mathchardef\mhyphen="2D
\newcommand{\estimate}[1]{\widetilde{#1}}
\renewcommand{\E}{\mathop{\mathbf{E}}}
\newcommand{\GLstar}{\textsf{GL}^\star}

\newcommand{\cAC}{\mathsf{AC}}

\newcommand{\sC}{\mathcal C}
\newcommand{\sD}{\mathcal D}
\newcommand{\sF}{\mathcal F}
\newcommand{\sU}{\mathcal U}
\newcommand{\sA}{\mathcal A}
\newcommand{\sH}{\mathcal H}

\newcommand{\sG}{\mathcal G}
\newcommand{\sR}{\mathcal R}
\newcommand{\sB}{\mathcal B}
\newcommand{\sQ}{\mathcal Q}
\newcommand{\sE}{\mathcal E}

\begin{document}

\title{On the Power of Interactive Proofs for Learning}
\author{ 
\hspace*{-1.15cm}\begin{tabular}[h!]{cccc}
\FormatAuthor{Tom Gur\thanks{TG is supported by UKRI Future Leaders Fellowship MR/S031545/1, EPSRC New Horizons Grant EP/X018180/1, and EPSRC RoaRQ Grant EP/W032635/1.}}{tom.gur@cl.cam.ac.uk}{University of Cambridge} &
\FormatAuthor{Mohammad Mahdi Jahanara}{mjahanar@sfu.ca}{Simon Fraser University} &
\FormatAuthor{Mohammad Mahdi Khodabandeh\thanks{MMK is supported by NSERC Discovery Grant RGPIN/06236-2019.}}{mmk25@sfu.ca}{Simon Fraser University} & \\ \\
\FormatAuthor{Ninad Rajgopal\thanks{NR is supported by Tom Gur's UKRI Future Leaders Fellowship MR/S031545/1. A part of this work was done when the author was affiliated with the University of Warwick.}}{nr549@cam.ac.uk}{University of Cambridge} &
\FormatAuthor{Bahar Salamatian\thanks{This work was done as part of research assistantship at SFU.}}{salamatianbahar@gmail.com}{Qualcomm} &
\FormatAuthor{Igor Shinkar\thanks{IS is supported by NSERC Discovery Grant RGPIN/06236-2019.}}{ishinkar@sfu.ca}{Simon Fraser University} &
\end{tabular}
} %
\date{}

\maketitle

\begin{abstract}
We continue the study of doubly-efficient proof systems for verifying agnostic PAC learning, for which we obtain the following results. 

\begin{itemize}
  \item We construct an interactive protocol for learning the $t$ largest Fourier characters of a given function $f \colon \Bits^n \to \Bits$ up to an arbitrarily small error, wherein the verifier uses $\poly(t)$ random examples.
  This
  improves upon the \emph{Interactive Goldreich-Levin} protocol of \cite{GoldwasserRSY21},
  whose sample complexity is $\poly(t,n)$.

  \item For agnostically learning the class $\cAC^0[2]$ under the uniform distribution,
  we build on the work of \cite{CIKK17_agnostic} and design an interactive protocol, where given a function $f \colon \Bits^n \to \Bits$, the verifier learns the closest hypothesis up to $\polylog(n)$ multiplicative factor, using quasi-polynomially many \emph{random examples}. In contrast, this class has been notoriously resistant even for constructing realisable learners (without a prover) using random examples.

  \item For agnostically learning $k$-juntas under the uniform distribution, we obtain an interactive protocol, where the verifier uses $O(2^k)$ random examples to a given function $f \colon \Bits^n \to \Bits$. Crucially, the sample complexity of the verifier is independent of $n$.
\end{itemize}
We also show that if we do not insist on doubly-efficient proof systems,
then the model becomes trivial.
Specifically, we show a protocol for an arbitrary class $\sC$ of Boolean functions in the distribution-free setting, where
the verifier uses $O(1)$ labeled examples to learn $f$.
\end{abstract}

\newpage
\tableofcontents
\doclearpage

\section{Introduction}
Can we verify the results of a supervised learner much cheaper than running the learning algorithm from scratch? This presents a fundamental question at a time where training modern machine learning models demands vast amounts of data, meticulously collected and labeled, coupled with substantial computational resources. The infeasibility of repeating this process multiple times, raises the criticality of verifying the results of the learning algorithm. 

Recently, Goldwasser, Rothblum, Shafer and Yehudayoff \cite{GoldwasserRSY21} formalised the question of verifying the results of supervised machine learning tasks  in the context of PAC Learning \cite{Valiant1984}. They initiated the study of \emph{PAC-verification}, with the aim of developing interactive proof systems that enable a verifier to check the results of an untrusted learner (or a prover). The goal is to achieve this verification while conserving computational resources and reducing data access, either quantitatively or qualitatively. This notion of utilising \textit{interactive proofs} to verify learning algorithms, is rooted in the understanding that interaction with a prover can be remarkably helpful for many computational tasks, evidenced by major results like $\mathsf{IP} = \mathsf{PSPACE}$ or $\mathsf{MIP}^*= \mathsf{RE}$. 

In more detail, \cite{GoldwasserRSY21} focus on \emph{agnostic learning}. In the setting of agnostic learning over the \textit{uniform distribution} for a class of Boolean functions $\sC$, for an arbitrary $f : \{0,1\}^n \rightarrow \{0,1\}$, the learning algorithm has access to a \textit{random example oracle} that provides labeled examples of the form $(x, f(x))$, where each $x$ is drawn uniformly at random from $\{0,1\}^n$. A learner $(\alpha,\eps,\delta)$-agnostically learns a class $\sC$ over the uniform distribution using random examples, for some $\alpha \geq 1$, if it outputs a hypothesis $h$ with error at most $\alpha \cdot \mathsf{opt}_\sC(f) + \varepsilon$ with probability at least $1-\delta$, where $\mathsf{opt}_\sC(f) = \min_{c \in \sC} \{\dist(f,c)\}$ is the best possible approximation of $f$ by any function in $\sC$. Another well-studied variant is one where the learner is required to meet the requirements of $(\alpha,\eps,\delta)$-agnostic learning, but is given a stronger form of access to $f$ via \textit{membership queries}, and not just random examples. Agnostic learning goes beyond ``realisable" PAC learning, which assumes that $f$ always comes from $\sC$; instead, it represents a more realistic scenario where there is no ``ground truth" about the labeling function (see \cref{sec:agnostic_learn_prelim} for a formal definition).  

In the $(\alpha, \eps, \delta)$-PAC-verification model over the uniform distribution, a verifier, which accesses $f$ using a random example oracle over the uniform distribution, interacts with an untrusted prover (or learner) that has query access to $f$, to output a hypothesis $h$ with error at most $\alpha \cdot \mathsf{opt}_{\sC}(f) + \varepsilon$, with probability at least $1-\delta$. For the \textit{completeness requirement} of the proof system, the honest prover should be able to convince the verifier to output a ``good" hypothesis with high probability (this could also include scenarios where the prover sends a purported hypothesis first and then tries to convince the verifier about its accuracy). For the \textit{soundness requirement}, a malicious prover, even if computationally unbounded, should not be able to convince the verifier to accept a hypothesis that deviates significantly from $\alpha \cdot \mathsf{opt}_{\sC}(f)$, with high probability. This represents a PAC-verification task with the goal of achieving a \textit{qualitative} difference between the verifier and the prover, as queries are considered to be a more powerful form of access.

Another equally appealing setting of PAC-verification is one where both verifier and prover have the query access (or random example access) to $f$, and the goal is for the verifier to output a good hypothesis with high probability, making much fewer queries (or samples) than the prover, providing a \textit{quantitative} difference. We refer to \cref{sec:agnostic_ip_model_formal} for a formal definition of the model. 

Following the works of \cite{GoldwasserRSY21} and its follow-ups \cite{MS23,CHINS23}, we are mainly interested in proof systems that are \textit{doubly efficient}, where the honest prover is also an efficient algorithm. The notion of doubly efficient proof systems is intimately connected to the notion of \textit{delegating} a computational task to an efficient, yet untrusted prover \cite{GoldwasserKR08}. While such proof systems have found widespread use theoretically and in practice, their study in the setting of delegating learning tasks is nascent.

In this work, we establish a set of general tools for PAC-verification and illustrate their utility by constructing PAC-verification protocols for certain classes of Boolean functions, whose study is foundational in the field of computational learning theory. Our constructions showcase the power of interacting with a prover, by achieving quantitative and/or qualitative improvements for verifying the results of an agnostic learner, than actually performing the learning task without a prover. 

\subsection{Our Results}\label{sec:our-results}

To begin with, we construct a sample-efficient \textit{interactive Goldreich-Levin protocol}, for the problem of learning heavy Fourier coefficients, which improves upon an analogous result by \cite{GoldwasserRSY21}. Following this, we construct the first PAC-verifier for classes of functions computable by \textit{constant-depth circuits} like $\cAC^0[2]$, as well as for functions that are \textit{$k$-juntas}. Finally, we illustrate the power of such proof systems with \textit{unbounded honest provers}, with PAC-verifiers for any arbitrary class of Boolean functions, using very few samples, even in the distribution-free setting. We elaborate on these results next.

\subsubsection{Learning Heavy Fourier Characters}
\label{sec:our-results-heavy}

One way of understanding the structure of a class of Boolean functions $\sC$ is by studying the \textit{Fourier spectrum} of functions in $\sC$ (see Section \ref{sec:boolean-functions} for definitions about Fourier analysis of Boolean functions). This idea was further made concrete in the context of learning theory, through general results that show efficient learnability of classes for which the Fourier spectrum is concentrated on low-degree monomials \cite{LMN93}, or concentrated on a sparse set \cite{KushilevitzM93}.

In particular, \cite{KushilevitzM93} used the Goldreich-Levin algorithm \cite{GoldreichL89} for computing the set of \textit{all heavy Fourier characters} of a Boolean function $f$ (i.e., characters with large Fourier coefficients), given \textit{membership query access} to it.\footnote{We view a character in the Fourier expansion of an $n$-variate Boolean function as a vector in $\{0,1\}^n$. More precisely, for any fixed $a \in \Bits^n$, define the character as an $n$-variate Boolean function $\chi_a(x) = (-1)^{\ip{a,x}}$, where $\ip{a,x}$ is the inner product of $a$ and $x$ over modulus $2$.} Following these results, the ability to learn heavy Fourier characters of a Boolean function is used in multiple works in learning theory (as well as other areas like cryptography and coding theory). 

For any $n$-variate Boolean function, $t \in \N$ and $\eps > 0$, we say that a set of characters $S$ is $\eps$-close to the set of $t$ heaviest Fourier characters $\Lambda_t$ of $f$, if no character from $S$ with the smallest (absolute) coefficient is replaceable by any character outside $S$ whose (absolute) coefficient is larger by a value of $\eps$. Note that, this implies that for each $U \in \Lambda_t$, there exists a distinct $\tilde{U} \in S$, such that $\vert \hat{f}(U) \vert - \vert \hat{f}(\tilde{U}) \vert \leq \eps$, where $\hat{f}(U)$ and $\hat{f}(\tilde{U})$, are the respective Fourier coefficients. For our first result, we show an algorithm that learns the $t$-heaviest Fourier characters of a given Boolean function up to an error $\eps$, where the verifier only uses $\poly(t/\eps)$ random labeled examples. Specifically, we have

\begin{theorem}[Learning heavy Fourier characters (See \cref{thm:fourier-random} for formal statement)]
    \label{thm:our-results-topfourier}
    There exists an interactive proof, such that for any $f : \{0,1\}^n \rightarrow \{0,1\}$, any $\eps > 0$ and any $t \in \N$, the verifier uses at most $\poly(t/\eps)$ many random examples and outputs a set $\Tilde{\Lambda_t}$ that $\eps$-approximates $\Lambda_t$ with probability at least $0.9$. 
    
    Moreover, the prover is efficient in the sense that it is essentially as efficient as the learning algorithm it invokes.
\end{theorem}
We emphasize that Theorem \ref{thm:our-results-topfourier} presents a \textit{qualitative, as well as a quantitative improvement} over the prover that runs the Goldreich-Levin algorithm, and thus, the query complexity is $\poly(n,1/\eps)$. As an example, for an application where we just want to find the character that approximates the heaviest Fourier coefficient of $f$ up to a small constant error, then the verifier only needs $O(1)$ random examples to learn it using a prover, whereas without a prover it needs to make $\poly(n)$ queries to $f$!

The interactive Goldreich-Levin theorem is obtained as an application of \cref{thm:our-results-topfourier}. For any $\gamma > 0$, let $f^{\geq \gamma/2}$ be the set of Fourier characters whose coefficients are at least $\gamma/2$. From Parseval's theorem (see \cref{fact:parseval}), we see that $\vert f^{\geq \gamma/2}\vert \leq 4/\gamma^2$. To find $f^{\geq \gamma}$, we set $\Lambda_t = f^{\geq \gamma/2}$ in \cref{thm:our-results-topfourier}, where $t = 4/\gamma^2$ and $\eps = \gamma/2$. The verifier outputs a list $L$ of characters that contain $f^{\geq \gamma}$, with high probability, using only $\poly(1/\gamma)$ many random examples, whereas the prover makes at most $\poly(n/\gamma)$ queries.\footnote{Indeed, by definition, $\Lambda_t$ contains every character whose Fourier coefficient is larger than $\gamma/2$. If there was a $U$ such that it's Fourier coefficient $\vert \hat{f}(U) \vert \geq \gamma$ and $U$ is not in the output $L$, from the fact that \cref{thm:our-results-topfourier} returns exactly $t$ characters, $L$ contains a character $\hat{U}$ such that $\vert \hat{f}(\hat{U}) \vert < \gamma/2$. Then, $\vert \hat{f}(U) \vert - \vert \hat{f}(\hat{U}) \vert > \gamma/2$, which contradicts the fact that $L$ be a $\gamma/2$-close to $\Lambda_t$.} We highlight that our interactive Goldreich-Levin protocol vastly improves the sample complexity over that of \cite{GoldwasserRSY21} (Lemma 2.2), which has a sample complexity of $\poly(n,1/\gamma)$. In particular, the sample complexity of our interactive Goldreich-Levin protocol is \textit{independent of $n$}.

Theorem \ref{thm:our-results-topfourier} builds upon a novel algorithm that (approximately) computes the highest $t$ \textit{Fourier coefficients} (but not the associated characters), by making $\poly(t/\eps)$ membership queries to $f$, that in turn gives a PAC-verification protocol in which the verifier uses much fewer queries ($\poly(t/\eps)$) than the prover. We then show a \textit{framework for query-to-sample reduction} for PAC-verification, based on the query pattern of the verifier, and apply it to this protocol. We find these tools to be of independent interest and potentially applicable for designing other PAC-verification protocols, and provide more intuition for that in \cref{sec:tech-overview-igl-juntas}.

\subsubsection{Learning $\cAC^0[2]$ Circuits}
\label{sec:our-results-ac0[2]}
Let $\cAC^0[2]$ be the class of functions computable by constant-depth circuits of polynomial-size with $\mathsf{AND,OR,NOT,XOR}$ gates of unbounded fan-in. Similarly, for any prime $p > 2$, we define $\cAC^0[p]$ be the class of functions computable by constant-depth circuits of polynomial-size (number of gates) with $\mathsf{AND,OR,NOT}$ and $\mathsf{MOD}_p$ gates of unbounded fan-in. 

Suppose the random example oracle over the uniform distribution is such that $\opt(f,\cAC^0[2])$ is non-negligible (i.e., it is at least $1/n^{\log(n)}$). In the following theorem, we show that for such functions, $\cAC^0[2]$ can be $\polylog(n)$-PAC-verified using quasi-polynomially many random examples.\footnote{For convenience, we drop $\eps$ from $(\alpha, \eps, \delta)$-PAC-verifiability, when $\mathsf{opt}(f,\sC)$ is non-negligible. Similarly, when $\alpha = 1$, we remove it from the notation.} 
\begin{theorem}[PAC-verification for $\cAC^0${$[p]$} (See \cref{thm:agnostic_ip_ac0} for a formal statement)]
    \label{thm:our-results-ac0[2]-pac-verify}
    Let $f : \{0,1\}^n \rightarrow \{0,1\}$ be any function such that $\opt(f,\cAC^0[2]) > 0$ is non-negligible. 
    
    Then the class $\cAC^0[2]$ is $(\polylog(n), 1/10)$-PAC-verifiable over the uniform distribution, where the verifier uses at most quasi-polynomially many random examples. Moreover, the protocol is doubly efficient, where both the verifier and the honest prover (that has query access to $f$) run in quasi-polynomial time.
\end{theorem}
Moreover, we can easily extend this to get PAC-verification for $\cAC^0[p]$ (\cref{thm:aip_ac0_any_prime}) with similar complexities, for any prime $p > 2$.

The class $\cAC^0[2]$, or more generally $\cAC^0[p]$, is at the frontiers of Boolean function classes that are learnable. In the \textit{realisable} setting, \cite{CIKK16} obtain a learner for $\cAC^0[p]$ over the uniform distribution using membership queries in quasi-polynomial time, getting a learner for these classes using random examples has been a long-standing open question. Moreover, removing the membership queries in the agnostic learner for $\cAC^0[2]$, would give quasi-polynomial time algorithms for two notoriously difficult problems: learning parities with noise (LPN), and for $\cAC^0[p]$, learning with errors (LWE). We refer to \cite{CIKK17_agnostic} (section 5) for a further discussion on implications to LPN or LWE. 

In contrast, we show the power of PAC-verification protocols for $\cAC^0[2]$, where the verifier uses random examples to agnostically learn $\cAC^0[2]$, upon interaction with a quasi-polynomial time prover. This also generalises another result by \cite{GoldwasserRSY21}, who show that parities are PAC-verifiable, thus implying that the LPN assumption might not be relevant in the setting of PAC-verification.

Our techniques crucially use the structure of the agnostic learner with membership queries by \cite{CIKK17_agnostic}. In particular, they use tools from pseudo-randomness and meta-complexity in the construction of the learner (the Nisan-Wigderson reconstruction algorithm \cite{NW94,IW01}). We perform a careful analysis of the query patterns in this algorithm, and extend the query-to-sample reduction framework to use this. Considering the ubiquity of the Nisan-Wigderson generator in complexity theory, we find this analysis to have other potential applications.

We briefly mention some subtleties about \cref{thm:our-results-ac0[2]-pac-verify}. Our protocol builds over the \cite{CIKK17_agnostic} agnostic \textit{improper} learner (hypothesis is not in $\cAC^0[2]$), which translates into our sample and running time complexities, as well as the hypothesis approximation error of our PAC-verification protocol. Indeed, we also inherit the non-negligibility criterion over $\mathsf{opt}(f,\cAC^0[2])$ from them, which ensures quasi-polynomial sample and running time complexities. 

\paragraph*{General Boolean Circuit Classes.} Let $\sC$ be any well-known Boolean circuit class like $\mathsf{ACC}^0, \mathsf{NC}^1$, or $\Ppoly$. The agnostic learner for $\cAC^0[p]$ over the uniform distribution by \cite{CIKK17_agnostic} is in fact obtained via the meta-algorithmic ``circuit lower bounds to learning algorithms" framework. 

In more detail, a \textit{$\tau(n)$-tolerant natural property} (in the sense of \cite{DBLP:journals/jcss/RazborovR97}) against functions that are $\tau(n)$-close to $\sC$ (i.e., $\mathsf{opt}(f,\sC) \leq \tau(n)$), is an ``efficient" algorithm that distinguishes functions that are $\tau(n)$-close to $\sC$ from a dense subset of all functions (see \cref{sec:tol_nat_prop} for formal definitions). Tolerant natural properties refer to constructive average-case circuit lower bound techniques, which not only show a lower bound for a given function against $\sC$, but also extend the lower bound against $\sC$ to a non-negligible fraction of all functions. \cite{CIKK17_agnostic} show that a $(1/2-\tau(n))$-tolerant natural property against $\sC$ implies agnostic learners for $\sC$ using membership queries over the uniform distribution, in a very general and black-box fashion. 

The generality of our techniques used in the PAC-verification protocol for $\cAC^0[2]$ allows it to be extended to the ``tolerant natural properties to agnostic learners" framework. Thus, we show conditional PAC-verifiers for any typical circuit class $\sC[\poly(n)]$ (functions computable by $\sC$-circuits of polynomial size).\footnote{We consider $\mathsf{ACC}^0, \mathsf{TC}^0, \mathsf{NC}^1$ or $\Ppoly$, as some typical Boolean circuit classes.} In more detail, we show that if there exists a $(1/2-1/\poly(n))$-tolerant natural property against $\sC[\poly(n)]$, then $\sC[\poly(n)]$ is $(n^\gamma,1/10)$-PAC-verifiable over the uniform distribution, for $0 < \gamma < 1$ where the verifier uses at most sub-exponentially many random labeled examples. Moreover, both the verifier and the honest prover run in sub-exponential time. We refer to \cref{thm:tolerant_natural_implies_aip} for a formal statement.

\subsubsection{Learning Juntas}
\label{sec:our-results-juntas}
Let $f : \{0,1\}^n \rightarrow \{0,1\}$ be a Boolean function that depends only on an \textit{unknown} subset of $k \ll n$ variables. The class of such functions are called $k$-\textit{juntas}. The class of $k$-juntas have been widely explored in testing and learning literature, eg., \cite{MosselOS03,DeMN19,CNY23,HS23} and questions about the complexity of learning it still remains unclear, in both realisable and agnostic settings. We construct the following PAC-verification protocol for $k$-juntas.

\begin{theorem}[PAC-verification for $k$-juntas (see \cref{thm:formal-statement-junta-random} for a formal statement)]
    \label{thm:our-results-juntas-informal}
    For any integer $k \leq n$, the class of functions $k$-juntas is $(\eps,1/10)$-PAC-verifiable with respect to the uniform distribution with high probability. 
    
    Moreover, the verifier uses at most $2^k \cdot \poly(k/\eps)$ random examples, and the honest prover is efficient in the sense that it is essentially as efficient as \emph{any} $k$-junta learning algorithm that it invokes.
\end{theorem}
While the honest prover runs an agnostic learner for $k$-juntas from \cite{HS23}, and uses at most $k2^k/\eps^2 \cdot \log(n)$ random examples. On the other hand, the verifier in the protocol from \cref{thm:our-results-juntas-informal} uses $2^k \cdot \poly(k/\eps)$ random examples. It is worth stressing the quantitative improvement here, where the verifier sample complexity is \textit{independent} of $n$, unlike the honest prover. For example, a PAC-verifier for an $O(1)$-junta only uses $O(1)$ random examples to output a hypothesis with a constant error. 

The PAC-verifier for $k$-juntas is obtained by another general transformation from query-based tolerant testers for a class $\sC$ to PAC-verifiers for $\sC$, followed by an application of the query-to-sample reduction from \cref{sec:our-results-heavy} for this protocol.\footnote{A tolerant tester for $\sC$ is a sub-linear algorithm that accepts an input $f$, if $\opt(f,\sC) \leq \beta$, and rejects it, if $\opt(f,\sC) > \beta + \eps$, for any input parameters $\beta, \eps > 0$. Clearly a tolerant tester implies a suitable tolerant natural property for $\sC$, since a random function is far from $\sC$ (for a typical circuit class). 

However, it is unclear if a tolerant natural property gives us a tolerant tester, which has a more stringent rejection criterion. In particular, the results of \cite{parnas2006tolerant} only imply a tolerant tester from a proper agnostic learner for $\sC$ (the hypothesis is from $\sC$), and the learner from \cite{CIKK17_agnostic} is improper. This may indicate the reason for the difference in hypothesis error for PAC-verifiers obtained using our techniques - additive error (tolerant tests) vs multiplicative error (tolerant natural properties).}

\subsubsection{The Power of PAC-Verification with Unbounded Provers}
Finally, we study the power of PAC-verification if we allow the honest prover to have unbounded computational power. For this section, we extend PAC-verification model from \cref{def:PAC-ver-formal} to the \textit{distribution-free} agnostic learning setting, where for any unknown, yet fixed distribution $\sD$ over the examples, the verifier is now given access to the random example oracle with respect to $\sD$.
The completeness and soundness requirements are stronger; for every fixed, but unknown underlying distribution $\sD$, the hypothesis output by the interaction should have error at most $\opt_{\sD}(f,\sC) + \eps$, where the distance is now measured over $\sD$.

Let $\Ppoly$ be the class of functions computable by general Boolean circuits of polynomial size, that captures efficient non-uniform computation. We show the following PAC-verification protocol.
\begin{theorem}[PAC-verification for $\Ppoly$ (see \cref{thm:aip-ppoly-erm} for a formal statement)]
    \label{thm:our-results-df-pac-verify-ppoly}
       For any $\varepsilon > 0$, $\Ppoly$ is \textit{distribution-free, proper,} $(\eps,1/10)$-PAC-verifiable, where the verifier uses at most $O(1/\eps)$ many random labeled examples and runs in $\poly(n/\eps)$ time. Moreover, the honest prover is computationally unbounded.
\end{theorem}

For contrast, assuming the existence of standard cryptographic primitives (like one-way functions), we know that $\Ppoly$ cannot be learnt in polynomial-time over the uniform distribution using membership queries, even in the realisable setting \cite{GGM86}. Thus, providing access to an unbounded prover, can help a verifier learn $\Ppoly$ even in the significantly stronger distribution-free, proper, agnostic learning setting.

\cref{thm:our-results-df-pac-verify-ppoly} is proved using a more general result that shows that any arbitrary class of functions $\sC$ is distribution-free PAC-verifiable using $O(1/\eps)$ random examples, where the verifier runs in $\poly(\log(\vert \sC \vert)/\eps)$. The interactive proof is constructed by viewing agnostic learning as a suitable \textit{Empirical Risk Minimisation} task (ERM) (which does an exhaustive search) and delegating this computational task to the prover using \cite{GoldwasserKR08}. More details can be found in \cref{sec:aip_erm_unbounded_prover}.

\subsection{Technical Overview}
\label{sec:intro-tech-overview}
In this section, we highlight the proofs of Theorems \ref{thm:our-results-topfourier} and \ref{thm:our-results-ac0[2]-pac-verify}. Our interactive proofs are constructed using a sequence of novel ideas for obtaining general transformations, and below, we highlight some of the key ideas used to obtain them. These interactive proofs are constructed through an interplay of various techniques and tools from interactive proofs, property testing, Fourier analysis, and pseudorandomness. We refer to the individual technical sections for further details.

\subsubsection{Proof Outline of Theorem \ref{thm:our-results-topfourier}}
\label{sec:tech-overview-igl-juntas}
We start by briefly recalling the Goldreich-Levin algorithm. Then, we present our approach for query-based PAC-verification of heavy Fourier characters. Finally, we outline how to derive sample-based PAC-verification of heavy Fourier characters via our query-to-sample reductions.

\paragraph*{Overview of the Goldreich-Levin algorithm (\textsf{GL}).} For the outline, we focus on the goal of learning the set $\Lambda_t$ of \textit{all} the $\gamma$-heavy Fourier characters (i.e., vectors $a \in \Bits^n$ such that the associated Fourier coefficient $\vert \hat{f}(a) \vert \geq \gamma$), given \textit{query access} to the input function $f : \Bits^n \rightarrow \{1,-1\}$ (for the output, we use the equivalent representation of $\Bits$ as $\{-1,1\}$). With additional work, we generalise the ideas highlighted below for PAC-verifying $t$ heavy Fourier characters, for any $t$, thus obtaining \cref{thm:our-results-topfourier}.

\vspace{0.1in}
This algorithm follows a divide-and-conquer strategy over the inputs, that can be viewed as a binary tree (see Section 3.5 in \cite{odonnell_analysis_2014}). For any $1 \leq i \leq n$, and any node at depth $i$ in the tree (that represents a subcube of $\Bits^n$), \textsf{GL} considers both subcubes obtained from this node by restricting the $(i+1)^{\text{th}}$-coordinate to either $0$ or $1$. It then estimates the Fourier weights of both subcubes up to a small error using $O(1/\gamma^2)$ many queries to $f$, and eliminates those with Fourier weight at most $\gamma^2/4$, from future consideration.\footnote{The Fourier weight of a subspace is the sum of squared Fourier coefficients of all vectors in the subspace. This can be estimated using standard techniques from Fourier analysis of Boolean functions \cite{odonnell_analysis_2014}.} \textsf{GL} stops at depth $n$, where each leaf consists of just one vector, and thus the set of leaves contains $\Lambda_t$. Of course, since there are at most $O(1/\gamma^2)$ many subcubes with large Fourier weight (by Parseval's theorem) at any depth of the tree, \textsf{GL} roughly makes $\poly(n/\gamma)$ queries overall. It is crucial that \textsf{GL} needs to go all the way down to depth $n$ in order to identify all the heavy vectors in learn $\Lambda_t$. 

\paragraph*{Query-based PAC-verification of $\Lambda_t$.} Suppose we could stop this process at depth $O(\log(1/\gamma))$ instead of going all the way to $n$, with the hope of getting the query  complexity down to $\poly(1/\gamma)$. We already encounter a significant challenge: since each subcube at any depth is obtained by partitioning over a fixed variable in the layers above, it could contain multiple $\gamma$-heavy vectors. Furthermore, \textsf{GL} only estimates the Fourier weight of each subcube, and does not help us with the identifying the actual heavy vectors within.

To resolve this, we consider a partitioning scheme based on \textit{random linear functions}, that also appears in \cite{gopalan2011testing}. In more detail, we take a set $\{r_1, \dots, r_s \}$ of uniformly random vectors from $\F_2^n$, where $s = O(\log(1/\gamma))$ is suitably large. Define $S$ as $\mathsf{span}(r_1, \dots, r_s)$ of dimension $s$ (with high probability). First, we show that the $2^s$ possible affine shifts of $S^{\bot}$ (the subspace orthogonal to $S$) form a partition of $\F_2^n$ such that each of these subspaces contains \textit{at most one} $\gamma$-heavy vector.\footnote{For comparison, \cite{gopalan2011testing} study testing sparse Boolean functions and show under such a partitioning scheme, at most one non-zero Fourier coefficient of a sparse Boolean function falls in each subspace. In such a case, computing the Fourier $L_2$-norm of each subspace suffices for their testing algorithm. While they work with the decision task of testing, in our learning/Goldreich-Levin setting, we get an arbitrary Boolean function without the nice structure provided by sparsity, that could render any such subspace to contain a large number of non-zero Fourier coefficients. It is unclear how estimating the $L_2$-norm of the entire subspace would help the learner identify the heaviest character (up to some error).}

Following this, we prove a new method of estimating the heaviest vector from each of these affine subspaces based on the following lemma: for any affine subspace $\Gamma$ which is ``rare" for $\gamma$-heavy vectors, the $L_4$-Fourier norm is an additive approximation to the $L_\infty$-Fourier norm, i.e., $\sqrt[4]{\sum_{a \in \Gamma} \hat{f}(a)^4}$ is roughly equal to $\argmax_{a\in \Gamma}|\hat{f}(a)|$ with a small additive error. Using the $L_4$-Fourier norm, as opposed to the $L_2$-Fourier norm utilized in \cite{gopalan2011testing}, seems crucial for estimating the heaviest Fourier coefficient. Now, if we estimate the $L_4$-Fourier norm of each of these affine shifts, we get all the coefficients for the $\gamma$-heavy vectors. In \cref{lem:fourth-power-sum} by proving that the sum $\sum_{a\in\Gamma}\hat{f}(a)^4$ can be expressed as an expectation, we show that this estimate can be made, using $\poly(1/\gamma)$ queries to $f$. Thus, we get an algorithm that makes $\poly(1/\gamma)$ queries to $f$ and \textit{computes the coefficients} for all the $\gamma$-heavy vectors, without identifying the corresponding characters.\footnote{One can view this as a full binary tree, where the nodes in layer $i$ represent all the $2^i$ affine shifts of the subspace orthogonal to $\mathsf{span}(r_1, \dots, r_i)$. It is worth noting that unlike the iterative process used by \textsf{GL}, we only compute the $L_4$-estimates at the end.} 

As a solution to this, we just ask the prover (which runs \textsf{GL}) to send us the $\gamma$-heavy vectors! The honest prover sends the $\gamma$-heavy vectors to the verifier, and completeness is ensured since we can easily estimate the Fourier coefficient of a vector using $O(1/\gamma^2)$ many random queries (up to an error, say $\gamma/4$). For soundness, we observe that if a malicious prover sends a coefficient $a'$ that is not heavy, we can efficiently find the affine shift of $S^{\bot}$ it belongs to, since $S$ has a succinct representation using just $s$ vectors. Using the $L_4$-norm based estimation technique above, we can find the maximum Fourier coefficient $\tilde{w}$ in this affine subspace and reject, if $\Tilde{w}$ is too far from $\vert \hat{f}(a') \vert$. Put together, we obtain a PAC-verification protocol for learning $\Lambda_t$, where the verifier makes $\poly(1/\gamma)$ queries and the honest prover makes $\poly(n/\gamma)$ queries to $f$, obtaining a significant \textit{quantitative improvement in the number of queries}, that is independent of $n$.

For the sake of exposition, we ignore the effect of the additive errors on various estimates made by the verifier as well as the honest prover that runs \textsf{GL}; accounting for this gives us a list of coefficients of which the set of all $\gamma$-heavy vectors is a subset. We refer to Sections \ref{sec:alg-for-learning-fourier-values} and \ref{sec:pac-verification-top-fourier-coefficients} for details.

\paragraph*{Query-to-sample reductions.} In order to obtain a qualitative improvement as well, we need the verifier in the aforementioned PAC-verification protocol to run using only random examples over the uniform distribution, without having a large blow-up on the sample complexity. Our main challenge is that query access is typically a more powerful resource than just random labeled examples to $f$ (formalised by \cite{Fel09} for agnostic learning over uniform distribution). We prove that for a PAC-verification protocol where the verifier has query access to $f$, the verifier can use \textit{only random examples} and a prover to answer them, when the queries have suitable structure.

To understand this framework, we make the following observations about the queries made by the verifier in the aforementioned PAC-verification protocol for $\Lambda_t$. Firstly, we observe that the \textit{marginal} distribution of every query is uniform. In more detail, to estimate the $L_4$-norm of an affine subspace using \cref{lem:fourth-power-sum}, the verifier queries on uniformly random $x,y,z \sim \Bits^n$, and on $w$ which is the sum of these vectors with a uniformly drawn vector from the random subspace $S$. Next, we see that the queries are non-adaptively generated by the verifier, i.e., there exists a query construction algorithm $V_{\mathsf{query}}$ that uses its internal randomness to generate all the queries before they are made. Further, the queries are independent of the proof sent by the prover. 

Suppose that the verifier makes $Q$ queries to $f$. Now, a verifier picks a query index $j \sim [Q]$ uniformly at random and ``hides" a random example $(w,f(w))$ from the uniform distribution as the $j^{\text{th}}$ query. The important thing here is that $V_{\mathsf{query}}$ can generate the rest of $Q-1$ queries, \textit{conditioned} on the $j^{\text{th}}$ query being $w$. Indeed, the non-trivial case here is if this query corresponds to one that uses a random vector from the subspace $S$, that is constructed beforehand by $V_{\mathsf{query}}$ using $s$ random vectors; we prove in \cref{lem:linear-query-pattern} that $V_{\mathsf{query}}$ can still generate the rest of the queries conditioned on this query being $w$.

For the reduction, the verifier sends the entire \textit{embedded query set} to the prover and asks for the query answers back. The honest prover answers correctly and the completeness follows from that of the query-based protocol. On the other hand, since the marginal distribution of each query is uniform, which is the same as the underlying distribution, a malicious prover $P$ can't distinguish between the ``embedded" query $w$ and any other query which is generated by $V_{\mathsf{query}}$. Thus, the best $P$ can do is pick a query at random and answer it incorrectly hoping that it is not $w$, while providing correct answers everywhere else (note that $V$ has $f(w)$ and will catch a prover if it lies on $w$). A malicious prover can get away with cheating on a single query with probability $(1-1/Q)$; repeating this process $O(Q)$ number of times ensures that we catch the prover with high probability. Finally, we get a PAC-verification protocol for learning $\Lambda_t$ where the verifier makes just $O(Q^2)$ samples.

\paragraph{Discussion.} More generally, if the query construction algorithm of the verifier satisfies \textit{embeddability}, i.e., the query set can be generated, conditioned on a random query index being fixed to a given query $w$, as well as the marginal distribution of each query being the same as the underlying distribution, then we get a PAC-verification protocol where the verifier only uses random examples, with sample complexity having a quadratic blow-up. We refer to \cref{sec:query-to-sample-reduction} for more details.

While embeddability seems necessary for the hiding process, \cref{sec:tech-overview-ac0[2]} extends this framework towards handling more general query marginal distributions. Moreover, we use non-adaptivity to embed the random example in the query set, and we leave the study of extending this framework to handle adaptive queries, which may or may not depend on the interaction with the prover, as an interesting direction for future work.

\subsubsection{Proof Outline of \cref{thm:our-results-ac0[2]-pac-verify}}
\label{sec:tech-overview-ac0[2]}

Our interactive proofs in \cref{thm:our-results-ac0[2]-pac-verify} crucially relies on the \cite{CIKK17_agnostic} agnostic learner, which we recall next. Subsequently, we show how to extend our framework of query-to-sample reductions to support the query patterns of the Nisan-Wigderson reconstruction algorithm that underlies the \cite{CIKK17_agnostic} agnostic learner.

\paragraph*{Overview of the \cite{CIKK17_agnostic} agnostic learner.} We start by recalling the \textit{Nisan-Wigderson (NW)} generator. A family of sets $S_1, \dots, S_L$ over a universe of size $m$ is called an NW-set design, if each $S_i$ has size $n$ and for every $i \neq j$, the overlap between $S_i$ and $S_j$ is very small, i.e., is at most $\log(L)$. Given this, for any $f : \Bits^n \rightarrow \Bits$, we define the NW-generator $G^f : \{0,1\}^m \rightarrow \Bits^L$, with seed length $m$ and stretch $L$, as $G^f(z) = f(z \mid_{S_1}), \dots, f(z \mid_{S_L})$, where $z \vert_{S_i}$ is the restriction of the seed $z$ to the indices in $S_i$. 

Next, let $D : \{0,1\}^L \rightarrow \{0,1\}$ be an efficient algorithm that can distinguish between the outputs of $G^f$ and uniformly random strings in $\{0,1\}^L$, i.e, $\big\vert \Pr_{z \sim \Bits^m}[D(G_m(z)) = 1] - \Pr_{y \sim \Bits^L}[D(y) = 1] \big\vert \geq 1/10$. \cite{NW94,IW01} show that for any $f$, given a circuit $D$ that distinguishes between the outputs of $G^f$ and uniformly random strings in $\{0,1\}^L$, there exists a uniform \textit{reconstruction algorithm $\sA$} that makes membership queries to $f$ and outputs a circuit that computes $f$ on a $(1/2+1/L)$-fraction of the inputs from $\Bits^n$. Alternatively, $\sA$ can be viewed as \textit{weak learner} for $f$, given membership \textit{query access} to it and $D$'s description.

Their algorithm is obtained by a ``play-to-lose" argument. Let $f$ be the underlying $n$-variate Boolean function, such that $\opt(f,\cAC^0[2]) = \beta^*(n)$, for a non-negligible $\beta^*$. \cite{CIKK17_agnostic} show an efficient algorithm $\sR$ that takes as input truth tables of length $L$, and rejects every function $g$ over $\ell$ inputs that has distance at most $1/\ell^3$ from $\cAC^0[2]$-circuits of size $2^{\ell^\gamma}$ (for a fixed $\gamma$), while accepting a random truth table with constant probability. This is nothing but a $(1/\ell^3)$-tolerant natural property.

The main idea is that with high probability over its seeds $z$ and stretch $L$ set to $2^{\polylog(n)}$, $G^f$ produces the truth table of a function $g_z : \{0,1\}^{\ell} \rightarrow \{0,1\}$, where $\ell = \log(L) = \polylog(n)$, such that the distance between $g_z$ and $\cAC^0[2]$-circuits of size $2^{\ell^\gamma}$ is at most $4\beta^*(n)$. As long as $\beta^*(n)$ is less than $1/4\ell^3$, the tolerant natural property acts as a distinguisher to $G^f$; indeed, by definition, it rejects almost all every truth table that could be output by $G^f$, whereas it accepts a constant fraction of strings from $\{0,1\}^L$. Thus, we can directly use the natural property as a distinguisher in the NW-reconstruction algorithm to get a circuit that computes $f$ on a $(1/2 + 1/L)$-fraction of the input in $\Bits^n$. This gives a weak agnostic learner for $\cAC^0[2]$ over the uniform distribution using membership queries.

\paragraph*{Query-to-sample reductions for NW-reconstruction algorithms.} Our main technical idea for this section is to show that the verifier can run the NW-reconstruction algorithm $\sA$ using random examples over the uniform distribution, by interacting with a prover to answer the queries. For this, we extend the query-to-sample reduction to the NW-reconstruction algorithm.

To do this, we need to understand the query pattern of $\sA$. Define a random \textit{subcube} $\sE$ as one obtained by setting a fixed set $S$ of coordinates in $[n]$ with a uniformly random string $\rho \in \Bits^{\vert S \vert}$, and considering the set of all Boolean strings in $\Bits^n$ that are consistent with $\rho$. Further, we define a \textit{subcube membership query} over $\sE$ as one which queries $f$ over all the $2^{n-\vert S \vert}$ many strings in $\sE$.

Our first insight is that $\sA$ only makes a set of random (possibly non-disjoint) subcube queries, where each subcube is of size at most $2^{\ell}$ (apart from a few uniformly random queries for making empirical estimates). The reconstruction algorithm is based on the hybrid argument, which informally implies that, if $D$ can distinguish between $f(z \vert_{S_1}), \dots, f(z \vert_{S_L})$ and a uniformly random string $(y_1, \dots, y_L)$, then for a random $i \in [L]$, $D$ also distinguishes between $(f(z \vert_{S_1}), \dots, f(z \vert_{S_{i-1}}), y_i, \dots y_L)$ and $(f(z \vert_{S_1}), \dots, f(z \vert_{S_{i-1}}), f(z \vert_{S_i}),y_{i+1} \dots y_L)$ reasonably well. As such, the subcubes $\sE_1, \dots, \sE_{i-1}$ are generated by projecting the random seed $z$ over the sets $S_1 \setminus S_i, \dots, S_{i-1} \setminus S_i$. Of additional interest to us, is the fact that these queries are non-adaptively generated.

To hide a random example $w$ in the query set, our next idea is that the query construction based on the hybrid argument allows us to embed $w$ in the seed $z$. Indeed, we pick a subcube $\sE_j$ at random and embed $w$ in the seed $z$ at a location that corresponds to $\sE_j$ (i.e., $z$ projected onto $S_j \setminus S_i$). Following this, we prove the NW-marginal query distribution lemma, in which we show that for any query index $q$ that falls in a subcube $\sE_j$, the probability that the $q^{\text{th}}$-query is a fixed string is equal to $1/2^{n-\log(\vert \sE_j \vert)}$. This is the most technical lemma for this section, and the main challenge here is to account for the effect of all the overlapping sets coming from the set design on the $q^{\text{th}}$ marginal distribution. We refer to \cref{lem:nw_query_marginal} for more details. 

The query-to-sample reduction follows a similar strategy as described earlier. While completeness holds for the same reasons, for soundness, for any query $q$ that lies in a subcube $\sE_j$, no malicious prover can distinguish between the distribution of $w \vert_{S_j \setminus S_i}$, where $w$ is the uniformly random example embedded in the query set, and the $q^{\text{th}}$ query marginal distribution, since they are identical. The rest of the analysis follows from ideas similar to \cref{sec:tech-overview-igl-juntas}.

In order to boost the error of the hypothesis to $\polylog(n) \cdot \beta^*$, the \cite{CIKK17_agnostic} learns the amplified function $\mathsf{Amp}^f_k : \Bits^{nk+k} \rightarrow \{0,1\}$, defined as $\mathsf{Amp}^f_k(x_1,\dots,x_k, b_1, \dots, b_k) = \sum_{j=1}^k f(x_j) b_j \pmod 2$, where $k$ is set to $\frac{1}{4\ell^3\beta^*}$. While dealing with the final CIKK-learner, it is worth highlighting a subtlety which becomes a \textit{critical} issue for the PAC-verification model. In order to get a hypothesis with low error, the learner learns $\mathsf{Amp}^f_k$, for which it needs to know the value of $\beta^*$ to set $k =\frac{1}{4\ell^3\beta^*}$. However, the main challenge for PAC-verification is that the verifier does not know $\beta^*$, otherwise the model becomes trivial. We overcome this, by running the protocol over multiple guesses of the unknown $\beta^*$, and use a \textit{finer analysis} of the CIKK-learner to find the right hypothesis. Details on how we extend the ideas highlighted above to the NW-reconstruction algorithm for $G^{\mathsf{Amp}^f_k}$, and the finer analysis of \cite{CIKK17_agnostic} can be found in \cref{sec:proof_aip_ac0}.

\subsection{Related Work}
Following \cite{GoldwasserRSY21}, \cite{MS23} extend PAC-verification to statistical query learning algorithms, and \cite{CHINS23} extend it for classical verification of quantum agnostic learners. A related model is that of covert learning by \cite{CK21}, where the goal is to ensure that no untrusted intermediary who monitors the interaction, gains any information about the function or the learner.

Of particular interest, \cite{MS23} show a verifier sample complexity lower bound of $\Omega(\sqrt{d})$, where $d$ is the \textsf{VC}-dimension of the underlying hypothesis class, in the \textit{distributional} agnostic learning setting, where the learner gets samples from an unknown joint distribution $\sD$ over $\{0,1\}^n \times \{0,1\}$. Our results on PAC-verifying $\Ppoly$ (Theorem \ref{thm:our-results-df-pac-verify-ppoly}) indicate that we can beat this verifier sample complexity lower bound in the setting of \textit{functional} agnostic learning with an unbounded prover, i.e., the labels come from an arbitrary Boolean function (over arbitrary example marginal distributions), or in the case of $k$-juntas (Theorem \ref{thm:our-results-juntas-informal}), by additionally also considering distributions with restricted marginals over the examples like the uniform distribution. 

The study of interactive proofs for testing properties of distributions, initiated by \cite{CG18} (see also \cite{HR22}), is also relevant here. In particular, \cite{GoldwasserRSY21} note that any PAC-verification task can be formulated as a property of distributions that can be tested using an interactive proof. While a direct application of \cite{CG18} shows that any property over distributions supported on $\Bits^n$ can be tested using $O(2^{n/2}/\eps^2)$ many samples given an \textit{unbounded prover}, \cref{thm:our-results-df-pac-verify-ppoly} gives a much stronger sample complexity of $O(1/\eps)$ for the specific case of PAC-verification in this scenario.

Another related model is that of delegating a property testing task to an untrusted prover (interactive proofs of proximity or IPPs), where the verifier has query access to the input \cite{RVW13}. In particular, \cite{GR22_sample} show query-to-sample reductions for IPPs. However, it is unclear whether such techniques for IPPs can be extended to PAC-verification. Moreover, extending our framework to IPPs incurs a quadratic blow-up for the sample complexity, which could render interesting query complexity regimes for IPPs trivial. 

\cite{GK23} show that easiness of certain \textit{meta-complexity problems} implies agnostic learnability of $\Ppoly$ in polynomial time using random examples over polynomially-samplable distributions. It remains open whether this can be extended for typical circuit class restrictions.

\subsection{Future Directions}
\label{sec:open-questions}
Our work establishes PAC-verification protocols for learning some fundamental classes in learning theory like heavy Fourier coefficients, constant-depth circuits, or $k$-juntas, with techniques that could have more general applicability. We state some directions for possible future exploration.
\begin{itemize}
    \item We show a general procedure of embedding a random example into specific query construction algorithms and further, prove that this gives a PAC-verification protocol, for query distributions where each marginal is identical to the underlying distribution (\cref{sec:tech-overview-igl-juntas}), or for a set of uniformly random subcube queries (\cref{sec:tech-overview-ac0[2]}). An interesting direction is to explore the full generality of this framework and obtain a full distributional characterisation of the query-to-sample reduction. Along these lines, it is worth studying whether the property of embeddability can be extended to adaptive query construction algorithms.
    
    \item Another question is whether we can use the algebraic structure of $\cAC^0[2]$ arising from \cite{Raz87,Smo87} to get better PAC-verification protocols that output a hypothesis with just an additive error, i.e., an error of $\opt(f,\cAC^0[2]) + \eps$.

    \item The role of round complexity in PAC-verification is not well-understood. Is there a hypothesis class that requires more than $2$ rounds of interaction? Is there a round-hierarchy for hypothesis classes that have doubly efficient PAC-verification protocols?
    
   \end{itemize}

\subsection{Acknowledgements}
\label{sec:ack}
We thank the anonymous STOC reviewers for their comments and suggestions. NR is grateful to Igor Carboni Oliveira, Marco Carmosino and Shuichi Hirahara for helpful discussions.


\section{Preliminaries}\label{sec:prelims}

In this section we establish some notation and describe some basic background. A Boolean function is a function of the form $f:\Bits^n\to\Bits$. Depending on convenience, we will alternatively switch the domain between $\{0,1\}^n$ and $\Field_2^n$, where $\Field_2$ denotes the finite field of size $2$. Similarly, we also use $\{-1,1\}$ as the codomain when necessary. 

Throughout the paper the notation $\ip{\cdot,\cdot}$ has been used with different meanings; the reader is assured of understanding the meaning from the context. This notation is used for the outcome of interaction between the verifier and the prover. On the other hand, if the arguments of  $\ip{\cdot,\cdot}$ are vectors from $\Field_2^n$, it denotes their dot product. For a distribution $\mathcal D$ we write $x\sim\mathcal{D}$ to mean $x$ is drawn according to $\mathcal D$, and for a finite set $S$ we write $x\sim S$ to mean $x$ is drawn uniformly at random from $S$. We also denote the VC-dimension of a class $\H$ by $\VC(\H)$.

Finally, we use $\Chernoff(\eps,\delta) = \frac{3}{\eps^2}\log\frac{2}{\delta} = O(\log(1/\delta)/\eps^2)$, as the number of independent samples from a Bernoulli random variable that estimates its mean within $\pm\eps$, with probability at least $1-\delta$ by the Chernoff bound. This sample quantity occurs in numerous places in this paper.

\subsection{Agnostic Learners}
\label{sec:agnostic_learn_prelim}

Let $\sF = \{\sF_n\}_{n \in \mathbb{N}}$, where $\sF_n$ is the set of all Boolean functions on $n$ bits. We consider classes of Boolean functions defined as $\sC = \{\sC_n\}$, where each $\sC_n \subseteq \sF_n$. Next, let $\sD = \{\sD_n\}$, be a family of fixed distributions, where each $\sD_n$ is a distribution over $\{0,1\}^n$. In particular, we use $\{\sU_n\}$, where $\sU_n$ is a uniform distribution over $\{0,1\}^n$.

For any $n$-variate Boolean functions $f$ and $g$, we define $\dist(f,g)$ as the relative Hamming distance between $f$ and $g$, or in other words, the distance between $f$ and $g$ over uniform distribution over $\{0,1\}^n$. Analogously, we also define distance $\dist_\sD(f,g)$ over any arbitrary distribution $\sD_n$ over $\{0,1\}^n$ as $\dist_\sD(f,g) = \Pr_{x \sim \sD}[f(x) \neq g(x)]$.

We next define the notion of agnostic learning. Let $f : \{0,1\}^n \rightarrow \{0,1\}$. We define two kinds of input access to $f$, the \textit{random example oracle} $\mathsf{Ex}(\sU_n,f)$ that provides independent and identical samples of the form $(x,f(x))$, where $x \sim \sU_n$, and the \textit{membership query oracle} that can takes an input $x \in \{0,1\}^n$ and outputs $f(x)$. 

For any concept class $\sC$, we define the best possible approximation of $f$ using $\sC$, through the expression $\opt(f,\sC) = \displaystyle \min_{c \in \sC_n} \dist(f,c)$. We analogously define the optimal distance to $\sC$ over any arbitrary distribution $\sD_n$ as $\mathsf{opt}_\sD (f,\sC) = \displaystyle \min_{c \in \sC_n} \dist_\sD(f,c)$. Then, we have

\begin{definition}[Agnostic Learning a class $\sC$ with hypothesis class $\sH$]
    \label{def:agn_learn}
    Let $\sC = \{\sC_n\}$ be a concept class and $\sH = \{\sH_n\}$ a hypothesis class. For $\varepsilon, \delta > 0$, we say that $\sC$ is $(\alpha,\varepsilon, \delta)$-agnostically learnable with $\sH$ using random examples in time $T \coloneqq T(n, 1/\varepsilon, 1/\delta)$, if there exists a learner $\sA$ that for every $n \geq 1$, for any $n$-variate Boolean function $f : \{0,1\}^{n} \rightarrow \{0,1\}$, given access to $\mathsf{Ex}(\sU_n,f)$, runs in time at most $T$ and outputs $h \in \sH_n$ such that with probability at least $1-\delta$,
    \begin{equation*}
        \dist(f,h) \leq \alpha \cdot \mathsf{opt}(f,\sC) + \varepsilon.
    \end{equation*}
    We say the learning algorithm is efficient if the running time is $\poly(n, 1/\varepsilon, 1/\delta)$.

    Similarly, we say that $\sC$ is $(\alpha,\varepsilon, \delta)$-agnostically learnable with $\sH$ using membership queries, if for any $n$-variate Boolean function $f$, the learning algorithm satisfies the guarantees of $(\alpha,\eps,\delta)$-agnostic learnability, given query access to $f$.
\end{definition}

Agnostic learning goes beyond standard PAC learning (also called \textit{realisable} learning) as it doesn't assume that $f$ always comes from $\sC$; instead, it represents a more realistic scenario where there is no ``ground truth" about the labeling function.\footnote{Indeed, agnostic learning when $\opt(f,\sC) = 0$, is the same as realisable learning.} Moreover, when $\sC = \sH$, the learner is said to be \textit{proper}, whereas if the hypothesis is not restricted to be from within $\sC$, then the learner is \textit{improper}. 

We also define \textit{distribution-free} $(\alpha,\epsilon,\delta)$-agnostic learning using membership queries (or just random examples), where for any unknown, yet fixed distribution $\sD$ over the examples, the learner is given access to the random example oracle $\mathsf{Ex}(\sD,f)$ with respect to $\sD$, as well as oracle access to $f$. The task is of the learner is to output a hypothesis with error at most $\alpha \cdot \opt_{\sD}(f,\sC) + \eps$, for every unknown underlying distribution $\sD$.

Note that we focus on the \textit{functional} version of agnostic learning, where the examples are labeled according to an arbitrary function $f : \{0,1\}^n \rightarrow \{0,1\}$, and not \textit{distributional} agnostic learning, where the underlying distribution is over both examples and labels, i.e., $\sD$ is a joint distribution over $\{0,1\}^n \times \{0,1\}$. In the functional setting, for any sample $x \in \{0,1\}^n$, the associated label is fixed as $f(x)$, whereas in the distributional setting, different samples of $x$ may provide different labels.

\subsection{Interactive Proofs for Agnostic Learning}
\label{sec:agnostic_ip_model_formal}
Below we define interactive proofs in the standard way (cf. \cite{AB09}). We use (a special case of) the definition by \cite{GoldwasserRSY21} for PAC-verification. 

\begin{definition}[$(\alpha,\varepsilon,\delta)$-PAC-verifying $\sC$]
\label{def:PAC-ver-formal}
For parameters $\delta, \varepsilon > 0$ and $\alpha \ge 1$, a concept class $\sC$ is $(\alpha, \varepsilon, \delta)$-\textbf{PAC Verifiable} with a hypothesis class $\sH$ over distribution $\sD$, if there exists an interactive protocol between a verifier $V$ and a prover $P$, taking explicit inputs $1^n, \delta, \varepsilon$, and an oracle access to $f$ (either through membership queries to $f$ or random examples $(x,f(x))$ with $x \sim \sD$). In the end of the protocol the verifier outputs either a hypothesis $h \in \sH$ or the value `reject' ($\bot$) and satisfies the following properties:
\begin{enumerate}

\item \textbf{Completeness:}
For every $f : \{0,1\}^n \rightarrow \{0,1\}$, there exists an \emph{honest} prover $P^*$ such that the output of the interaction satisfies
\begin{equation*}
    \Pr[\dist(f,h) \leq \alpha \cdot \dist(f,\sC) + \varepsilon] \geq 1-\delta.
\end{equation*}

\item \textbf{Soundness:}
For every $f : \{0,1\}^n \rightarrow \{0,1\}$ and every (computationally unbounded) prover $P$ the output of the interaction satisfies
\begin{equation*}
    \Pr[(h \neq \bot) \land (\dist(f,h) > \alpha \cdot \dist(f,\sC) + \varepsilon)] \leq \delta.
\end{equation*}
\end{enumerate}
Here, the probability is over the randomness of the samples, as well as the internal randomness of $P$ and $V$.
\medskip

\noindent Furthermore, we say that $P$ and $V$ have the parameters
\begin{equation*}
    \InteractiveParams{\random/\membership}{q}{c}{r}{T_v}{T_p}
    \enspace,
\end{equation*}
if given an oracle access to the input $f$, the verifier makes at most $q$ queries to $f$, the communication complexity of the protocol is $c$, the number of messages is $r$, the running time of the verifier is $T_v$, and the running time of the prover is $T_p$. We may represent the prover's running time as \textsf{learn}, emphasizing that the honest prover merely executes a learning algorithm and sends its result.

Note that, all these complexity parameters are functions of $n, \alpha, 1/\eps$ and $1/\delta$. If the verifier uses membership queries to the input, then the input access will be ``$\membership$", and if it only accesses the random example oracle, the input access will be denoted as ``$\random$". Typically, the prover will use membership queries to access the input $f$.
\end{definition}

\noindent Some remarks are in order here. 
\begin{enumerate}
    \item At times when $\dist(f,\sC) > 0$ is non-negligible, it will be convenient to drop the additive accuracy parameter $\varepsilon > 0$ from the definition to simplify notation and only use the term $(\alpha,\delta)$-PAC-verification. Similarly, if $\alpha = 1$, we just refer to $(\eps,\delta)$-PAC-verifiability. Finally, for the sake of readability, we drop constants and asymptotics from the parameter representation of the interactive proof. 
    
    \item This model trivialises for realisable learning, as the verifier can estimate the error of the purported hypothesis sent by a prover with respect to $f$ using a small enough number of random labeled examples, and accept if this estimate is close to $\eps$ (cf. \cite{GoldwasserRSY21}). On the other hand, estimating the closeness of the purported hypothesis with respect to the best approximation of $f$ with respect to reasonably powerful $\sC$ is quite challenging, since it is unclear how to estimate $\mathsf{opt}_\sC(f)$ using a prover with bounded computational power.
    
    \item In the case where both the verifier and the honest prover have only query access to $f$ (or only have random example access), we stress that the interactive protocol is only interesting if there is quantitative improvement in the verifier query (or sample) complexity. Another appealing setting is one where the prover has query access and the verifier has random example access to $f$, in which case the improvement is qualitative (and possibly quantitative as well).
\end{enumerate}

This model can  be extended in a natural way for \textit{distribution-free} agnostic learning. Another model we consider (in \cref{sec:aip_erm_unbounded_prover}) is PAC-verification where the \textit{honest prover is computationally unbounded}. Of course, in such a case, the prover can actually find the best possible $\sC$-circuit $C_n$ that approximates $f$, i.e., $\opt(f,\sC) = \dist(f,C_n)$.

\subsection{Fourier Analysis of Boolean Functions}\label{sec:boolean-functions}

We recall a few facts about Boolean functions.
Any function $f:\Bits^n\to\mathbb{R}$ can be uniquely written as
\[
f(x)=\sum_{\gamma\in\Bits^n}\hat{f}(\gamma)\chi_\gamma(x),
\]
which is denoted as the Fourier expansion of $f$ where $\chi_\gamma = (-1)^{\sum_{i\in\gamma}x_i}$ and is called a \emph{character}. It is easy to see that $\chi_\gamma(x+y)=\chi_\gamma(x)\chi_\gamma(y)$ when we view the domain as $\F_2^n$. The real number $\hat{f}(\gamma)$ is called the \emph{Fourier coefficient} of $f$ on $\chi_\gamma$ and the set of all coefficients is called the \emph{Fourier spectrum} of $f$.

The set of all functions $f:\Bits\to\mathbb{R}$ forms a vector space over $\mathbb{R}$, and by defining the inner product $\ip{f, g} = \E_{x\sim\Bits^n}[f(x)g(x)]$ the set of all characters $\{\chi_\gamma : \gamma \in\Bits^n\}$ forms an orthonormal basis for this linear space, since $\ip{f,\chi_\gamma} = \hat{f}(\gamma)$. The following well-known fact is used in \cref{sec:learning-fourier-coeffs}.

\begin{fact}[Parseval's identity]\label{fact:parseval}
For any Boolean function $f:\Bits^n\to\mathbb{R}$ we have $\ip{f,f}=\sum_{\gamma\in\Bits^n}\hat{f}(\gamma)^2$. In particular, if $f$ is a function with codomain $\{-1,1\}$, then $\ip{f,f} = 1$.
\end{fact}

 In \cref{sec:learning-fourier-coeffs} we also use $\hat{f}^{>\eps}$ to denote the set of characters $\{\gamma\in\Bits^n: |\hat{f}(\gamma)|>\eps\}$.

\subsection{Distance Estimators}
Agnostic PAC-verification in the realizable setting is much easier than in the agnostic setting. Indeed, it is the arbitrary distance of the input function to the hypothesis class that poses a hurdle in agnostic learning. In the agnostic verification setting, even if a prover provides the closest hypothesis, the task of verifying is non-trivial as the verifier does not know the true distance of the input function from the class. Thus, a natural algorithmic device to provide this utility is the notion of distance estimator which we will define. It must be mentioned that \cite{parnas2006tolerant} introduced distance approximation, a notion almost identical to ours. However, for our purposes, we redefine it slightly differently.

\begin{definition}
[$(\eps, \delta)$-distance estimator]
For parameters $\eps,\delta > 0$
and a class $C$ of functions, an \emph{$(\eps,\delta)$-distance estimator for $C$}
is an algorithm which is given oracle access to a function $f \colon \Bits^n \to \Bits$
and outputs $\tilde{d} \in [0,1]$ satisfying the guarantee
\begin{equation*}
    \Pr \left[\abs{\dist(f,C) -\tilde{d}} > \eps \right] < \delta
    \enspace.
\end{equation*}
The \emph{query complexity} of the estimator is the maximal number of queries made by the algorithm to $f$.
\end{definition}

A natural way of obtaining distance estimators is via tolerant testers. As we will show, noise tolerant testers and distance estimators are similar, and in some sense equivalent objects as the existence of one implies the existence of the other.

\begin{definition}[$(c_u, c_\ell)$-noise tolerant tester] \label{def:noise-tolerant-test}
For a class of Boolean functions $C$ and parameters $1 > c_u > c_\ell \geq 0$, a \emph{$(c_u, c_\ell)$-noise tolerant tester for $C$}
is an algorithm which given oracle access to a function $f \colon \Bits^n \to \Bits$.

\begin{enumerate}
\item If $\min_{g \in C} \dist(f, g) \leq c_\ell$, then the tester \emph{accepts} with probability at least $\frac{2}{3}$.
\item If $\min_{g \in C} \dist(f, g) \geq c_u$, then the tester \emph{rejects} with probability at least $\frac{2}{3}$.
\end{enumerate}
The \emph{query complexity} of the tester is the maximal number of queries made by the algorithm to $f$.
\end{definition}

Similar to \cite{parnas2006tolerant}, by a simple binary search argument we show that a noise tolerant tester for a class of functions implies a distance estimator.

\begin{claim}\label{clm:distance-estimator-from-tolerant-test}
Let $C$ be a class of functions, and suppose that for every $\eps>0$ and every $d \in [0,1]$
the class $C$ has a $(d+\eps,d-\eps)$-noise tolerant tester with query complexity $q_{\eps}$.
Then $C$ has an $(\eps, \delta)$-distance estimator
with query complexity $O(q_{\frac{\eps}{2}}\cdot\log(\frac{1}{\eps})\cdot (\log(\frac{1}{\delta}) + \log\log(\frac{1}{\eps})))$.

\end{claim}

\begin{proof}
    Given $\eps,\delta>0$ and query access to a function $f$, our goal is to find an estimation $\estimate{d}$ of $d^\star = \dist(f, C)$ such that $\Pr[|d^\star-\estimate{d}| \le \eps]\ge 1-\delta$. We run a binary search on the interval $[0,1]$, each time deleting half of the search space until the length of the remaining interval becomes sufficiently small.
    
    Let $d_m$ denote the middle of our active interval $[\ell, r]$ where at the beginning this interval is $[0,1]$. We run the $(d_m+\frac{\eps}{2}, d_m-\frac{\eps}{2})$-noise tolerant test for $O(\log\log(1/\eps)+\log(1/\delta))$ times and take the majority output. If the result is accept we keep $[\ell, d_m]$ as our active interval, and otherwise we keep $[d_m, r]$. We stop when the length of the interval becomes $\le \eps$ in which case we output the midpoint $\estimate{d}=d_m$ as our estimation of $d^\star$.
    In the case that $d^\star$ is in our final interval, clearly our answer is $\eps$-close to it. In the case that $d^\star$ is not in our last interval, with high probability it must be $\frac{\eps}{2}$-close to the interval, and therefore $\eps$-close to it's midpoint.
    
    To see why consider a noise tolerant test which always outputs the correct answer. Suppose $d^\star \in [\ell, r]$. Then, if $d^\star \notin [d_m-\frac{\eps}{2}, d_m+\frac{\eps}{2}]$, we will still have $d^\star$ in our active interval of the next step. If $d^\star \in [d_m-\frac{\eps}{2}, d_m+\frac{\eps}{2}]$, we may lose $d^\star$ from our active interval due to lack of guarantee from the test; however it will be $\frac{\eps}{2}$-close to the active interval. And it will remain $\frac{\eps}{2}$-close to it in the future steps as our tolerant test will always give us the correct half now that $d^\star$ is not in the interval.

    We can easily extend the above argument to the case that the noise tolerant is probabilistic (\cref{def:noise-tolerant-test}). Clearly, our binary search only needs $\log(1/\eps)$ steps. At each step by running our test for $O(\log\log(1/\eps) + \log(1/\delta))$ steps and taking the majority, we make sure that by Chernoff bound the output is correct with probability at least $1-\delta/\log(1/\eps)$. Finally, by union bound over all $\log(1/\eps)$ steps, all the executions of the noise tolerant test are correct with probability at least $1-\delta$.
\end{proof}

\section{General Transformations for PAC-Verification}
\label{sec:general-techniques}
In this section, we present two transformations that can be applied to get PAC-verifiers for various classes in a general fashion. First, we show an interactive proof that uses a tolerant tester for a class $\sC$ (possibly using membership queries) to verify agnostic learners for $\sC$ in \cref{sec:tolerant-to-aip}.  Following this we show our next transformation, where any verifier in a PAC-verification protocol that uses membership queries and satisfies certain conditions, allows us to construct an interactive proof where the verifier only uses \emph{random} examples, by delegating the membership queries to an (untrusted) prover in \cref{sec:query-to-sample-reduction}.  

\subsection{From Tolerant Testing to PAC-Verification}
\label{sec:tolerant-to-aip}

Below we show a transformation that uses a distance estimator for a class $\H$ with query complexity $q$, and uses it to design an interactive protocol for learning $\H$, where the verifier makes $q$ membership queries induced by running the given distance estimator.

\begin{theorem}\label{thm:distance-estimator-implies-IP}
    Let $\H$ be a class of Boolean functions, and suppose that $\H$ has an $(\eps,\delta)$-distance estimator $D$ with query complexity $q = q(\eps,\delta)$. Then $\H$ is $(\eps,\delta)$-PAC-verifiable with verifier $V$ and prover $P$ with the following interactive parameters

    \begin{equation*}
    \InteractiveParams
    {\membership}
    {q(\frac{\eps}{6},\frac{\delta}{2})}
    {\log(|\H|)}
    {\proofonly}
    {T_D(\frac{\eps}{6},\frac{\delta}{2}) + T_{\text{dist}}(\frac{\eps}{6},\frac{\delta}{2})}
    {\textsf{learn}}
    \enspace,
    \end{equation*}
    where $T_D$ is the time it takes to run the distance estimator $D$ with parameters $(\frac{\eps}{6},\frac{\delta}{2})$, and $T_\text{dist}$ is the running time it takes to estimate the distance between the given input and an hypothesis. The interaction $\ip{P, V}$ outputs $h\in\H\cup\{\bot\}$ (reject) such that the following is satisfied.
    \begin{itemize}
        \item \textbf{Completeness:} There exists an honest prover $P^*$ such that the output of interaction $\ip{P^*,V}=h$ satisfies $\Pr[\dist(f,h) \le \dist(f,\H) + \eps/3] \ge 1-\delta$.
        \item \textbf{Soundness:}  For any (possibly unbounded) prover $P$, the output of interaction $\ip{P, V}=h$ satisfies $\Pr[(h\neq\bot)\land(\dist(h, f) > \dist(\H, f) + \eps)] < \delta.$
    \end{itemize}
\end{theorem}

\begin{proof}
    We start by describing the protocol.

\begin{transform}[H]
\caption{An $(\eps, \delta)$-PAC-Verification for $\H$ via a distance estimator $D$}\label{alg:tolerant-test-to-IP-for-learning}

    \SetKwInOut{FunIn}{Function input}
    \SetKwInOut{ExpIn}{Explicit inputs}
    \FunIn{A function $f:\Bits^n\to\Bits$, where the verifier has membership access}
    \ExpIn{Parameters $\eps,\delta>0$}

    The prover sends a hypothesis $h \in \H$ to the verifier, claiming it is ${\frac{\eps}{3}}$-close to optimal.

    The verifier takes $\Chernoff(\frac{\eps}{6},\frac{\delta}{2}) = O( \log(1/\delta)/\eps^2 )$ random labeled samples and estimates $\dist(h, f)$ up to error $\frac{\eps}{6}$ with probability at least $1 - \frac{\delta}{2}$. Let $\estimate{d}_h$ denote this estimation.
    
    The verifier runs $D$ using $q(\frac{\eps}{6}, \frac{\delta}{2})$ membership queries and estimates $\dist(f,\H)$ up error $\frac{\eps}{6}$ with probability at least  $1 - \frac{\delta}{2}$. Let $\estimate{d}_\H$ denote this estimation.

    If $\abs{\estimate{d}_h - \estimate{d}_\H} \leq \frac{2\eps}{3}$, the verifier outputs $h$. Otherwise, the verifier rejects.
\end{transform}

Let $d_h = \dist(f, h)$ and $d_\H = \dist(f, \H)$. By Chernoff bound, the verifier's estimation of $d_h$, denoted by $\estimate{d}_h$ is $\frac{\eps}{6}$-close to $d_h$ with probability at least $1- \frac{\delta}{2}$. Also, the verifier's estimation $\estimate{d}_\H$ is $\frac{\eps}{6}$-close to the actual distance $d_\H$ with probability at least $1-\frac{\delta}{2}$.
Therefore, by union bound, with probability at least $1-\delta$, both $\estimate{d}_h$ and $\estimate{d}_\H$ are within $\pm\frac{\eps}{6}$ their target value. Therefore, with probability at least $1-\delta$ we have

\begin{equation}\label{eq:estimation-is-close}
\abs{
\abs{\estimate{d}_h-\estimate{d}_\H} - \abs{d_h-d_\H} } \leq 
\abs{\estimate{d}_h-d_h} + \abs{\estimate{d}_\H-d_\H} \le \eps/3.
\end{equation}

\paragraph{Completeness:}
Suppose $\abs{d_h - d_\H} \le \eps/3$. Then, by \cref{eq:estimation-is-close}, with probability at least $1-\delta$ the estimated $\abs{\estimate{d}_h-\estimate{d}_\H} \le 2\eps/3$ and hence the verifier outputs $h$. 

\paragraph{Soundness:}
Suppose the hypothesis $h$ sent by the prover is incorrect, i.e. $\abs{d_h - d_\H} > \eps$. Then, from \cref{eq:estimation-is-close}, with probability at least $1-\delta$ the estimated $\abs{\estimate{d}_h-\estimate{d}_\H} > 2\eps/3$ and the verifier rejects.

\paragraph{Parameters:}
For a natural class of functions $\H$ the prover needs to only communicate $\log(|\H|)$ bits to send a hypothesis. All remaining parameters can be easily verified.
\end{proof}

\begin{remark}
It is worth noting that the above protocol can in fact, directly use a tolerant test $T$ for $\mathcal{H}$, obviating the need of constructing a distance estimator for $\mathcal{H}$ from $T$ (from Claim \ref{clm:distance-estimator-from-tolerant-test}), and thus saving a factor of $\log(1/\eps)$ induced by binary search. Indeed, the verifier runs $T$ with parameters $(\estimate{d}_h -\eps/2, \estimate{d}_h - 5\eps/6)$. If $T$ rejects the input, the verifier outputs $h$, and otherwise, the verifier rejects. For completeness, with high probability $\estimate{d}_h - \eps/2 \le d_\H$, and $T$ rejects the input. In the soundness case, with high probability $d_\H < \estimate{d}_h - 5\eps/6$ and the tolerant test accepts. Note that we still need to boost the confidence of the tolerant test by repeating it.
\end{remark}


\subsection{Transformation From Membership Queries to Random Examples}
\label{sec:query-to-sample-reduction}

In this section we show under certain conditions one can transform an interactive protocol for learning that uses membership queries to one that only uses random examples.

If each query has the same marginal distribution as the random examples of the verifier, we can reduce the PAC-verification protocol where the verifier uses membership queries, into one where the verifier takes random examples and prover has membership query access. The idea is as follows: the verifier selects one of the queries uniformly at random and fixes its value to a labeled random example. Note that the answer to this query is already known to her. Then she has to sample the rest of queries conditioned on the one fixed query; assuming this conditional sampling can be done computationally efficiently the verifier then sends these queries to the prover and asks him to label them. The verifier can finally check the label of the one query that is known to her to verify that prover did not cheat on the queries. They can repeat this interaction several times to increase the soundness as desired.

\begin{theorem}
\label{thm:membership-to-random-transform}
Let $\H=\H_n$ be a hypothesis class, and let $\D$ be a distribution over inputs to functions $h \in \H$. Suppose that $\H$ is $(\eps,\delta)$-PAC-verifiable with respect to $\D$ with the following parameters
    \begin{equation*}
    \InteractiveParams
            {\membership}
            {q = q(\eps, \delta)}
            {C = C(\eps, \delta)}
            {\proofonly}
            {T_v}
            {T_p}
    \enspace.
    \end{equation*}
Furthermore, suppose that (1) the marginal distribution of each (membership) query made by the verifier is $\D$ (2) for any string $x$ it is possible to efficiently ``construct a membership query set\footnote{A formalization of this notion can be found in \cref{sec:app_embed}.}" conditioned on $x$ being a randomly chosen query index.

Then, $\H$ is $(\eps,2\delta)$-PAC-verifiable with respect to $\D$ such that the verifier only takes random examples and has the following parameters
    \begin{equation*}
    \InteractiveParams
            {\random}
            {q\cdot \log(1/\delta)}
            {C + n\cdot q^2\cdot \log(1/\delta)}
            {3}
            {(T_v + T_\mathsf{query})\cdot q\cdot \log(1/\delta)}
            {T_p + n\cdot q^2\cdot \log(1/\delta)}
    \enspace,
    \end{equation*}
    where $T_\mathsf{query}$ is the time it takes to construct a membership query set.
\end{theorem}

\begin{proof}
The transformed protocol is as follows: 

\begin{transform}[H]
    \SetAlgoLined
    \KwIn{A function $f:\Bits^n\to\Bits$, and $\eps,\delta > 0$ and $P_M$ -- the protocol with membership queries.}
    
    Prover sends a proof $\pi$ ($P_M$ is proof only).

    Verifier takes $t = 200\cdot q(\eps, \delta)\cdot \ln(1/\delta)$ random examples $\{(x_i,f(x_i))\}_{i=1}^t$. Then for each $x_i$ ``constructs a membership query set'' $Q_i\subseteq\Bits^n$ for $P_M$. That is, the verifier in $P_M$ makes these queries.
    
   Verifier sends all of the query sets $\{Q_i\}_{i=1}^t$ without the labels to the prover.

   Prover sends all of the purported answers $\tilde{f}(Q_i) = \{(x, \tilde{f}(x)) : x \in Q_i\}$ to the verifier, where $\tilde{f}$ denotes the labeling of the prover.

    \ForAll{$i \in [t]$ executions}{

        Verifer checks $\tilde{f}(x_i) = f(x_i)$. If not, she rejects the query set and outputs $\bot$ (reject).

        Verifier simulates $P_M$ with proof $\pi$ and queries $\tilde{f}(Q_i)$, and outputs $h_i\in \H\cup\{\bot\}$. Note that the proof $\pi$ remains fixed across all iterations.
    }

    The verifier rejects, if more than $t/2$ iterations reject. Otherwise, she outputs $h_i$, the output of an arbitrary non-rejecting iteration.
    
\caption{$P_R$ -- the equivalent protocol with random labeled queries only.}
\end{transform}

Before proceeding with the proof, we describe what we mean by constructing a membership query set around $x$. Let $Q$ of size $q$ denote the membership query set generated by $P_M$. Then, consider the following process.
\begin{enumerate}
\item Sample a random index $j\sim [q].$
\item Take a random example $(w, f(w))$.
\item Construct a valid query set $Q$ \emph{conditioned} on the $j^{\text{th}}$ query being $w$. Here, being valid means that $Q$ has the same distribution of queries made by the $P_M$ verifier.
\end{enumerate}
Our query-to-sample reduction assumes that the last step of the above construction is feasible. Since the marginal distribution of all the queries is the same and since $j$ was chosen uniformly at random, the prover gains no advantage in detecting $j$. 

\medskip 
In the rest of this proof, $V$ denotes the verifier and $P$ denotes the prover, and we assume $\delta < \frac{1}{3}$. For each $i\in[t]$, let $f(Q_i) = \{(x, f(x)) : x\in Q_i\}$ denote the correct labeling of the $i$-th query set. We define $t_\mathsf{bad} = \abs{\{i \in [t] : \tilde{f}(Q_i) \neq f(Q_i)\}}$ to be the number of dishonest query sets provided by the prover.

\paragraph{Completeness:}

While interacting with an honest prover we have $t_\mathsf{bad} = 0$. Note that even though an honest prover provides a valid proof $\pi$ and faithful query sets, the verifier may still reject some iterations, but with probability no more than $\delta$ individually. Thus the expected number of non-rejecting executions for any $\delta < 1/3$ is at least $(1-\delta)t > 2t/3$. By Chernoff bound, since $t > 36\log(2/\delta)$ then the number of non-rejecting executions is within $\pm t/6$ of its expectation with probability at least $1-\delta$. Thus with the same probability, the number of non-rejecting iterations is at least $(1-\delta)t - t/6 > t/2$. Therefore, the verifier does not reject, and outputs a hypothesis that is guaranteed to be $\eps$-close to optimal, with probability at least $1-2\delta$.

\paragraph{Soundness:}
Let $R_1$ denote the event that verifier rejects due to rejecting a query set, and let $R_2$ denote the event that verifier rejects due to rejecting more than half of the iterations. Note that $R_1$ and $R_2$ are mutually exclusive, and therefore $\Pr[\ip{P, V}=\bot] = \Pr[R_1] + \Pr[R_2]$. Assuming the prover sends an invalid proof $\pi$ we show at least one of $R_1$ or $R_2$ happens with high probability. Let $q = q(\eps, \delta)$. We have two cases.

\paragraph{Case $t_\mathsf{bad} \ge q\cdot\ln(\frac{1}{\delta})$:} Then

$$
\Pr[\neg R_1]
\le \left(1- \frac{1}{q} \right)^{t_\mathsf{bad}}
\le  e^{-t_\mathsf{bad}/q}
\le e^{-\frac{\ln(1/\delta)q}{q}}
\le \delta.
$$

Hence the verifier rejects with high probability, i.e. $\Pr[\ip{P, V}=\bot] \geq 1 - \delta$.

\paragraph{Case $t_\mathsf{bad} < q\cdot\ln(\frac{1}{\delta})$:} We show in this case with overwhelming probability at least half the executions will reject, thus $P_R$ rejects.

Let $t_\mathsf{good} = t - t_\mathsf{bad}$, and notice that $t_\mathsf{good} \ge 199q\ln(1/\delta)$. We show among the $t_\mathsf{good}$ faithful query sets more than $t/2$ of them reject with overwhelming probability. Note that with probability no more than $\delta$ a faithful query set $Q_i$ fails to catch a wrong proof, rendering it ineffective. Define the random variable $X$ to be the number of faithful query sets for which $P_M$ rejects.

Each faithful query set rejects with probability at least $1-\delta$. Let $\alpha = 1 - \frac{100}{199(1-\delta)}$, $\beta=1-\frac{3\cdot 100}{2\cdot 199}$, and $\mu = \E[X]$. Notice that $\mu \ge (1-\delta)199q\ln(1/\delta)$. By Chernoff bound

\begin{align*}
\Pr[\neg R_2] &\le \Pr[X \le t/2] \\
&\le \Pr[X \le (1-\alpha)\mu] \\
&\le \Pr[X \le (1-\beta)\mu] \\
&\stackrel{(*)}{\le} \exp\left(-\frac{\beta^2}{2}\mu\right) \\
&\stackrel{(**)}{\le} \exp\left(-\frac{\beta^2}{2}(1-\delta)199q\ln(1/\delta)\right) \\
&\stackrel{(***)}{\le} \exp\left(-\frac{\beta^2}{3}199\ln(1/\delta)\right) \\
&= \delta^{199\beta^2/3} \le \delta.
\end{align*}

The inequality $(*)$ follows from Chernoff bounds, and $(**)$ follows because $\mu$ is at least $(1-\delta)199q\ln(1/\delta)$. Finally $(*\!*\!*)$ holds because $(1-\delta)\cdot q > 2/3$. Therefore $\Pr[\ip{V, P}=\bot] \ge \Pr[R_2] \ge 1-\delta$.

\end{proof}

\section{Sample-Efficient Interactive Goldreich-Levin Protocol}
\label{sec:learning-fourier-coeffs}

We view any Boolean function $f$ in its equivalent form as $f:\Field_2^n\to\{-1,1\}$ instead. Recall from Section \ref{sec:boolean-functions}, that any Boolean function $f:\F_2^n\to\{-1,1\}$ has the following Fourier expansion $f(x) = \sum_{\gamma\in\F_2^n}\hat{f}(\gamma)\chi_\gamma(x)$, where $\hat{f}(\gamma)$ is the Fourier coefficient of $f$ associated with $\gamma$. Since Fourier characters and $\F_2^n$ have a one-to-one correspondence in this expansion, we refer to characters or vectors in $\F_2^n$ interchangeably for the rest of this section. In other words, the domain $\F_2^n$ also acts as a means to index the Fourier coefficients of $f$. Finally, any $\gamma\in\F_2^n$ such that the absolute value of its Fourier coefficient $|\hat{f}(\gamma)|$ is ``large'', is called a \emph{heavy vector}. The terms ``large'' or ``heavy'' will obtain a more precise meaning in our detailed description.

In \cref{sec:alg-for-learning-fourier-values} we provide an algorithm that outputs the value of the top $t$ Fourier coefficients of a given function. We remark that we find this result of independent interest. Then using the idea of this algorithm, in \cref{sec:pac-verification-top-fourier-coefficients} we design an interactive protocol for learning the top $t$ characters associated with the top $t$ Fourier coefficients. The output of this protocol is a ``top'' set in the sense of the following definition.

\begin{definition}[$\eps$-Top Set]
   A set $\Gamma\subseteq\Field_2^n$ is called \emph{$\eps$-top} with respect to a Boolean function $f:\F_2^n\to\{-1,1\}$ if for any $\alpha\in \Field_2^n\setminus\Gamma$ and any $\beta\in\Gamma$ we have $|\hat{f}(\alpha)| \le |\hat{f}(\beta)| + \eps$.
\end{definition}

Finally, we apply our general query-to-sample reduction to get the following result:

\begin{theorem}[Formal version of \cref{thm:our-results-topfourier}]\label{thm:fourier-random}
    There exists an interactive protocol with verifier $V$ and prover $P$ where, given $\eps, \delta > 0$ and a number $t\in\N$, and verifier has random example access to a function $f:\F_2^n\to\{-1,1\}$, the verifier uses $\poly(t,\frac{1}{\eps}, \log\frac{1}{\delta})$ random examples (independent of $n$) to $f$ and outputs $\Lambda$, which is either a size-$t$ ``top'' subset of $\F_2^n$ or $\bot$ (reject). This protocol has the following interactive parameters

    \begin{equation*}
         \InteractiveParams{\random}{\poly(t,1/\eps, \log(1/\delta))}
         {n\cdot \poly(t,1/\eps, \log(1/\delta))}
         {3}
         {\poly(n,t,1/\eps,\log(1/\delta))}
         {\poly(n,t,1/\eps,\log(1/\delta))}
        \enspace,
    \end{equation*}

    and satisfies the following completeness and soundness conditions.

    \begin{itemize}
        \item \textbf{Completeness:} There exists an honest prover $P^*$ such that the output of interaction $\ip{P^*,V}=\Lambda$ satisfies $\Pr\left[\Lambda \text{ is } \frac{\eps}{3}\text{-top}\right] \geq 1-\delta-o(1)$,
    
        \item \textbf{Soundness:} For any (possibly unbounded) prover $P$, the output of interaction $\ip{P,V}=\Lambda$ satisfies $\Pr\left[(\Lambda \neq \bot) \land (\Lambda \text{ is not } 
 \eps\text{-top})\right] \leq \delta + o(1).$
    \end{itemize}
\end{theorem}

\subsection{Algorithm for Computing the Top Fourier Coefficients}
\label{sec:alg-for-learning-fourier-values}

Fix a Boolean function $f:\F_2^n\to\{-1,1\}$. We follow a similar strategy as the Kushilevitz-Mansour \cite{KushilevitzM93} algorithm (based on Goldreich-Levin \cite{GoldreichL89}) for learning heavy Fourier coefficients of a Boolean function $f$, given membership query access to it. However, unlike the KM algorithm, where the bucketing algorithm iteratively splits with respect to a fixed enumeration over the input variables to $f$, we instead split over \textit{random linear functions}.\footnote{This splitting idea is also used by \cite{gopalan2011testing}, but for the case of testing sparse Boolean functions.} This helps us get a much stronger query complexity (independent of $n$.)

\paragraph{Algorithm Overview:} The goal is to find the $t$ largest Fourier coefficients of $f$, for a given $t\in\N$. Our key idea for this algorithm is that by partitioning $\Field_2^n$ into a few (but at least $t$) randomly constructed, low dimensional affine subspaces, we are likely to \emph{separate} all the heavy vectors (which are not too many by Parseval's theorem) into different subspaces. In other words, with high probability, each subspace in such a random partitioning contains at most one heavy vector. Therefore, a natural approach to accomplish our goal, would be to estimate the largest associated coefficient in from each subspace in the partition, and output the $t$ top ones. Moreover, the extra structure offered by affine nature of partition, enables us to develop tools from Fourier analysis for estimating the largest corresponding coefficient in it. It is worth stressing that while this algorithm only returns the Fourier coefficients, at the end of this, the vectors corresponding to these coefficients are still unknown.

\paragraph{Detailed Description:} Let $\Lambda_t = \{\gamma_1,\dots, \gamma_t\}$ be the set of vectors in $\F_2^n$ corresponding to the $t$ largest Fourier coefficients of $f$. We want to estimate the value of $|\hat{f}(\gamma_i)|$. To start with, we construct our random partition of $\F_2^n$ as follows. Pick $s$ random vectors $r_1,\dots,r_s$ from $\Field_2^n$, for $s$ to be defined later. These random vectors, assuming they are linearly independent, partition $\Field_2^n$ into $2^s$ affine subspaces (or \textit{cosets}), each of which is described by a vector in $\Field_2^s$. More precisely, we partition $\F_2^n$ into the following sets, each characterized by a vector $a\in\F_2^s$
\begin{align}\label{equ:cosets}
    V_a = \{\gamma\in\Field_2^n\ :\ \langle\gamma,r_i\rangle = a_i\ \text{for each }i\in[s]\}.
\end{align}

Amongst all these subsets of $\F_2^n$, an important observation here is that the subspace characterized by the all zeroes vector, $V_0$, is the subspace orthogonal to the span of the vectors $r_1, \dots, r_s$, while the rest are affine subspaces defined as its cosets. Indeed, for every $a\in\Field_2^s$, we see that each $V_a=V_0 + h_a$ where $h_a$ is any vector in $V_a$. Following this, we use $V_0+h$ to refer to an arbitrary coset of $V_0$.

Now suppose we had an algorithm, given a coset $V_0 + h$, could output $\max_{\gamma\in V_0 + h} |\hat{f}(\gamma)|$. Then, we attain our goal of computing the top coefficients easily: for suitably picked $s$, find $r_1, \dots, r_s$ such that each $V_a$ contains at most one of the vectors $\Lambda_t = \{\gamma_1,\dots, \gamma_t\}$ (i.e. vectors of $\Lambda_t$ are separated across the different cosets of $V_0$), and output the $t$ largest values from the coefficients estimated by this algorithm, over each coset of $V_0$. This motivates us to examine the probability that such a random partition separates $\Lambda_t$. Before that, it would be convenient for us to introduce the following notation.

\begin{definition}[$X$-Rare Set]
Let $X$ and $Y$ be any subsets of $\F_2^n$. Then $Y$ is called \emph{$X$-rare}  if $|X\cap Y| \le 1$. 

Moreover, for any Boolean function $f$, (abusing notation) we say that $Y$ is an $\eps$-rare set as a shorthand for a set that is $(\hat{f}^{>\eps})$-rare, i.e., $Y$ has at most one vector whose corresponding Fourier coefficient has absolute value at least $\eps$.
\end{definition}

\begin{claim}\label{claim:separate-whp}
    For any set $X\subseteq\Field_2^n$ and 
 any $0<\delta<1$, if we draw $s \ge 2\log_2(|X|) + \log_2(1/\delta)$ random vectors $r_1,\dots, r_s$ then all cosets of $V_0$ will be $X$-rare with probability at least $1-\delta$.
\end{claim}
\begin{proof}
For \emph{any} $\alpha, \beta\in\Field_2^n$ such that $\alpha\neq\beta$ the probability that they end up in the same coset of $V_0$ is at most $2^{-s}$. Thus, using the union bound, for any $X\subseteq\Field_2^n$ the probability that a coset of $V_0$ contains two vectors from $X$ is at most $\binom{|X|}{2} \cdot 2^{-s}$, and consequently at most $|X|^2\cdot 2^{-s}$. Hence choosing $s \ge 2\log_2(|X|) + \log_2(1/\delta)$ bounds this probability by $\delta$.
\end{proof}

\cref{claim:separate-whp} shows that we need not many random vectors to separate any set (and in particular $\Lambda_t$) with high probability. Therefore, it only remains to estimate the maximum corresponding coefficient in each coset of $V_0$ to finish our algorithm. The following lemma paves the path to achieving this by offering a way to estimate the sum of fourth power of the Fourier coefficients corresponding to the vectors in an affine subspace.

\begin{lemma}\label{lem:fourth-power-sum}
  Let $V + h$ be an affine subspace of $\F_2^n$. For any $f:\F_2^n\to\{-1, 1\}$ the following holds
    \begin{equation*}
    \sum_{s \in V+h} \fhat(s)^4 = 
    \E_{\substack{x, y, z \sim \F_2^n\\ w \sim V^\perp }}
            [\chi_h(w) \cdot f(x) \cdot f(y) \cdot f(z) \cdot
            f(x+y+z+w)].\\
    \end{equation*}
\end{lemma}
\begin{proof}
Let $A$ denote the right-hand side of the equation. We start by substituting each term for its Fourier expansion

\begin{align*}
  &A = \E_{\substack{x, y, z \sim \F_2^n \\ w \in V^\perp }}
            [\chi_{h}(w) \cdot f(x) \cdot f(y) \cdot f(z) \cdot
            f(x+y+z+w)]\\
    &= \E_{\substack{x, y, z \sim \F_2^n \\ w \sim V^\perp \\ r:= x+y+z+w}}
    \left[
    \chi_{h}(w)
    \left(\sum_{s \in \F_2^n}\hat{f}(s)\chi_s(x)\right)
    \left(\sum_{t \in \F_2^n}\hat{f}(t)\chi_t(y)\right)
    \left(\sum_{u \in \F_2^n}\hat{f}(u)\chi_u(z)\right)
    \left(\sum_{v \in \F_2^n}\hat{f}(v)\chi_v(r)\right)
    \right]\\
    &= \E_{\substack{x, y, z \sim \F_2^n \\ w \sim V^\perp }}
    \left[
    \chi_{h}(w)
    \sum_{s, t, u, v \in \F_2^n}
            \hat{f}(s)\hat{f}(t)\hat{f}(u)\hat{f}(v)
            \chi_s(x)\chi_t(y)\chi_u(z)\chi_v(x+y+z+w)
    \right].
\end{align*}

By linearity of expectation

\begin{align*}
    A &= \sum_{s, t, u, v \in \F_2^n}
            \hat{f}(s)\hat{f}(t)\hat{f}(u)\hat{f}(v)
            \E_{\substack{x, y, z \sim \F_2^n \\ w \sim V^\perp }}
              \left[
                \chi_{h}(w)\chi_s(x)\chi_t(y)\chi_u(z)\chi_v(x+y+z+w)
              \right]\\
    &= \sum_{s, t, u, v \in \F_2^n}
            \hat{f}(s)\hat{f}(t)\hat{f}(u)\hat{f}(v)
            \E_{\substack{x, y, z \sim \F_2^n \\ w \sim V^\perp }}
              \left[
                \chi_{h}(w)\chi_s(x)\chi_t(y)\chi_u(z)\chi_v(x)\chi_v(y)\chi_v(z)\chi_v(w)
              \right].
\end{align*}

Since the random variables $x, y, z, w$ are all independent

\begin{align*}
    A &= \sum_{s, t, u, v \in \F_2^n}
            \hat{f}(s)\hat{f}(t)\hat{f}(u)\hat{f}(v)
            \mathop{\E}_{w\sim V^\perp}[\chi_h(w)\chi_v(w)]
            \mathop{\E}_{x\sim\F_2^n}[\chi_s(x)\chi_v(x)]
            \mathop{\E}_{y\sim\F_2^n}[\chi_t(y)\chi_v(y)]
            \mathop{\E}_{z\sim\F_2^n}[\chi_u(z)\chi_v(z)].
\end{align*}

The expectation $\E_{x\sim\F_2^n}[\chi_s(x)\chi_v(x)]$ is $1$ only when $s = v$ and is $0$ otherwise. This is also true for the expectations over $y$ and $z$. Therefore the terms in the above summation in which $s = t = u = v$ does not hold cancel to $0$. So we simplify the summation to the following

\begin{align*}
    A &= \sum_{s \in \F_2^n}
            \hat{f}(s)^4
            \E_{w\sim V^\perp}[\chi_h(w)\chi_s(w)]
    = \sum_{s \in \F_2^n}
            \hat{f}(s)^4
            \E_{w\sim V^\perp}[\chi_{h + s}(w)]
    = \sum_{s \in V + h}
            \hat{f}(s)^4.
\end{align*}

The last equality uses the fact that $\E_{w\sim V^\perp}[\chi_{h + s}(w)] = \mathds{1}(s\in V+h)$. To see why, consider the following two cases:

\begin{enumerate}[label=(\roman*)]
  \item $s\in V + h$: In this case $\ip{h+s,w}=0$ for all $w$ in $V^\perp$, yielding $\E_{w\sim V^\perp}[\chi_{h + s}(w)] = \E_{w\sim V^\perp}[(-1)^{\ip{h+s, w}=0}] = 1$.

  \item $s\notin V + h$: Let $\beta = \{b_1, \dots, b_k\}$ be a basis for $V^\perp$. From $s + h \notin V$ it follows that the set $M_{s + h} = \{b\in\beta : \ip{b,h+s} = 1\}$ is non-empty. Consequently
    \begin{equation*}\begin{aligned}
      \Pr_{w\sim V^\perp}[\ip{w, s + h} = 1] = \Pr_{w\sim \mathsf{span}(M_{s + h})}[\ip{w, s + h} = 1] = \frac{1}{2}.
    \end{aligned}\end{equation*}
    From this we conclude $\E_{w\sim V^\perp}[\chi_{h + s}(w)] = 0.\qedhere$
\end{enumerate}
\end{proof}

Next, we show if the set $V_0 + h$ contains at most one heavy vector, then the fourth root of sum of fourth power is close to the maximum corresponding coefficient in $V_0 + h$. Recall that for two real numbers $r_1, r_2\in \mathbb{R}$ we say they are $\eps$-close if $|r_1-r_2|\le\eps$.

\begin{claim}\label{clm:fourth-sum-good-estimate-for-max}
Let $\Gamma\subseteq \Field_2^n$ be $\eps$-rare for some $0 < \eps < 1$. Then $\sqrt[4]{\sum_{\gamma\in \Gamma}\hat{f}(\gamma)^4}$ is $\sqrt{\eps}$-close to $|\hat{f}(\gamma^\star)|$ where $\gamma^\star = \argmax_{\gamma\in \Gamma}|\hat{f}(\gamma)|$.
\end{claim}
\begin{proof}
We have $\hat{f}(\gamma^\star)^4\le \sum_{\gamma\in \Gamma}\hat{f}(\gamma)^4$ as a trivial lower bound for the sum. On the other hand

\begin{align*}
\sum_{\gamma\in \Gamma}\hat{f}(\gamma)^4
= \hat{f}(\gamma^\star)^4 +  \sum_{\substack{\gamma\in \Gamma\\ \gamma\neq\gamma^\star}}\hat{f}(\gamma)^4
\stackrel{(*)}{\le} \hat{f}(\gamma^\star)^4 + \eps^2\sum_{\substack{\gamma\in \Gamma\\ \gamma\neq\gamma^\star}}\hat{f}(\gamma)^2
\stackrel{(**)}{\le} \hat{f}(\gamma^\star)^4 + \eps^2,\end{align*}

Here, in the inequality given by $(*)$, we use the fact that when $\gamma\ne\gamma^\star$ we have $\hat{f}(\gamma)\le\eps$, and in $(**)$, we use Parseval's identity (\cref{fact:parseval}). Putting the lower and the upper bounds together, combined with the fact that the function $\sqrt[4]{\cdot}$ is  monotonically increasing and sub-additive, we get

\begin{align*}
|\hat{f}(\gamma^\star)| \le \sqrt[4]{\sum\nolimits_{\gamma\in \Gamma}\hat{f}(\gamma)^4} \le \sqrt[4]{\hat{f}(\gamma^\star)^4 + \eps^2} \le \sqrt[4]{\hat{f}(\gamma^\star)^4} + \sqrt[4]{\eps^2} = |\hat{f}(\gamma^\star)| + \sqrt{\eps}.
\end{align*}

In other words, $|\hat{f}(\gamma^\star)|$ is $\sqrt{\eps}$-close to $\sqrt[4]{\sum_{\gamma\in \Gamma}\hat{f}(\gamma)^4}$.
\end{proof}

Since we do not have the exact value of $\sum_{V_0 + h}\hat{f}(\gamma)^4$, we next provide a way to estimate this value and thus, estimate the maximum coefficient in $V_0 + h$.

\begin{claim}\label{clm:estimate-of-fourth-power}
    For $a \in \{0,1\}^s$, let $V_a = V_0+h$ be an $\eps$-rare coset of $\Field_2^n$. There exists a randomized algorithm that, given $\eps, \delta > 0$ and query access to a function $f:\Field_2^n\to\{-1,1\}$, makes $O(\log(1/\delta)/\eps^4)$ queries to $f$ and outputs an estimate of $|\hat{f}(\gamma^\star)|$ up to error $\pm 2\sqrt{\eps}$ with probability at least $1-\delta$, where $\gamma^\star = \argmax_{\gamma\in V_0+h}|\hat{f}(\gamma)|$.
\end{claim}
\begin{proof}

We describe the randomized process for estimating $|\hat{f}(\gamma^\star)|$. Denote by $\omega$ the sum $\sum_{\gamma\in V_0+h}\hat{f}(\gamma)^4$. Notice that the vector $h$ is found by finding a solution to the following in $\poly(n)$ time

\[
\begin{bmatrix}
    \text{--- } r_1^T \text{ ---} \\
    \vdots \\
    \text{--- } r_s^T \text{ ---}
\end{bmatrix}\cdot h = a.
\]

We start by estimating $\omega$ using Lemma~\ref{lem:fourth-power-sum}. To evaluate the expression $\chi_h(w) \cdot f(x) \cdot f(y) \cdot f(z) \cdot f(x+y+z+w)$, we construct the following ``small" query set. Pick $x, y, z$ uniformly at random, and sample $w\sim V_0^\perp$ by picking a uniformly random linear combination of $r_1,\dots, r_s$. Then we query $f$ on the four inputs $x, y, z, x+y+z+w$. The value $\chi_h(w)$ is computed without consulting $f$ (we remark that samples $(x, f(x)), (y, f(y)), (z, f(z))$ can even be drawn from a random example oracle; in any case, for $f(x + y + z + w)$ we need query access.)

Repeating this procedure for $O(\Chernoff(\eps^2,\delta)) = O(\log(1/\delta)/\eps^4)$ independently chosen query sets and empirically computing the average $\estimate{\omega}$ over these sets, a Chernoff bound guarantees that $\estimate{\omega}$ is $\eps^2$-close to $\omega$ with probability at least $1-\delta$. Therefore with the same probability it holds that $\sqrt[4]{\estimate{\omega}}$ is $\sqrt{\eps}$-close to $\sqrt[4]{\omega}$, and as a result of \cref{clm:fourth-sum-good-estimate-for-max} it is $(2\sqrt{\eps})$-close to $|\hat{f}(\gamma^\star)|$.
\end{proof}

Observe that now not only do we want cosets of $V_0$ to be $\Lambda_t$-rare, but also we want them to be $\eps$-rare (for some small $\eps$) so that our estimations of the largest coefficients have small error. We are ready to state the algorithm and the following theorem. 

\begin{theorem}\label{thm:fourier-alg}
    For parameters $\eps, \delta>0$ and number $t \in \N$, there exists a randomized algorithm that, given membership query access to a Boolean function $f \colon \F_2^n \to \{-1,1\}$, makes $\poly(\log t, \frac{1}{\eps}, \log\frac{1}{\delta})$ queries to $f$, and outputs $\sigma_1 \geq \dots \geq \sigma_t$ in $[0,1]$ such that
    \begin{equation*}
      \Pr\left[\mbox{$\abs{|\fhat(\gamma_i)| - \sigma_i} \leq \eps$ for all $i \in [t]$} \right] \geq 1-\delta-o(1),
    \end{equation*}
    where $\gamma_1,\dots,\gamma_t$ are the vectors in $\F_2^n$ corresponding to the highest $t$ Fourier coefficients of $f$ in descending order of magnitude.
\end{theorem}

\begin{proof}
We start by describing the algorithm.

\begin{algorithm}[H]
    \SetAlgoLined
    \SetKwInOut{FunIn}{Function input}
    \SetKwInOut{ExpIn}{Explicit inputs}
    \FunIn{A function $f:\Field_2^n\to\{-1,1\}$, accessed through membership queries}
    \ExpIn{Parameters $t\in\N$, and $\eps,\delta>0$}
    
    Sample $s = 2\log_2(\frac{16}{\eps^4} + t) + \log_2(\frac{2}{\delta})$ uniformly random vectors $r_1,r_2,\dots,r_s$ from $\Field_2^n$. 

    For each coset $V_a$, where $a \in \{0,1\}^s$ (described in Equation \ref{equ:cosets}), estimate the maximum coefficient through $4\cdot\Chernoff(\frac{\eps^2}{2}, \frac{\delta}{2^{s+1}})$ membership queries using \cref{clm:estimate-of-fourth-power}, and call it $\estimate{\omega}_a$.
    
    Output $\sigma_1,\dots,\sigma_t$, the $t$ largest $\estimate{\omega}_a$ from all the $V_a$'s, in descending order.
\caption{$\GLstar(t, \eps, \delta)$ Finding the top $t$ Fourier coefficients}\label{alg:gl-star}
\end{algorithm}

Suppose our randomly chosen vectors $r_1,\dots, r_s$ are linearly independent. Let $\Lambda_t = \{\gamma_1,\dots,\gamma_t\}$ be the set of vectors corresponding to the top $t$ Fourier coefficients. More precisely for any $\gamma\in\F_2^n\setminus\Lambda_t$ we have $\hat{f}(\gamma_1) \ge\dots\ge\hat{f}(\gamma_t)\ge\hat{f}(\gamma)$. Notice that by Parseval's identity (\cref{fact:parseval}) $|\hat{f}^{>\eps^2/4}|\le\frac{16}{\eps^4}$, and therefore, $|\hat{f}^{>\eps^2/4}\cup\Lambda_t| \le \frac{16}{\eps^4} + t$, and thus by \cref{claim:separate-whp}, every coset of $V_0$ will be $(\eps^2/4)$-rare, \emph{as well as} $\Lambda_t$-rare\footnote{Some cosets may contain no element from either of $\hat{f}^{>\eps^2/4}$ or $\Lambda_t$. By definition, they are still rare.}, with probability at least $1-\delta/2$. This ensures that all $\Lambda_t$ coefficients are separated and can be estimated individually up to error $\eps$, which follows from \cref{clm:fourth-sum-good-estimate-for-max} and \cref{clm:estimate-of-fourth-power}. In addition, by union bound over the $2^s$ cosets, all of these estimates happen as intended with probability at least $1-\delta/2$.

The probability that $s$ uniformly random vectors in $\F_2^n$ are not linearly independent is at most $2^s/2^n = o(1)$. Put together, we conclude that for the highest $t$ Fourier coefficients $\gamma_1,\dots,\gamma_t$,
\[
\Pr\left[\abs{|\hat{f}(\gamma_i)| - \sigma_i} \le \eps \text{ for all } i\in [t]\right] \ge 1 - \delta - o(1).
\]

The total number of membership and uniformly random queries used is $O(\Chernoff(\frac{\eps^2}{2}, \frac{\delta}{2^{s+1}}))=\poly(\log t, \frac{1}{\eps}, \log\frac{1}{\delta})$, and the running time is $\poly(n, t, \frac{1}{\eps}, \log\frac{1}{\delta})$.
\end{proof}

\subsection{Interactive Protocol for Learning Top Fourier Characters}
\label{sec:pac-verification-top-fourier-coefficients}

In this section we provide an IP for learning the indices of heavy Fourier coefficients, i.e., vectors in $\F_2^n$ corresponding to the top coefficients of a given Boolean function. It is worth pointing out that this is a harder task than just outputting the largest $t$ Fourier coefficients, and is not achieved by \cref{alg:gl-star}. In particular, this is because our subspace partitioning system (bucketing system) from \cref{sec:alg-for-learning-fourier-values}, only helps us estimate the maximum coefficient for each coset, without giving us any information about which vector in the coset corresponds to this value. Our solution to this is to use the prover for learning the vectors corresponding to the top coefficients (in addition to the coefficients themselves). 

In this section we first propose an interactive protocol for learning the top $t$ vectors associated with the top $t$ coefficients, where the verifier has membership query access to the input function. The high level idea is that the prover provides a proof: a set of top $t$ vectors associated with largest coefficients. Then the verifier uses \cref{alg:gl-star} to check the correctness of this proof. We then transform this protocol to finally prove \cref{thm:fourier-random}.

\begin{theorem}\label{thm:fourier-membership}
    There exists an interactive protocol with verifier $V$ and prover $P$ where, given $\eps, \delta > 0$ and a number $t\in\N$ and membership access to a function $f:\F_2^n\to\{-1,1\}$, the verifier makes $\poly(t,\frac{1}{\eps}, \log\frac{1}{\delta})$ \emph{membership} queries (independent of $n$) to $f$ and outputs $\Lambda$, which is either a size-$t$ ``top'' subset of $\F_2^n$ or $\bot$ (reject). This protocol has the following interactive parameters

    \begin{equation*}
         \InteractiveParams{\membership}{\poly(t,1/\eps, \log(1/\delta))}{t\cdot n}{\proofonly}{\poly(n,t,1/\eps,\log(1/\delta))}{\poly(n,t,1/\eps,\log(1/\delta))}
        \enspace,
    \end{equation*}

    and satisfies the following completeness and soundness conditions.

    \begin{itemize}
        \item \textbf{Completeness:} There exists an honest prover $P^*$ such that the output of interaction $\ip{P^*,V}=\Lambda$ satisfies $\Pr\left[\Lambda \text{ is } \frac{\eps}{3}\text{-top}\right] \geq 1-\delta-o(1)$.
    
        \item \textbf{Soundness:} For any (possibly unbounded) prover $P$, the output of interaction $\ip{P, V}=\Lambda$ satisfies $\Pr\left[(\Lambda \neq \bot) \land (\Lambda \text{ is not } 
 \eps\text{-top})\right] \leq \delta + o(1).$
    \end{itemize}
\end{theorem}

\begin{proof}
The interactive protocol is the following:

\begin{protocol}[H]
    \SetAlgoLined
    \SetKwInOut{FunIn}{Function input}
    \SetKwInOut{ExpIn}{Explicit inputs}
    \FunIn{$f:\Field_2^n\to\{-1,1\}$, accessed by verifier through membership queries}
    \ExpIn{Parameters $t\in\N$ and $\eps,\delta>0$}

    The prover sends a set $\Lambda'_t \subseteq \Field_2^n$ of size $t$ purportedly containing vectors corresponding to the top $t$ heaviest Fourier coefficients of $f$ (enumerated by the decreasing order).

    The verifier estimates all the coefficients in $\Lambda'_t$ using a total of $t\cdot\Chernoff(\frac{\eps}{6},\frac{\delta}{2t})$ random samples\footnote{We can empirically estimate $\hat{f}(\gamma)$ since $\hat{f}(\gamma)=\E_x[\chi_\gamma(x) f(x)]$.}. Let $\lambda\in\Lambda'_t$ be the vector that has the smallest estimated coefficient. Let $\estimate{c}_\lambda$ denote this estimate of $\hat{f}(\lambda)$.

    The verifier runs $\GLstar(t, \frac{\eps}{6}, \frac{\delta}{2})$ using query access to $f$ (as described in \cref{alg:gl-star}) and uses it to compute $\estimate{\omega}_a$ for all cosets $V_a = V_0 + h$ (see \cref{equ:cosets}), with respect to a random $V_0$ obtained during its run.

    The verifier outputs $\bot$ (reject) if there exists an $a\in\Field_2^s$ such that\footnote{For any vector $\gamma$, one can efficiently check if $\gamma\in V_a$ by looking at the inner product of $\gamma$ with each $r_1,\dots r_s$.} $V_a\cap\Lambda'_t=\emptyset$ and $|\estimate{c}_\lambda| + \frac{2\eps}{3} < \estimate{\omega}_a$. Otherwise the verifier outputs $\Lambda'_t$.
    
\caption{IP for learning top $t$ Fourier characters with error $\eps$ and confidence $\delta$}\label{alg:ip-for-top-t-fourier-coefficients}
\end{protocol}
\begin{remark}
It is worth pointing out that the verifier does not run $\GLstar$ as a black box, since it needs to know the vectors used to describe the random affine subspace partition given by the $V_a$'s. The verifier rather uses this algorithm to estimate the maximum coefficients for every coset of $V_0$.
\end{remark}

    First we show several desirable events happen with high probability, regardless of the set $\Lambda'_t$ that the prover sends, i.e., prover being honest or dishonest. By Chernoff bound estimating a single corresponding coefficient in $\Lambda'_t$ with $O(\Chernoff(\frac{\eps}{6}, \frac{\delta}{2t}))$ random examples up to error $\pm \eps/6$ will be correct with probability at least $1-\frac{\delta}{2t}$. Then by a union bound, all these coefficients are correctly estimated with probability at least $1-\delta/2$. Moreover, from the guarantees of \cref{alg:gl-star}, with probability at least $1-\delta/2$, all cosets of $V_0$ will be $(\Lambda'_t\cup \hat{f}^{>\eps/6})$-rare and all the estimates of $\estimate{\omega}_a$ are correct up to $\pm\eps/6$ (in fact, cosets of $V_0$ from \cref{alg:gl-star} with parameters $\GLstar(t, \frac{\eps}{6}, \frac{\delta}{2})$ are $(\Lambda'_t\cup \hat{f}^{>\eps^2/144})$-rare, but a weaker statement suffices for our purposes.) Put together, all of the aforementioned events happen with probability at least $1-\delta$. 
    
    Now, suppose that all the desirable events happened, we prove the interactive proof satisfies the completeness and soundness requirements.
    
    \paragraph{Completeness:}
    Suppose $\Lambda'_t$ is $\frac{\eps}{3}$-top. Fix any $\alpha \in \F_2^n\setminus\Lambda'_t$ and any $\beta \in \Lambda'_t$. Then we have $|\hat{f}(\alpha)| \le |\hat{f}(\beta)| + \eps/3$ by definition. Provided that all cosets of $V_0$ are $\Lambda'_t$-rare, for all cosets $V_a$ such that $V_a\cap\Lambda'_t = \emptyset$ we have $\estimate{\omega}_a \le |\estimate{c}_\lambda| + \frac{2\eps}{3}$ and the verifier outputs $\Lambda=\Lambda'_t$.
    
    \paragraph{Soundness:}
    Suppose $\Lambda'_t$ is not $\eps$-top. By definition there exists some $\alpha\in\F_2^n\setminus\Lambda'_t$ and some $\beta\in\Lambda'_t$ such that $|\hat{f}(\alpha)| > |\hat{f}(\beta)|+\eps$. Since each coset of $V_0$ is $(\Lambda'_t\cup \hat{f}^{>\eps/6})$-rare, and since $\alpha\in\hat{f}^{>\eps/6}$, there exists some coset $V_{a'}$ such that $\alpha\in V_{a'}$ and $V_{a'}\cap\Lambda'_t = \emptyset$. Without loss of generality, assume $\alpha$ is the vector with maximum corresponding coefficient inside $V_{a'}$. Since $\lambda\in\Lambda'_t$ is the vector such that $|\estimate{c}_\lambda|$ is estimated to be minimum among all vectors in $\Lambda'_t$, we have $|\estimate{c}_\lambda| \le |\hat{f}(\beta)| + \frac{\eps}{6}$ and as a result $|\estimate{c}_\lambda| < |\hat{f}(\alpha)| - \frac{5\eps}{6}$. Recall that $\estimate{\omega}_{a'}$ is $\frac{\eps}{6}$-close to $|\hat{f}(\alpha)|$ yielding $|\estimate{c}_\lambda| < \estimate{\omega}_{a'} - \frac{2\eps}{3}$, and thus the verifier rejects.
\end{proof}

Of additional interest is the fact that this protocol with verifier's membership queries is \textsf{proof only} (MA-like), and the query complexity of verifier is independent of $n$. We will now transform this protocol with membership queries into one where the verifier uses only random examples, thereby finally proving \cref{thm:fourier-random}.

\begin{proof}[\proofof{\cref{thm:fourier-random}}]
We prove this theorem by transforming \cref{alg:ip-for-top-t-fourier-coefficients} via our general membership to random transformation \cref{thm:membership-to-random-transform}. 

First, we remark that our IP with membership queries is \textsf{proof only}. Furthermore, all the queries made by the verifier are non-adaptive. Finally the query pattern is linear (\cref{def:linear-query-pattern}) and by \cref{lem:linear-query-pattern} it is possible to construct embedded query sets using fresh random examples. Thus all the necessary conditions of the membership to random transformation are satisfied and the transformation from \cref{thm:membership-to-random-transform} can be used in a black-box fashion, to get the desired interactive proof where the verifier only uses random examples. 
\end{proof}

\section{Interactive Proofs for Learning Boolean Circuits}
\label{sec:agnostic_verify_circuits}
In Section \ref{sec:aip_erm_unbounded_prover}, we show that if the honest prover is allowed to have unbounded computational power, then we can construct interactive proofs for agnostic learning $\Ppoly$, using just $O(1/\varepsilon)$ labeled examples, even in the distribution-free scenario. In this section, we study doubly-efficient interactive proofs for agnostic learning various Boolean circuit classes, starting with shallow-depth circuits.

Let $\cAC_d^0[2][\poly(n)]$ be the class of Boolean functions computable by $\poly$-sized depth-$d$ circuits, whose gate-set is $\{ \land, \lor, \neg, \oplus \}$ of unbounded fan-in. Similarly, we consider the class $\cAC_d^0[p][\poly(n)]$, whose gate-set includes $\mathsf{mod}_p$ gates of unbounded fan-in instead of $\oplus$ gates. 

Our main result for this section is an \textsf{IP} for agnostic learning $\cAC^0_d[2]$ using quasi-polynomially many samples, such that the verifier and the honest prover also have quasi-polynomial running time. These ideas can be easily extended to get similar interactive proofs for learning $\cAC^0_d[p]$, for any prime $p > 2$. 
\begin{theorem}[Formal statement of \cref{thm:our-results-ac0[2]-pac-verify}]
    \label{thm:agnostic_ip_ac0}
    Let $f : \{0,1\}^n \rightarrow \{0,1\}$ be any function such that $\opt(f,\cAC^0[2]) > 1/n^{\frac{\log(n)}{2}}$. Then, the class $\cAC^0[2][\poly(n)]$ is $(O(\log^{16d+8}(n)),1/10)$-agnostic verifiable over the uniform distribution with the following parameters.
    \begin{equation*}
        \InteractiveParams{\random}{\exp{(O(\log^{4d+2} (n)))}}{\exp{(O(\log^{4d+2} (n)))}}{\proofonly}{\exp{(O(\log^{4d+2} (n)))}}{\exp{(O(\log^{4d+2} (n)))}}
    \end{equation*}
\end{theorem}
In fact, we actually show that for every $c \geq 1$, we get a PAC-verification protocol for all $f$ such that $\mathsf{opt}_{\cAC^0[2]}(f) > 1/n^{\frac{\log^c(n)}{2}}$, where the verifier has random example access to $f$.

Finally, in \cref{thm:tolerant_natural_implies_aip}, we show that for any circuit class $\sC$ that contains $\cAC^0[2]$, the existence of certain kinds of constructive lower bounds against polynomial-sized $\sC$-circuits can be used to construct interactive proofs for \textit{improper} learning $\sC$ using sub-exponentially many samples, with a sub-exponential time honest prover (and verifier). All these results go through the framework of \cite{CIKK17_agnostic}, that use such constructive lower bounds against $\sC$ to get agnostic learners for $\sC$ that use membership queries. 

To begin with, we set up some preliminaries in Section \ref{sec:cikk_learn_prelims}, following which we sketch the agnostic learning framework for by \cite{CIKK17_agnostic} in Section \ref{sec:cikk_sketch} that forms the basis of the interactive proofs. Finally, we construct the interactive proofs for learning $\cAC^0[2]$ (and $\cAC^0[p]$, for prime $p > 2$), as well as for general circuit classes $\sC$ in Section \ref{sec:proof_aip_ac0}.

\subsection{Section Preliminaries}
\label{sec:cikk_learn_prelims}
The main premise of \cite{CIKK17_agnostic} is that proving ``constructive" \textit{average-case lower bounds} against a circuit class $\sC$ implies an agnostic learning algorithm for $\sC$ over the uniform distribution using membership queries, in a \textit{black-box} manner. This extends their previous work \cite{CIKK16} that shows a similar connection between worst-case lower bounds and standard learners for $\sC$ over the uniform distribution using membership queries.

\subsubsection{Nisan-Wigderson Generators}
\label{sec:nw_basics}

\begin{definition}[NW Set Design]
\label{def:nw_sets}
  For every $n, m, L \in \N$, define a sequence of sets $S_1, \dots, S_L \subseteq [m]$ as an $(n,L,d)$-design if
  \begin{enumerate}
  \item $\vert S_i \vert = n$, for all $1 \leq i \leq L$.
  \item For all $1 \leq i \neq j \leq L$, $\vert S_i \cap S_j \vert \leq d$.
  \end{enumerate}
\end{definition}

\begin{definition}[NW Generator \cite{NW94}]
  Let $f: \{0,1\}^n \rightarrow \{0,1\}$. For $m = O(n^2)$ and stretch $L = L(m)$, let $S_1, \dots, S_L \subseteq [m]$ be an $(n,L,\log L)$-design. Then, the NW generator $G^f : \{0,1\}^m \rightarrow \{0,1\}^{L(m)}$ is defined as $G^f(z) = f(z \vert_{S_1}), f(z \vert_{S_2}), \dots, f(z \vert_{S_L})$.
\end{definition}

\paragraph*{NW reconstruction algorithm \cite{IW01}}
The uniform NW reconstruction algorithm is an essential ingredient in the learning algorithms in Section \ref{sec:cikk_sketch}. To understand this, we first define the notion of a \textit{distinguisher circuit family}. For any $m$ and $L = L(m)$, let $\sG = \{ G_m : \{0,1\}^m \rightarrow \{0,1\}^L \}$ be a family of generators. For any $\varepsilon > 0$, we define $D_L^{G,\epsilon}$ as the set of all circuits $D$ on $L$ inputs such that
\begin{equation*}
\bigg\vert \Pr_{z \sim U_m}[D(G_m(z)) = 1] - \Pr_{y \sim U_L}[D(y) = 1] \bigg\vert \geq \epsilon.  
\end{equation*}

For any $n$-variate Boolean function $f$, the NW reconstruction algorithm gets a circuit that distinguishes the output of the NW generator $G^f$ from the uniform distribution over $L(m)$ bits and oracle access to $f$, and outputs a small-sized circuit that computes $f$ on a non-trivial fraction of strings of length $n$. We sketch it here to help us make further observations that will be useful later. The proof of correctness follows from a hybrid argument over the output distributions of $G^f$ and $U_L$ and can be found in \cite{NW94,IW01}.
\begin{lemma}
  \label{lem:iw01}
  Let $G^f : \{0,1\}^{m} \rightarrow \{0,1\}^{L}$ be the NW generator constructed based on $f:\{0,1\}^{n} \rightarrow \{0,1\}$ using an $(n,L,\log L)$-design $S_1, \dots, S_L$. Suppose $D$ is a distinguisher circuit in $D_L^{G^f, 1/L}$ for $G^f$. Then, there exists a randomized algorithm $\sA$ (or a weak learner) with oracle access to $f$ running in time $\poly (L)$, which takes input $D$ and outputs a circuit $C$ that agrees with $f$ on at least $\frac{1}{2}+ \Omega \left( \frac{1}{L} \right)$ fraction of the inputs. 
\end{lemma}

\begin{proof}[Proof Sketch]
The weak learner follows in Algorithm \ref{algo:iw01_reconstruct}. The learner is split into the query construction phase and the learning phase.
\begin{algorithm}
\caption{Weak learner based on the uniform reconstruction for the NW generator.}
\label{algo:iw01_reconstruct}
    \nonl \textbf{Function input:} Query access to some function $f:\{0,1\}^n \rightarrow \{0,1\}$ \\
    \vspace{0.05in}
    \textbf{Explicit inputs:} Distinguisher circuit $D \in D_L^{G^f, 1/L}$ and $\varepsilon > 0$. \\
    \begin{itemize}
        \item \textbf{Query construction algorithm $\sA_{\mathsf{query}}$.}
        \begin{enumerate}
            \item Pick $i$ uniformly at random from $[L]$.
            \item For each $1 \leq j \leq m$, with $j \notin S_i$, set the $j^{\text{th}}$ input to $G^f$ to a random bit $z_j$ .
            \item For each $1 \leq k < i$, generate all the $2^{\vert S_k \cap S_i \vert} \leq L$ strings that are consistent with the partial assignment to $z \vert_{S_k}$ made in the step 2. Let $\sQ$ be these set of queries.
        \end{enumerate}
            
        \item \textbf{Learning algorithm $\sA_{\mathsf{learn}}$ given membership query set $\sQ$.}
        \begin{enumerate}
        \item Query $f$ at every string in $\sQ$ and these answers in a table $T$.
        \item For each $i \leq k \leq L$, set the $k^{\text{th}}$ input to $D$ to a random bit $y_k$.
        \item \textbf{Hypothesis construction:} Using $T$, $y_k$'s, $z_j$'s, and the encoding of $D$, construct a hypothesis circuit $C$ that does the following on input $x \in \{0,1\}^{n}$.
            \begin{enumerate}
                \item Set $z \vert_{S_i} = x$ and get the full input $z \in \{0,1\}^{m}$ to $G^f$.
                \item Compute $D(f(z \vert_{S_1}), \dots, f(z \vert_{S_{i-1}}), y_i, \dots, y_L)$, by obtaining the corresponding values for $f(z \vert_{S_k})$, $1 \leq k < i$, from $T$.
                \item If the output of $D$ is $1$, output $y_i$, else output $\overline{y_i}$.
            \end{enumerate}
        \end{enumerate}
        \item Repeat the above algorithm $\poly(L)$ times, and compute an empirical estimate of the agreement of the obtained circuit with $f$ in each run (by making $\poly(L)$ fresh uniformly random queries to $f$). Finally, output the best circuit in the list. Using a Chernoff bound, we see that with high probability, the final output hypothesis circuit agrees with $f$ on at least a $1/2 + \Omega (1/L)$ fraction of the inputs.
    \end{itemize}
\end{algorithm}
\end{proof}

\subsubsection{Tolerant Natural Property}
\label{sec:tol_nat_prop}
Next, we formally define a constructive average-case lower bound via the notion of a tolerant natural property.
\begin{definition}[Tolerant Natural Properties]
  \label{def:tolerant}
  A combinatorial property of the functions $\{\sR_n\}_{n \geq 0}$ is $\sD$-natural that is $\tau$-tolerant useful against $\sC$-circuits of size $u(n)$, if it satisfies the following:
  \begin{enumerate}
  \item \textbf{$\sD$-Constructivity}: Given a truth table of a function $f_n$, there exists an algorithm computable in $\sD$ that decides if $f_n \in \sR_n$.
  \item \textbf{Largeness}: A uniformly random function satisfies $\sR_n$ with probability at least $1/\poly(2^n)$.
  \item \textbf{usefulness against $\sC[u]$ with $\tau(n)$-tolerance}: For every $f_n \in \sF_n$ ($n$ is large enough), if $\mathsf{err}(f_n, \sC[u(n)]) \leq \tau(n)$, then $f_n \notin \sR_n$.
  \end{enumerate}
\end{definition}
The usefulness reflects the circuit size against which the lower bound holds and the tolerance reflects the quality of the average-case hardness. 

Based on the lower bound by Razborov and Smolensky \cite{Raz87,Smo87}, \cite{CIKK17_agnostic} construct the $(1/n^3)$-tolerant natural property useful against $\cAC^0[2]$.  
\begin{lemma}
    \label{lem:tolerant_ac0_property}
    There exists a $\mathsf{P}$-natural property $\{\sR_n\}_n$ with largeness $1/2$, that is $(1/n^3)$-tolerant useful against $\cAC^0[2]$ circuits of size $\mathsf{exp}(\Omega(n^{1/2d}))$ and depth $d$. 
\end{lemma}

\subsection{Agnostic Learning from Tolerant Natural Properties}
\label{sec:cikk_sketch}

\paragraph*{Agnostic learning algorithm for $\cAC^0[2]$}
Using the tolerant natural property useful against $\cAC^0[2]$ circuits from Lemma \ref{lem:tolerant_ac0_property}, \cite{CIKK17_agnostic} construct the following agnostic learner. 
\begin{theorem}[Agnostic Learner for $\cAC^0${$[2]$} (building on \cite{CIKK17_agnostic})]
  \label{thm:cikk_ac0[2]_learn}
  Let $f : \{0,1\}^n \rightarrow \{0,1\}$ be a Boolean function such that $\dist(f,\cAC^0_d[2][\poly(n)]) = \beta^*$, where $\beta^*$ is a non-negligible number in $\mathbb{R}$ and $d$ is a constant in $\N$.
  
  Then, there exists a randomised algorithm $\sB$, that, given as inputs any $\beta$, where $\beta^* \leq \beta \leq 1/2$ and membership query access to $f$, outputs a hypothesis with error at most $O(\log^{8d+4} (n/\beta)) \cdot \beta$ (over the uniform distribution) with probability at least $0.9$. Moreover, $\sB$ runs in time $2^{O(\log^{2d+1} (n/\beta))}$ and makes at most $2^{O(\log^{2d+1} (n/\beta^*))}$ membership queries to $f$. 

  In particular, given as input the true distance parameter $\beta^*$, $\sB$ outputs a hypothesis that has a $O(\log^{8d+4} (n/\beta^*))$-factor approximation gap with respect to $\beta^*$.
\end{theorem}
Our statement of the agnostic learner is slightly different from that of \cite{CIKK17_agnostic}, as they assume that the learner gets the correct value of $\beta^*$ as input, whereas we extend their analysis to remove this assumption by giving any $\beta \geq \beta^*$ as input. In fact, as we shall in the agnostic IP, doing this ensures that we can search for the right value of $\beta^*$ up to constant factors.

The algorithm is obtained by a ``play-to-lose" argument. The main idea is that with high probability on the seeds, $G^f : \{0,1\}^{\poly(n)} \rightarrow \{0,1\}^L$ produces a truth table of a function over inputs of length $\log L$, which is $(1-4\beta^*(n))$-close to some function in $\cAC^0[2]$. As long as $\beta^*(n)$ is less than $\tau(\log L)$/4, the tolerant natural property acts as a distinguisher to $G^f$; indeed, by definition, it rejects almost all of the truth tables output by $G^f$, whereas it accepts a dense fraction of the uniformly random strings from $\{0,1\}^L$. Thus, we can directly use the natural property as a distinguisher in the NW-reconstruction algorithm from Lemma \ref{lem:iw01} to get a circuit that computes $f$ on a $(1/2 + 1/L)$-fraction of the $n$-length inputs. 

To get a stronger learner, the NW generator is instantiated with the concatenated encoding of $f$ by the Direct Product and Hadamard codes, i.e., $\mathsf{Had} \circ \mathsf{DP} (f)$. The learner now runs the local list-decoding algorithms for Hadamard and Direct Product codes \cite{IJKW10} in sequence using the circuit output by the reconstruction algorithm (that is a weak learner), to get stronger guarantees on the error of the hypothesis.

In more detail, for any given $f \in \sF_n$ and $k \in \N$, define the amplified function $\mathsf{Amp}^f_k : \{0,1\}^{nk+k} \rightarrow \{0,1\}^{}$ as the Hadamard code encoding of the $k$-wise direct product of $f$. In other words, for inputs $x_1, \dots, x_k \in \{0,1\}^{n}$ and $y_1, \dots, y_k \in \{0,1\}^{k}$, 
\begin{equation*}
\mathsf{Amp}^f_k(x_1, \dots, x_k, y_1, \dots, y_k) = \sum_{i=1}^k f(x_i) \cdot y_i \pmod 2    
\end{equation*}
Let $G^{\mathsf{Amp}^f_k}$ be the NW-generator instantiated using $\mathsf{Amp}^f_k$ having seed length $m = \poly(n,k)$ and stretch $L(n)$. 

We next state the following result that establishes the complexity of any output of an NW-generator, when treated as a truth table of a function over $\ell = \log L$ inputs.
\begin{lemma}[Theorem 3.2 in \cite{CIKK16}]
    \label{lem:cikk_nw_output}
    For any string $z \in \{0,1\}^m$ and any $f \in \cAC^0[2][\poly(n)]$, let $g_z : \{0,1\}^\ell \rightarrow \{0,1\}$ be defined as $g_z(i) = [G^{\mathsf{Amp}^f_k}(z)]_i = \mathsf{Amp}^f_k(z \vert_{S_i})$, for any $i \in \{0,1\}^\ell$. Then, $g_z \in \cAC^0[2][\poly(n,k)]$. 
\end{lemma}

\begin{proof}[Proof sketch of Theorem \ref{thm:cikk_ac0[2]_learn} (see \cite{CIKK17_agnostic} for detailed exposition)]
Let $\{\sR_{2^n}\}$ be a \textsf{P}-natural property from Lemma \ref{lem:tolerant_ac0_property} that is useful against $\cAC^0[2]$-circuits of depth $d$ and size $u(n) = \mathsf{exp}(O(n^{1/2d}))$ with tolerance $\tau(n) = 1/n^3$. For any $\beta$ such that $\beta^* \leq \beta \leq 1/2$, let $G^{\mathsf{Amp}^f_k}$ be the NW-generator instantiated using $\mathsf{Amp}^f_k$ having stretch $L = 2^{O(\log^{2d+1} (nk))}$, along with the amplification parameter being set to $k= \tau(\ell)/8\beta(n)$. Let $\ell = \log (L) = O(\log^{2d+1}(n/\beta))$.

We first outline the key intuition behind the learning algorithm. Firstly, using Lemma \ref{lem:cikk_nw_output}, it holds that for any $g$ in $\cAC^0[2][\poly(n)]$ and for every input seed $z$, $G^{\mathsf{Amp}_k^g}(z)$ is a function computable by $\cAC^0[2]$ circuits of size $\poly(n,k) \leq u(\ell) = 2^{O(\ell^{1/2d})}$, for the chosen value of $\ell = \log L$. Next, for the chosen value of $k$, with high probability over the choice of the seeds, the string output by $G^{\mathsf{Amp}_k^g}(z)$ is at a (relative) distance of at most $\tau(\ell)$ from $\cAC^0[2][u(\ell)]$.\footnote{In more detail, when $\dist(f,g) = \beta^*$, we can show that $\dist(\mathsf{Amp}^f_k,\mathsf{Amp}_k(g)) \leq 1/2 - (1-2\beta^*)^k/2$. Using standard inequalities, $1/2 \cdot (1-(1-2\beta^*)^k)) \leq 1/2(1-e^{-4\beta^*k}) \leq 2k\beta^*$. When $k=\tau(\ell)/4\beta$, observe that $\dist(\mathsf{Amp}^f_k,\mathsf{Amp}_k(g)) \leq \tau(\ell) \cdot \frac{\beta^*}{4\beta}$, that in turn is at most $\tau(\ell)/4$ for every $\beta \geq \beta^*$.} Thus, the $\tau(\ell)$-tolerant natural property $\sR_L$ useful against $\cAC^0[2][u(\ell)]$ can distinguish between outputs of $G^{\mathsf{Amp}^f_k}$ and $\sU_L$ (the uniform distibution over $\{0,1\}^L$) and in turn breaks the generator $G^{\mathsf{Amp}^f_k}$. As such, we can use $\sR_L$ as the distinguisher input to the NW reconstruction algorithm and obtain a hypothesis approximates $f$ with non-trivial advantage over a random guess.

The agnostic learner $\sB$ is described as follows. For any input function $f$, $\sB$ just runs the $\poly(L)$-time weak learner $\sA$ from Lemma \ref{lem:iw01} using query access to $\mathsf{Amp}^f_k$ and the distinguisher circuit $D$ that computes $\sR_L$, based on the NW-generator $G^{\mathsf{Amp}^f_k}$. The oracle access by $\sA$ to $\mathsf{Amp}^f_k$ is simulated using the oracle to the input $f$. $\sA$ outputs a circuit on $n$ inputs that agrees with $\mathsf{Amp}^f_k$ on at least $1/2 + \Omega(1/L)$ fraction. 

The uniform reconstruction algorithms that accompany hardness amplification techniques using the Hadamard code (i.e., the Goldreich-Levin theorem) and the Direct Product code \cite{IJKW10}, ensure that with high probability over its randomness, $\sB$ constructs a circuit $h$ on $n$ inputs that computes $f$ on at least $1-O(\ell/k)$-fraction of the inputs, in $\poly(n,L,k) \leq 2^{O(\log^{2d+1}(n/\beta))}$ time.\footnote{We can use standard error amplification techniques to make the learner output this hypothesis with probability at least $1-\delta$, where $\delta > 0$ is the confidence parameter~\cite{KV94_book}.} From Lemma \ref{lem:iw01}, we also see that the query complexity is at most $\poly(L) \leq 2^{O(\log^{2d+1}(n/\beta))}$. Finally, the output hypothesis $h$ computes $f$ on at least $1-O(\beta \ell/\tau(\ell))$ fraction of the inputs, i.e., $\mathsf{err}(h,f) \leq O(\log^{8d+4}(n/\beta))\cdot \beta$. 

\end{proof}

The following corollary follows from Theorem \ref{thm:cikk_ac0[2]_learn} when the input function is not too close to $\cAC^0[2]$, i.e., $\dist(f,\cAC^0[2]) \geq 1/n^{\frac{\log(n)}{2}}$.
\begin{corollary}
\label{cor:query_max_size}
Let $\beta^* \geq 1/n^{\frac{\log(n)}{2}}$ and let $f$ be a Boolean function such that $\dist(f,\cAC^0_d[2][\poly(n)]) = \beta^*$, for constant $d \in \N$. Then, for input $\beta$ such that $\beta^* \leq \beta \leq 1/2$, the agnostic learning algorithm $\sB$ uses membership queries to $f$ and outputs a hypothesis that has error at most $O(\log^{16d+8}(n)) \cdot \beta$ with probability at least $0.9$, having running time and query complexity at most $2^{O(\log^{4d+2}(n))}$.
\end{corollary}

\begin{proof}
    The proof is obtained by setting the right parameters in Theorem \ref{thm:cikk_ac0[2]_learn}. In particular, since the smallest value of $\beta$ is $1/n^{\log(n)}$, setting $\ell$ to be $O(\log^{2d+1} (n \cdot n^{\log(n)})) = O(\log^{4d+2}(n))$, ensures that for every $\beta \geq 1/n^{\log(n)}$, when the value of $k$ is $\frac{1}{8 \ell^3 \beta}$, for any $g \in \cAC^0[2][\poly(n)]$, the output of $G^{\mathsf{Amp}_k^g}$ on every seed is computable by a circuit of size $\poly(n,k) = \poly(n/\beta) \leq u(\ell) = \mathsf{exp}(O(\ell^{1/2d}))$, where $u(\ell)$ is the usefulness parameter of the natural property $\{ \sR_n \}$ against $\cAC^0[2]$. 

    Thus, the running time and query complexities of the algorithm are at most $\poly(n,L,k) = 2^{O(\log^{4d+2}(n))}$, where $L = 2^{\ell}$. Moreover, for every $\beta$, the guarantees of the learner $\sB$ ensures that it outputs a hypothesis that has error at most $\ell^4 \beta = O(\log^{16d+8}(n)) \cdot \beta$, with probability at least $0.9$.
\end{proof}

\paragraph*{Agnostic learning of $\cAC^0[p]$, for every prime $p > 2$}
The learner follows from the general strategy as $\cAC^0[2]$, however \cite{CIKK17_agnostic} extend the class $\cAC^0[p]$ to include functions of the form $f : \{0,1\}^n \rightarrow \F_p$, e.g., by having $p$ output gates and the true output selects the field element (constant inputs from $\F_p$ to the circuit can also be coded this way). Following this, they show the existence of a $(0.15)$-tolerant natural property useful against $\cAC^0[p]$ circuits of depth $d$ and size $\exp(\Omega(n^{1/2d}))$ under this extended notion. Finally, they use this natural property to construct an agnostic learner for $\cAC^0[p]$, where the input distribution is still Boolean, i.e., the distribution is $(\sU_n, f(\sU_n))$ for $f : \{0,1\}^n \rightarrow \{0,1\}$.

\begin{lemma}[Agnostic Learners for $\cAC^0${$[p]$}]
  \label{lem:cikk_ac0[p]_prime_learn}
  For any prime $p > 2$, let $f : \{0,1\}^n \rightarrow \{0,1\}$ be such that $\dist(f,\cAC^0_d[p][\poly(n)]) = \beta^*$, where $\beta^* \geq 1/n^{\frac{\log(n)}{2}}$ and $d$ is a constant in $\N$.
  
  Then, there exists a randomised algorithm $\sB$, that, given as inputs any $\beta$ such that $\beta^* \leq \beta \leq 1/2$, and membership query access to $f$, outputs a hypothesis that has error at most $O(\log^{16pd+8p}(n)) \cdot \beta$ with probability at least $0.9$, that has running time and query complexity at most $2^{O(\log^{4pd+2p}(n))}$.

\end{lemma}

The learning algorithm is the same as that of the case $\cAC^0[2]$ with some minor modifications. Firstly, the amplified function $\mathsf{Amp}^f_k : \{0,1\}^{nk} \times \F_p^k \rightarrow \F_p$ is different here and is given by $\mathsf{Amp}^f_k(x_1, \dots, x_k, b_1, \dots, b_k) = \sum_{i=1}^k f(x_i) \cdot b_i \pmod{p}$, where $x_1, \dots, x_k \in \{0,1\}^n$ and $b_1, \dots, b_k \in \F_p$. As can be seen, the output of the $\mathsf{Amp}^f_k$ is an element in $\F_p$, and thus combined with the set design (that is also constructible using polynomials over $\F_p$), the resulting truth table output by $G^{\mathsf{Amp}^f_k}$ would be a vector in $\F_p^L$. Since the tolerant natural property can now deal with such truth tables, the framework from Theorem \ref{thm:cikk_ac0[2]_learn} can now be applied here. In particular, the weak learner from Algorithm \ref{algo:iw01_reconstruct} is still used here, where oracle access to the non-Boolean function $\mathsf{Amp}^f_k$ is in fact simulated by making membership queries to the Boolean function $f$. 

Moreover, the stretch $L$ of the generator is set to $2^{O(\log^{4pd+2p}(n))}$, where $p$ comes up in the exponent because of the overhead of the Goldreich-Levin reconstruction in producing a binary string in $\{0,1\}^k$ as output (to represent the $f(x_1), \dots, f(x_k)$).

\paragraph*{Agnostic learning of a general circuit class $\sC$}
We end this section by stating the following general result for agnostically learning any circuit class $\sC$ that contains $\cAC^0[2]$, given any suitable strong tolerant natural property useful against $\sC$. 

\begin{lemma}[Tolerant natural properties imply agnostic learners]
  \label{lem:cikk_general_circuits}
    Let $\sR$ be a $(1/2-1/\poly(n))$-tolerant natural property useful against $\sC[\poly(n)]$, where $\sC$ is any circuit class that contains $\cAC^0[2]$.
    
    Then, for every $a \geq 1$ and $\gamma > 0$, there exists a randomised algorithm $\sB$, that given as input $\beta$ such that $ \frac{1}
    {n^{2a\gamma}} \leq \beta^* \leq \beta \leq 1/2$, and membership query access to $f$, such that $\dist(f,\sC[n^a]) = \beta^*$, outputs a hypothesis with error at most $O(n^\gamma/\log(n)) \cdot \beta$ over the uniform distribution, with probability at least 0.9. Moreover, $\sB$ runs in time $2^{O(n^\gamma)}$ and makes at most $2^{O(n^\gamma)}$ membership queries to $f$.

\end{lemma} 

\begin{proof}[Proof Sketch]
The general outline is the same as before, however, the learner uses a variant of the NW-generator in order to use the stronger tolerance of $1/2-1/\poly(n)$ more fruitfully. Indeed, adding a pairwise-independent string generator to the NW-generator, helps us in using a stronger concentration bounds, and allows us to show that with high probability over the choice of the seed, the functions output by this modified generator have distance at most $1/2-1/\poly(\ell))$ from $\sC[n^a]$. 

In more detail, let $\mathsf{PI} : \Bits^\ell \times \Bits^{m'} \rightarrow \{0,1\}^n$, be a pairwise independent generator. Then the modified NW generator $G'^f : \Bits^m \times \Bits^{m'} \rightarrow \Bits^L$, is defined as $G'^{f}(z_1,z_2) = f(z \vert_{S_1} \oplus \mathsf{PI}(1,z_2)), \dots, f(z \vert_{S_L} \oplus \mathsf{PI}(L,z_2))$. Following this, the rest of the algorithm follows in a similar fashion as \cref{thm:cikk_ac0[2]_learn}. \end{proof}

\begin{remark}
While we state Theorem \ref{lem:cikk_general_circuits} for tolerance $\tau(n) = 1/2 - 1/\poly(n)$ and usefulness $\poly(n)$ as an example, \cite{CIKK17_agnostic} provide a more generalised form for any varying parameters of usefulness $u(n)$ and tolerance $1/2 - \tau'(n)$; the larger the value of $u(n)$, the smaller the running time and the smaller the value of $\tau'(n)$, the better the multiplicative factors in the hypothesis approximation. In the best case, with exponential usefulness $u$ and inverse exponential tolerance $\tau'$, the agnostic learner runs in $\poly(n)$ time and outputs a hypothesis with error $\Theta(\beta^*)$.
\end{remark}

\subsection{PAC-Verifiers for Circuit Classes}
\label{sec:proof_aip_ac0}

\begin{proof}[\proofof{\cref{thm:agnostic_ip_ac0}}]
For the agnostic interactive proof, we first state the following randomised query construction algorithm $\Hat{\sA}_{\mathsf{query}}$ that takes a string $w \in \{0,1\}^n$ as input and outputs a set of queries for $\mathsf{Amp}_k^f$, that has $w$ ``embedded" at a uniformly random query index. $\Hat{\sA}_{\mathsf{query}}$ builds on the query construction step in Algorithm \ref{algo:iw01_reconstruct}. Following this, we present the agnostic interactive proof for $\cAC^0[2]$ in Protocol \ref{algo:aip_ac0}, and finally prove its correctness.

\begin{algorithm}[!h]
\caption{Embedded NW-query construction algorithm $\Hat{\sA}_{\mathsf{query}}$ for $\mathsf{Amp}_k^{f}$, where $f$ is some $n$-variate Boolean function.}
\label{algo:embed_query_iw01}
    \vspace{0.1in}
    \nonl \textbf{Input:} A string $w \in \{0,1\}^n$. \\
    \vspace{0.1in}
    \textbf{The Algorithm:} \\
    \begin{minipage}[t]{15.75cm}        
    Let $S_1, \dots, S_L$ be an $(nk+k, L, \log L)$-NW set design over the universe $[m]$, where $m = O((nk)^2)$. Each $S_i$ is of size $nk+k$ that represents a query to $\mathsf{Amp}_k^f$, that in turn can be computed by making $k$ queries to $f$ on the first $nk$ bits. Thus, each $S_i$ can be partitioned into disjoint sets $S_{i1}, \dots, S_{ik}$ of size $n$ and the set $b(S_i)$ of size $k$. By the promise of the set design, for any $i \neq j$, and any $1 \leq k' \leq k, \vert S_{ik'} \cap S_j \vert \leq \log L$ (since $\vert S_i \cap S_j \vert \leq \log L$).

    \begin{enumerate}
        \item Pick $i$ uniformly at random from $[L]$.
        \item Let $Q = \sum_{u=1}^{i-1} \sum_{v=1}^k 2^{\vert S_{uv} \cap S_i \vert}$. Pick index $(u^*,v^*)$ with probability $\frac{2^{\vert S_{u^*v^*} \cap S_i \vert}}{Q}$, where $u^* \leq i-1$ and $v^* \leq k$.
        \item For each $1 \leq j \leq m$, with $j \notin S_i \cup S_{u^*v^*}$, pick a random bit $z_j$. For any $j \in S_{u^*v^*} \setminus S_i$, set $z_j$ to correspond with the respective bits from the input $w$. 
        \item For each $1 \leq u < i$ and $1 \leq v \leq k$, generate all the $2^{\vert S_{uv} \cap S_i \vert}$ strings that are consistent with the partial assignment to $z \vert_{S_{uv}}$ made in the Step 3. Call each such set of strings as $\hat{Q}_{uv}$.
        
        Let $\hat{Q} = \bigcup_{u,v} \hat{Q}_{uv} \subset \{0,1\}^n$ be the set of all queries to $f$.
        
        \item For each $1 \leq u < i$, generate all the strings in $b(S_u)$ that are consistent with the partial assignment to $z \vert_{b(S_u) \setminus S_i}$ in Step 3. Call these strings as $\hat{\mathcal{Q}}_{b(S_u)}$. 
        
        Let $\hat{P} = \bigcup_{u} \left( \hat{Q}_{u1} \times \dots \times \hat{Q}_{uk} \times \hat{\mathcal{Q}}_{b(S_u)} \right)$ be the set of all queries to $\mathsf{Amp}^f_k$.
    \end{enumerate}    
    Repeat these steps $\poly(L)$ many times using fresh randomness (similar to the final step of  Algorithm \ref{algo:iw01_reconstruct}). Finally, output the union $\sQ \subset \{0,1\}^n$ of all the $\hat{Q}$'s that are queries to $f$, as well as the union $\mathcal{P} \subset \{0,1\}^{nk+k}$ of all the $\hat{P}$'s that are queries to $\mathsf{Amp}^f_k$.
    \end{minipage}
\end{algorithm}

\begin{protocol}[!htp]
\caption{Agnostic interactive proof for $\cAC^0[2]$.}
\label{algo:aip_ac0}
    \nonl \textbf{Function input.} For function $f:\{0,1\}^n \rightarrow \{0,1\}$, the verifier has access to a random example oracle $\mathsf{Ex}(U_n,f)$ that returns samples from $(U_n, f(U_n))$. The prover has query access to $f$. \\
    \vspace{0.1in}
    \textbf{Explicit inputs.} $\delta > 0$ and an encoding (as a circuit) of the tolerant natural property $\sR$ that is useful against $\cAC^0[2]$-circuits of depth $d$ and size $u(n) = \mathsf{exp}(\Omega(n^{1/2d}))$ with tolerance $\tau(n) = 1/n^3$. \\
    \vspace{0.1in}
    \textbf{Promise on the functional input.} $\dist(f,\cAC^0[2]) = \beta^*$, for some $1/n^{\frac{\log(n)}{2}} < \beta^* < 1/2$. 
    \vspace{0.1in}
    \textbf{The Interactive Proof:}
    \begin{minipage}[t]{15.75cm}
    
    \begin{itemize}
        \item Let $M = \log^2(n)/2$. For each $r$, where $2 \leq r \leq M$, the interactive proof does the following. 
        \begin{enumerate}
            \item Set $\beta_r(n) = 1/2^r$. Let $G^{\mathsf{Amp}^f_k}$ be the NW-generator with stretch $L = 2^{O(\log^{4d+2} n)}$ and the amplification parameter $k = \tau(\ell)/4\beta_r$, where $\ell = \log (L)$. Note that $\tau(\ell) = 1/\ell^3 = 1/O(\log^{12d+6}(n))$. Finally, let $q_{\max}$ be the worst-case upper bound on the query complexity of the learning algorithm $\sB$ from Theorem \ref{thm:cikk_ac0[2]_learn}, for any value of $\beta_r$.

            \item The verifier draws fresh labeled samples $\{(x_{r1}, f(x_{r1})), \dots, (x_{rt}, f(x_{rt})) \}$ from the random example oracle distributed according to $(U_n, f(U_n))$, where $t$ is set to $40/\delta \cdot q_{\max} \cdot \ln(M/\delta)$. 

            \item For each $x_{rj}$ such that $j \in [t]$, the verifier does the following.
            \begin{itemize}
                \item Runs $\hat{\sA}_{\mathsf{query}}(x_{rj})$ to get a set of queries $\sQ_{rj} \subset \{0,1\}^n$ that has $x_{rj}$ embedded at a randomly chosen point and another set of queries $\mathcal{P}_{rj} \subset \{0,1\}^{nk+k}$.
                \item Sends $\sQ_{rj}$ to the Prover and receives the set $\Tilde{f}(\sQ_{rj})$ consisting answers $\{\Tilde{f}(w) \mid w \in \sQ_{rj}\}$. 
                \item Next, the verifier checks if $f(x_{rj}) = \Tilde{f}(x_{rj})$. If yes, the verifier runs the rest of the learner $\sB$ (or $\hat{\sB}$) using the natural property $\sR$, the query set $\mathcal{P}_{rj}$ and its answers on $\mathsf{Amp}^f_k$ computed using $\Tilde{f}(\sQ_{rj})$. Let $h^{rj} : \{0,1\}^n \rightarrow \{0,1\}$ be the output of the learner. If not, the verifier rejects and outputs `$\bot$'. 
            \end{itemize}

            \item The verifier picks a hypothesis $h^r$ uniformly at random from the set $\{h^{rj}\}_{j \in [t]}$ and adds it to the set of candidate hypothesis $H_{\mathsf{final}}$.
        \end{enumerate}

    \item The verifier picks $W_\mathsf{test} = O \left( \frac{1}{\ell^4 \beta_M^2} \cdot \ln \left( \frac{2M}{\delta} \right) \right) = n^{O(\log(n))} \log \left( \frac{2M}{\delta} \right)$ fresh labeled samples from the random example oracle. Next, for each candidate hypothesis circuit $h^r$ in $H_{\mathsf{final}}$, the verifier estimates the distance between $f$ and $h^r$ using these $W_{\mathsf{test}}$ samples, i.e., number of samples $x_j$, for $j \in [W_{\mathsf{test}}]$, such that $f(x_j) \neq h^r(x_j)$. 
    
    Finally, the verifier outputs the hypothesis $h^* : \{0,1\}^n \rightarrow \{0,1\}$ in $H_{\mathsf{final}}$ with the smallest estimated distance amongst all the candidate hypotheses.
    \end{itemize}
    \end{minipage}
\end{protocol}

\paragraph*{Completeness.} Firstly, observe that for every $\beta_r$ and every $j \in [t]$, the honest prover answers each query in each set $\sQ_{rj}$ correctly, i.e., $\Tilde{f}(\sQ_{rj}) = f(\sQ_{rj})$. In particular, for each $x_j$, the verifier gets the correct value $f(x_j)$ and thus, never rejects any query set. We first state the following lemma about the size of the query sets generated by the algorithm.

Now, suppose that for some $\mathcal{T}$, the true distance parameter $\beta^*$ lies in the interval between $\left( 1/2^{\mathcal{T}+1}, 1/2^{\mathcal{T}} \right]$. Then we have,
\begin{lemma}
    \label{lem:exists_good_hypothesis}
    With probability at least $1-\delta/2$, the hypothesis $h^{\mathcal{T}} \in H_{\mathsf{final}}$ has error at most $2 \ell^4 \beta^*$, where the probability is over the internal randomness of the verifier and samples taken from the random example oracle in the $\mathcal{T}^{\text{th}}$ round (i.e., when the generator is instantiated using distance parameter $1/2^{\mathcal{T}}$.
\end{lemma}

\begin{proof}
    Let $\sB$ be the quasi-polynomial time agnostic learner using membership queries for $\cAC^0[2]$ over the uniform distribution from Theorem \ref{thm:cikk_ac0[2]_learn}. Moreover, let $\hat{\sB}$ be the learner obtained by sampling a uniformly random $w \in \{0,1\}^n$ and running $\hat{\sA}_{\mathsf{query}}(w)$ to generate the membership queries made to $\mathsf{Amp}^f_k$ in $\sB$, instead of $\sA_{\mathsf{query}}$. First, we observe that $\hat{\sB}$ also satisfies the guarantees of the learner $\sB$, because $w$ itself is now being used as a source of randomness in the query construction.
    
    From Corollary \ref{cor:query_max_size}, since every $\beta_r \geq 1/n^{\log(n)}$, we see that when $\ell = O(\log^{2d+1}(n \cdot n^{\log(n)}) = O(\log^{4d+2}(n)$ , for every $\beta_r \geq \beta^*$, $\hat{\sB}$ outputs a hypothesis whose error is at most $\ell^4 \beta_r$. In other words, for every $j \in [t]$, given the correct answers to the queries in the set $\sQ_{\mathcal{T}j}$, it outputs a hypothesis $h^{\mathcal{T}j}$ that errs on at most $\ell^4 \beta_\mathcal{T}$ fraction of the inputs, with probability at least $1-\delta/20$ over its internal randomness. 
    
    From our choice of $t$ that is larger than $O(1/\delta^2 \log(1/\delta))$, using a Hoeffding concentration bound, at least $(1-\delta/5)t$ of the hypotheses in $\{h_{\mathcal{T}j}\}_{j \in [t]}$ have error at most $\ell^4 \beta_{\mathcal{T}}$, with probability all but $\delta/4$. In other words, with probability at least $1-\delta/5$, a hypothesis that is chosen uniformly at random from the set $\{h_{\mathcal{T},j}\}_{j \in [t]}$ has error at most $\ell^4 \beta_\mathcal{T}$. Thus, the total failure probability that $h_\mathcal{T}$ is not a good hypothesis is at most $\delta/4 + \delta/5 < \delta/2$. Combined with the fact that $\beta_{\mathcal{T}} \leq 2\beta^*$, $h^{\mathcal{T}}$ has error at most $2 \ell^4 \beta^*$ and the claim follows.
\end{proof}
To summarise, using Lemma \ref{lem:exists_good_hypothesis}, we have established that with probability at least $1-\delta/2$, there exists a hypothesis in $H_{\mathsf{final}}$ that has error at most $2 \ell^4 \beta^*$.

Finally, we show that with probability at least $1-\delta/2$ over the choice of $W_{\mathsf{test}}$ fresh samples in the final step, either $h^\mathcal{T}$ or another hypothesis with error at most $4 \ell^4 \beta^*$ is output by the verifier. Overall, this shows that our agnostic interactive proof satisfies the completeness requirement, by showing that with probability at least $1-\delta$ over its internal randomness and the samples received during the interactive proof, the verifier outputs $h^*$ that has error at most $4 \ell^4 \beta^*$. 

For this, observe that the verifier outputs $h^*$ such that $\dist(h^*, f) > 4 \ell^4 \beta^*$, when either $h^{\mathcal{T}}$ has a large error estimate on the $W_{\mathsf{test}}$ samples, or there exists a hypothesis $h^r \in H_{\mathsf{final}}$ such that $\dist(f,h^r) > 4 \ell^4 \beta$ but still has a smaller error estimate on the $W_{\mathsf{test}}$ samples than $h^\mathcal{T}$.

In more detail, using Hoeffding's concentration bound, the estimated error between $h^\mathcal{T}$ and $f$ is larger than $3 \ell^4 \beta^*$, only with probability at most $\delta/2M$, since our choice of $W_{\mathsf{test}}$ ensures that the estimate is made over at least $O \left( \frac{1}{\ell^4 \beta_\mathcal{T}^2} \cdot \ln \left( \frac{2M}{\delta} \right) \right)$ many samples. Next, for any hypothesis $h^r$ and $f$ such that $\dist(f,h^r) > 4 \ell^4 \beta^*$, again using Hoeffding's bound, the estimated error of $h^r$ over the $W_{\mathsf{test}}$ samples is less than $3 \ell^4 \beta^*$, with probability at most $\delta/2M$. Since the number of such hypotheses is at most $M-1$, using a union bound, we see that the probability that the verifier outputs $h^*$ such that $\dist(h^*, f) > 4 \ell^4 \beta^*$ is at most $\delta/2M + (M-1) \cdot \delta/2M \leq \delta/2$. 

\paragraph*{Soundness.} Recall that the generator $G^{\mathsf{Amp}^f_k}$ is constructed using an $(nk+k, L, \log (L))$-NW set design $S_1, \dots, S_L$, over the universe $[m]$. Further, each $S_i$ can be partitioned into disjoint sets $S_{i1}, \dots, S_{ik}$ of size $n$ and $B_i$ of size $k$ (since the queries are made to $\mathsf{Amp}^f_k$). By the promise of the set design, for any $i \neq j$, and any $1 \leq k' \leq k, \vert b(S_i) \cap S_j \vert \leq \log(L)$ and moreover, $\vert S_{ik'} \cap S_j \vert \leq \log (L)$ (since $\vert S_i \cap S_j \vert \leq \log (L)$).

Consider a run of the query construction algorithm $\sA_{\mathsf{query}}$ from Algorithm \ref{algo:iw01_reconstruct} for $\mathsf{Amp}^f_k$. An important observation here is that for $i$ picked uniformly randomly from $[L]$, $\sA_{\mathsf{query}}$ generates random (possibly non-disjoint) subcubes $\sE_{1,1}, \dots, \sE_{1,k}, \dots, \sE_{(i-1),1}, \dots, \sE_{(i-1),k}$ and only makes subcube membership queries over them (i.e., queries each subcube completely in lexicographic order). Indeed, each $\sE_{u,v}$ is the subcube obtained by generating all strings of length $n$ that are consistent with the partial assignment obtained from the random bits set in Step 2. Clearly, $\sA_{\mathsf{query}}$ generates non-adaptive queries to $\mathsf{Amp}^f_k$, which in turn also results in non-adaptive queries to $f$. 

We first prove the following result that characterizes the marginal distribution of the $p^{\text{th}}$ query to $f$ generated by $\sA_{\mathsf{query}}$.

\begin{lemma}[The reconstruction query distribution lemma]
\label{lem:nw_query_marginal}
Suppose that $\sA_{\mathsf{query}}$ makes membership queries over subcubes $\sE_{1,1}, \dots, \sE_{(i-1),k}$ to $\mathsf{Amp}^f_k$, for some $i \leq L$, using an $(nk+k, L, \log(L))$-set design $(S_1, \dots, S_L)$. 

Let $\sQ$ be the query set for $f$ generated by $\sA_{\mathsf{query}}$ and $\sum_{u=1}^{i-1} \sum_{v=1}^k 2^{\vert S_{u,v} \cap S_i \vert}$ be its size. Then, for any query index $p \in \left[ \sum_{u=1}^{i-1} \sum_{v=1}^k 2^{\vert S_{u,v} \cap S_i \vert} \right]$, where $p \in \sE_{u,v}$ for some $u \leq i-1$ and $v \leq k$, and for any $w \in \{0,1\}^n$, the marginal distribution of $\sQ[p]$, the $p^{\text{th}}$-query, satisfies one of the two properties.
\begin{itemize}
    \item $\Pr_{\sA_{\mathsf{query}}}[\sQ[p] = w] = 0$, or
    \item $\Pr_{\sA_{\mathsf{query}}}[\sQ[p] = w] = \frac{1}{2^{n - \vert S_{u,v} \cap S_i \vert}}$.
\end{itemize}
\end{lemma}

\textbf{Notation for the proof of Lemma \ref{lem:nw_query_marginal}:} For any variables $x_1, \dots, x_n$, let $\rho : [n] \rightarrow \{0,1,*\}$ be a \textit{restriction} over the variables obtained by setting $x_i = \rho(i), \forall i \in [n]$, where ``$\rho(i)=*$" denotes that $x_i$ is left free (i.e., unset). 

For any restriction $\rho$ and subcube $\sE \subseteq \{0,1\}^{n}$, we say that $\sE$ is generated by the restriction $\rho$, if it is the subcube given by the set of all $n$-length binary strings consistent with $\rho$ on the indices in $\rho^{-1}(\{0,1\})$. Note that the size of $\vert \sE \vert = 2^{\vert \rho^{-1} (*) \vert}$. For notational simplicity, for the rest of the proof we index the subcubes queried over by $\sA_{\mathsf{query}}$, with a \textit{single value} from $1 \leq j \leq N$ in a natural way, where $N = (i-1)k$. 

Let $z$ be the random bits used by $\sA_{\mathsf{query}}$ to generate the queries (i.e., the seed of the NW-generator). For any $1 \leq j \leq N$, define $\widehat{S_{j}}$ as the set $S_{j} \setminus S_i$. For any $1 \leq j \leq N$, the subcube $\sE_{j}$ is generated by the \textit{random restriction} $\rho_{j}: [n] \rightarrow \{0,1,*\}$, such that $\rho_{j}[a] = z_a$, if $a\in \widehat{S_{j}}$ and $\rho_{j}[a] = *$, if $a \in S_{j} \cap S_i$. In other words, the random bits in $\rho_{j}^{-1}(\{0,1\})$ generate a random subcube $\sE_{j}$. From the properties of the set design, we see that $\vert \rho_{j}^{-1}(*) \vert \leq \log (L)$ (in the degenerate case, $\vert \rho_{j}^{-1}(*) \vert = 0$ and the size of $\sE_{j}$ is just $1$). 

For any restriction $\rho$ and any string  $r \in \{0,1\}^{n- \vert \rho^{-1}(*) \vert}$, we abuse notation by using $\rho = r$ to denote the variables in $\rho^{-1}(\{0,1\})$ being set to $r$ in the natural way. For every $1 \leq j \leq N$, denote by $\Vec{\rho}^{(j)}$ as the tuple of restrictions $(\rho_j, \dots, \rho_{N})$. In particular, the tuple of all restrictions $(\rho_1, \dots, \rho_N)$ is given by $\Vec{\rho}^{(1)}$. For any tuple of strings $\Vec{r}^{(j)} = (r_j, \dots, r_{N})$, where each $r_a \in \{0,1\}^{n-\vert \rho_a^{-1}(*) \vert}$, we say that $\Vec{\rho}^{(j)} = \Vec{r}^{(j)}$, if $\rho_a = r_a$ for every $j \leq a \leq N$. 

In a similar vein, for every $1 \leq j \leq i-1$, let $\sigma_j : [k] \rightarrow \{0,1,*\}$ be random restriction generated by projecting $z$ over the set $B_j$ ($\sigma_j$ is set to `$*$' on every element in $B_j \setminus S_i$). It is worth recalling that the membership queries to $f$ are only made on the subcubes induced by $\rho_1, \dots, \rho_N$. For any tuple of strings $\Vec{u}^{(j)} = (u_j, \dots, u_{i-1})$, where each $u_a \in \{0,1\}^{k - \vert \sigma_a^{-1}(*) \vert}$, we say that $\Vec{\sigma}^{(j)} = \Vec{u}^{(j)}$, if $\sigma_a = r_a$ for every $j \leq a \leq i-1$. 

\begin{proof}[Proof of Lemma \ref{lem:nw_query_marginal}]
Let $\sE_{j^*} \in [N]$ be the subcube that contains the $p^{\text{th}}$ query and let $\sE_{j^*}$ be generated by the restriction $\rho_{j^*} : [n] \rightarrow \{0,1*\}$ (this is the same as $\sE_{uv}$ in the other scheme of indexing).

Recall that for each $1 \leq j \leq N$, $\sA_{\mathsf{query}}$ only makes subcube queries, by going over all possible assignments to $\rho_{j}^{-1}(*)$ in lexicographic order. Suppose that the variables in $\rho_{j^*}^{-1}(*)$ are fixed with the string $\alpha \in \{0,1\}^{\vert \rho_{j^*}^{-1}(*) \vert}$, corresponding to the location of $p$ within the lexicographic ordering of the elements in $\sE_{j^*}$. 

Now, if $w$ restricted to $\rho_{j^*}^{-1}(*)$ is not $\alpha$, then the $p^{\text{th}}$ query can never be $w$, i.e., $\Pr_{\sA_{\mathsf{query}}}[\sQ[p] = w] = 0$. Thus, for the rest of this proof, we assume that $w$ restricted to $\rho_{j^*}^{-1}(*)$ is $\alpha$, and refer to $w$ with the tuple $\left( w \vert_{\hat{S}_{j^*}}, \alpha \right)$, since $\rho_{j^*}^{-1}(*)$ contains the indices in $\hat{S_{j^*}} = S_{j^*} \cap S_i$. 

Let $\sQ^{-p}$ be the set of queries $\sQ \setminus \sQ[p]$. We have
\begin{equation*}
    \begin{split}    
    \Pr_{\sA_{\mathsf{query}}} [\sQ[p] = w] = \sum_{\mathcal{Y} \in \left( \{0,1\}^n \right)^{Q-1}} &\Pr_{z} [\sQ[p] = w \text{ and } \sQ^{-p} = \mathcal{Y}] \\
    = \sum_{\vec{r}^{(1)},\vec{u}^{(1)}} \Pr_{z} &\left[ w = \left( z \vert_{\hat{S}_{j^*}}, \alpha \right) \big\vert \vec{\rho}^{(1)} = \vec{r}^{(1)} \land \Vec{\sigma}^{(1)} = \vec{u}^{(1)} \right]
    \cdot \Pr_{z} \left[ \vec{\rho}^{(1)} = \vec{r}^{(1)} \land \Vec{\sigma}^{(1)} = \vec{u}^{(1)} \right] \\
    = \sum_{\vec{r}^{(1)},\vec{u}^{(1)}} \Pr_{z} &\left[ w = \left( z \vert_{\hat{S}_{j^*}}, \alpha \right) \big\vert \vec{\rho}^{(1)} = \vec{r}^{(1)} \land \Vec{\sigma}^{(1)} = \vec{u}^{(1)} \right] \\
    &\cdot \Pr_{z} \left[ \vec{\rho}^{(2)} = \vec{r}^{(2)} \land \Vec{\sigma}^{(1)} = \vec{u}^{(1)} \big\vert \rho_1 = r_1 \right]
    \cdot \Pr_{z} \left[ \rho_1 = r_1 \right] \\
    = \sum_{r_1 \in \{0,1\}^{n-m_1}} &\frac{1}{2^{n-m_1}} \cdot \sum_{\vec{r}^{(2)},\vec{u}^{(1)}} \Pr_{z} \left[ w = \left( z \vert_{\hat{S}_{j^*}}, \alpha \right) \big\vert \vec{\rho}^{(1)} = \vec{r}^{(1)} \land \Vec{\sigma}^{(1)} = \vec{u}^{(1)} \right] \\
    &\cdot \Pr_{z} \left[ \vec{\rho}^{(2)} = \vec{r}^{(2)} \land \Vec{\sigma}^{(1)} = \vec{u}^{(1)} \big\vert \rho_1 = r_1 \right] \\
    \end{split}
\end{equation*}
The second expression holds because the query algorithm only makes subcube queries, each of which is represented by a restrictions $\rho_1, \dots, \rho_N$, and is possibly dependent on $\sigma_1, \dots, \sigma_{i-1}$ by virtue of the set design. On the other hand, the final expression comes by noting that the restriction $\rho_1$ is the random string $z \vert_{\hat{S_1}}$ by definition, and thus, $m_1 = \vert S_1 \cap S_i \vert$.

Next, observe that by setting $\rho_1$ to $r_1$, in turn we fix $z \vert_{\hat{S_1}}$ to be $r_1$. Let $\mathsf{cons}(r_1)$ be the set of strings $r_2$ of length $n - \vert \hat{S_2} \vert$ such that when $\rho_1 = r_1$, $r_2$ is \textit{consistent} with the values in $z$ fixed previously by $r_1$ with respect to the underlying set design. In other words, the projection of $r_2$ onto the indices in $\widehat{S_1 \cap S_2} = (S_1 \cap S_2) \setminus S_i$ is equal to $z \vert_{\widehat{S_1 \cap S_2}} $. 

Let $m_2 = \vert S_2 \cap S_i \vert + \vert \widehat{S_1 \cap S_2} \vert$. We have
\begin{equation*}
    \begin{split}
        \sum_{r_1 \in \{0,1\}^{n-m_1}} \frac{1}{2^{n-m_1}} \cdot \sum_{\vec{r}^{(2)},\vec{u}^{(1)}} &\Pr_{z} \left[ w = \left( z \vert_{\hat{S}_{j^*}}, \alpha \right) \big\vert \vec{\rho}^{(1)} = \vec{r}^{(1)} \land \Vec{\sigma}^{(1)} = \vec{u}^{(1)} \right] \\
        &\cdot \Pr_{z} \left[ \vec{\rho}^{(2)} = \vec{r}^{(2)} \land \Vec{\sigma}^{(1)} = \vec{u}^{(1)} \big\vert \rho_1 = r_1 \right] \\
        = \sum_{r_1 \in \{0,1\}^{n-m_1}} \frac{1}{2^{n-m_1}} \cdot \sum_{\vec{r}^{(2)},\vec{u}^{(1)}} &\Pr_{z} \left[ w = \left( z \vert_{\hat{S}_{j^*}}, \alpha \right) \big\vert \vec{\rho}^{(1)} = \vec{r}^{(1)} \land \Vec{\sigma}^{(1)} = \vec{u}^{(1)} \right] \\
        &\cdot \Pr_{z} \left[ \vec{\rho}^{(3)} = \vec{r}^{(3)} \land \Vec{\sigma}^{(1)} = \vec{u}^{(1)} \big\vert \rho_1 = r_1 \land \rho_2 = r_2 \right] \cdot \Pr_{z} \left[ \rho_2 = r_2 \vert \rho_1 = r_1 \right] \\
    \end{split}
\end{equation*}

\begin{equation*}
\begin{split}
    = \sum_{r_1 \in \{0,1\}^{n-m_1}} \frac{1}{2^{n-m_1}} \cdot \sum_{r_2 \in \mathsf{cons}(r_1)} \frac{1}{2^{n-m_2}} \cdot \sum_{\vec{r}^{(3)}} &\Pr_{z} \left[ w = \left( z \vert_{\hat{S}_{j^*}}, \alpha \right) \big\vert \vec{\rho}^{(1)} = \vec{r}^{(1)} \land \Vec{\sigma}^{(1)} = \vec{u}^{(1)} \right] \\
    &\cdot \Pr_{z} \left[ \vec{\rho}^{(3)} = \vec{r}^{(3)} \land \Vec{\sigma}^{(1)} = \vec{u}^{(1)} \big\vert \rho_1 = r_1 \land \rho_2 = r_2 \right] \\
\end{split}
\end{equation*}
The last expression follows from the fact that for any $r_2 \in \mathsf{cons}(r_1)$, $\Pr_{z} \left[ \rho_2 = r_2 \vert \rho_1 = r_1 \right]$ is just the probability that the random bits in $z \vert_{\hat{S_2}} = r_2$, given that the $z \vert_{\hat{S_1} \cap \hat{S_2}}$ is already fixed using the corresponding bits from $r_1$. Given this, we see that $\Pr_{z} \left[ \rho_2 = r_2 \vert \rho_1 = r_1 \right] = 1/2^{n-m_2}$.\footnote{Note that since $S_1, \dots, S_k$ are all mutually disjoint, none of the $r_j$'s have any bits fixed by $r_1, \dots, r_{j-1}$. The same argument holds for $B_1$ with respect to $S_1, \dots, S_k$ too. We maintain full generality in the expressions above for ease of exposition. Indeed, for any $0 \leq a \leq i-2$ any string amongst $r_{ak+1}, \dots, r_{(a+1)k)}$ may have some bits fixed by any of the previous $r_j$'s or the $u_j$'s.}

We proceed in this iterative fashion, where for any $j \leq N$, we look at the effect of the restrictions $\rho_1, \dots, \rho_k, \sigma_1, \rho_{k+1}, \dots, \rho_{2k}, \sigma_2, \rho_{2k+1}, \dots, \rho_{\lfloor j/k \rfloor \cdot k}, \sigma_{\lfloor j/k \rfloor}, \rho_{\lfloor j/k \rfloor \cdot k + 1}, \dots, \rho_{j-1}$ on the $p^{\text{th}}$ query. The reason for this comes from the nature of the set design for  $\mathsf{Amp}^f_k$ - any set amongst $S_{ak+1} \dots, S_{(a+1)k}, B_a$ could have an overlap with any of the sets $S_1, \dots, S_{ak}$ or $B_1, \dots, B_{a-1}$.

Thus, we define for any $j$, $\mathsf{cons}(r_1, \dots, r_{j-1})$ as the set of strings $r_j$ such that the bits fixed in $z$ by $r_1, \dots, r_{j-1}$ and $u_1, \dots, u_{\lfloor j/k \rfloor}$ are consistent with the projection of $r_j$ onto the indices in $\hat{S}_j \cap \left( \left( \bigcup_{a=1}^{j-1} \hat{S}_a \right) \cup \left( \bigcup_{a=1}^{\lfloor j/k \rfloor} \hat{B}_a \right) \right)$. 

Moreover, define $m_j = \vert S_j \cap S_i \vert + \bigg\vert \left( \bigcup_{a=1}^{j-1} \hat{S}_a \right) \cap \hat{S}_j \bigg\vert + \bigg\vert \left( \bigcup_{a=1}^{\lfloor j/k \rfloor} \hat{B}_a \right) \cap \hat{S}_j \bigg\vert$, as the number of bits in $\hat{S_j}$ that overlap with the indices previously fixed in each $\hat{S}_a$. Similarly, we define $m'_j = \vert B_j \cap S_i \vert + \bigg\vert \left( \bigcup_{a=1}^{jk} \hat{S}_a \right) \cap \hat{B}_j \bigg\vert + \bigg\vert \left( \bigcup_{a=1}^{j-1} \hat{B}_a \right) \cap \hat{B}_j \bigg\vert$, as the number bits in $B_j$ that overlap with the indices previously fixed in each $\hat{B}_a$. Of course, in the scenario that all the bits in $z \vert_{\hat{S}_{j}}$ are fixed by $\rho_1, \dots, \rho_{j-1}$ and $\sigma_1, \dots, \sigma_{\lfloor j/k \rfloor}$, then $m_j = n$, and $\mathsf{cons}(r_1, \dots, r_{j-1})$ contains only one string. 

We repeat the above calculations $N-1$ times and finally, count over all mutually consistent ways $z$ can be fixed by $r_1, \dots, r_{N-1}$ and $u_1, \dots, u_{i-1}$. We have
\begin{equation*}
    \begin{split}
        \Pr_{\sA_{\mathsf{query}}} [\sQ[p] = w] = \prod_{a=1}^{N-1} \left( \frac{1}{2^{n-m_a}} \right) \cdot &\prod_{a=1}^{i-1} \left( \frac{1}{2^{k-m'_a}} \right) \sum_{\substack{r_1 \in \{0,1\}^{n-m_1} \\ r_2 \in \mathsf{cons}(r_1) \\ \cdots \\ r_N \in \mathsf{cons}(r_1, \dots, r_{N-1}, \\ u_1, \dots, u_{i-1})}} \Pr_{z} \left[ w = \left( z \vert_{\hat{S}_{j^*}}, \alpha \right) \big\vert \vec{\rho}^{(1)} = \vec{r}^{(1)} \land \Vec{\sigma}^{(1)} = \vec{u}^{(1)} \right] \\
        = \prod_{a=1}^{N-1} \left( \frac{1}{2^{n-m_a}} \right) \cdot &\prod_{a=1}^{i-1} \left( \frac{1}{2^{k-m'_a}} \right) \cdot \sum_{\substack{r_1 \in \{0,1\}^{n-m_1} \\ r_2 \in \mathsf{cons}(r_1) \\ \cdots \\ r_N \in \mathsf{cons}(r_1, \dots, r_{N-1}, u_1, \dots, u_{i-1})}} \frac{1}{2^{\vert \hat{S}_{j^*} \vert}} \\
    \end{split}
\end{equation*}
The last line follows from the fact that for each $r_{j^*} \in \{0,1\}^{n-m_{j^*}}$, the entirety of $z \vert_{S_j^* \setminus S_i}$ is fixed by $r_1, \dots, r_{j^*}$ and the probability that $w = z \vert_{S_{j^*} \setminus S_i}$ is uniform over $\{0,1\}^{\vert {S_{j^*} \setminus S_i} \vert}$.

Since for each $j$, the size of the set $\mathsf{cons}(r_1, \dots, r_j)$ is $2^{n-m_j}$, we have that the probability the $p^{\text{th}}$-query is the string $w$ is just $\frac{1}{2^{\vert S_{j^*} \setminus S_i \vert}}$, which is nothing but $\frac{1}{2^{n-\vert S_{j^*} \cap S_i \vert}}$.
\end{proof}

We have the following observation about the restrictions on strings sampled from the uniform distribution. 
\begin{proposition}
\label{prop:uniform_dist_subset}
    For any restriction $\rho : [n] \rightarrow \{0,1,*\}$ and any $w \in \{0,1\}^n$ sampled from $\sU_n$, the marginal distribution of $w \vert_{\rho}$ is the uniform distribution over strings of length $n-\vert \rho^{-1}(*) \vert$.
\end{proposition}
    
We now prove that Protocol \ref{algo:aip_ac0} satisfies the soundness requirements. To this end, consider round $r$ and uniform random example $x_{rj} \in \{0,1\}^n$, such that $\hat{\sA}_{\mathsf{query}}$ constructs query set $\sQ_{rj}$ of size $\vert \sQ_{rj} \vert$, that has $x_{rj}$ embedded at a uniformly random query index. The malicious prover can look at the query distribution of any $q^{\text{th}}$-query in $\sQ_{rj}$, where $q$ corresponds to a query in some subcube $\sE_{u,v}$ generated by the restriction $\rho_{u,v}$ and try to infer if this was the example embedded by the verifier in $\sQ_{rj}$. 

From Lemma \ref{lem:nw_query_marginal} and Proposition \ref{prop:uniform_dist_subset}, for every query index $q$ that corresponds to a query in the subcube $\sE_{u,v}$, the marginal distribution of the $q^{\text{th}}$-query made by $\sA_{\mathsf{query}}$ (the true query construction algorithm from Theorem \ref{thm:cikk_ac0[2]_learn}) is the same as that of the uniform distribution (i.e., the distribution over the samples) over $n$-length strings restricted on to $\rho_{u,v}$. This means that the malicious prover cannot distinguish whether the $q^{\text{th}}$-query came from a subcube induced by a restriction obtained from $\sA_{\mathsf{query}}$'s internal randomness, or came from a subcube induced by a restriction obtained by embedding $x_{rj}$. Thus, the prover has no advantage in picking a particular query with respect to which it cheats.\footnote{By cheating, we mean that the malicious prover answers correctly only on the query it believes to be as the embedded random example and lies on any other  query which the verifier does not check.} In other words, the best the malicious prover can do is pick a subcube $\sE_{u,v}$ at random with probability $\vert \sE_{u,v} \vert/\vert \sQ_{rj} \vert$. 

Furthermore, any point in $\sE_{u,v}$ has equal likelihood of being the embedded random example and thus, the malicious prover has no advantage over randomly guessing an element of $\sE_{u,v}$. Put together, the malicious prover gains no advantage over picking a query uniformly at random from $[\vert \sQ_{rj} \vert]$ and cheating with respect to it.

For any $r \leq M$, define $\mathcal{W}_r$ as the number of query sets $\sQ_{r,j}$ over which the prover responds dishonestly on at least one query in the set, i.e., let
\begin{equation*}
\mathcal{W}_r = \{j \in [t] \mid \exists w \in \sQ_{rj} \text{ such that } \Tilde{f}(w) \neq f(w) \}    
\end{equation*}
\noindent We finish our analysis by considering the following 2 cases based on the value of $\mathcal{W}_r$.
\begin{itemize}
    \item \textbf{Case 1:} There exists $r$ such that $\vert \mathcal{W}_r \vert > \ln(M/\delta) q_{\max}$. Then, for the $r^{\text{th}}$ round, we have
    \begin{equation*}
    \begin{split}
        \Pr \left[ \forall \sQ_{r,j} \in \mathcal{W}_r, \text{verifier accepts } \Tilde{f}(\sQ_{r,j}) \right] &= \Pi_{j \in \mathcal{W}_r} \left( 1 - \frac{1}{\vert \sQ_{rj} \vert} \right) \\
        &\leq \left( 1 - \frac{1}{q_{\max}} \right)^{\vert \mathcal{W}_r \vert} \leq e^{-\vert \mathcal{W}_r \vert/q_{\max}} \leq \delta/M \\
    \end{split}
    \end{equation*}
    The first inequality follows from Corollary \ref{cor:query_max_size} as $\vert \sQ_{rj} \vert \leq q_{\max}$, for all $r \leq M$ and $j \in [t]$. Using a union bound over the $M$ rounds, we see that with probability $1-\delta$, any prover gets caught by the verifier if it lies on at least $\ln(M/\delta) q_{\max}$ of the query sets in some round.
    
    \item \textbf{Case 2:} For every $r \leq M$, $\vert \mathcal{W}_r \vert \leq \ln(M/\delta) q_{\max}$. Again, let $\mathcal{T}$ be the value such that $\beta^*$ lies in the interval $\left( 1/2^{\mathcal{T}+1}, 1/2^{\mathcal{T}} \right]$. In particular, $\vert \mathcal{W}_{\mathcal{T}} \vert \leq \ln(M/\delta) q_{\max}$ and the fraction of query sets over which the prover cheats is at most $\vert \mathcal{W}_\mathcal{T} \vert / t \leq \delta/40$.

    For each query set $\sQ_{\mathcal{T}j}$ such that $j \in [t] \setminus \mathcal{W}_\mathcal{T}$, the prover answers correctly on every query in the set, and in turn the verifier computes a hypothesis $h^{\mathcal{T}j}$ that errs on at most $\ell^4 \beta_\mathcal{T}$ fraction of the inputs, with probability at least $1-\delta/20$ over its internal randomness. Let $\mathcal{Y}$ be the number of hypotheses in $\{h_{\mathcal{T}j}\}_{j \in [t] \setminus \mathcal{W}_\mathcal{T}}$ that have error at most $\ell^4 \beta_{\mathcal{T}}$ (i.e., are ``good" hypotheses). Again, since the number of faithful query sets (i.e., $t - \vert \mathcal{W}_\mathcal{T} \vert$) is larger than $O(1/\delta^2 \log(1/\delta))$, using a Hoeffding concentration bound, the event $\mathcal{Y} > (1-\delta/5)(t - \vert \mathcal{W}_\mathcal{T} \vert)$ occurs with probability all but $\delta/4$.
    
    In such a case, the probability that a uniformly randomly chosen hypothesis $h_{\mathcal{T},j^*}$ from the set $\{h_{\mathcal{T},j}\}_{j \in [t]}$ is good is given by
    \begin{equation*}
        \begin{split}
            \Pr_{\stackrel{j^* \sim [t]}{V}} [h_{\mathcal{T}j^*}\text{ is bad}] \leq &\Pr_{V} [h_{\mathcal{T}j^*} \text{ is bad} \mid j^* \in \mathcal{W}_\mathcal{T}] \cdot \Pr_{j^* \sim [t]}[j^* \in \mathcal{W}_\mathcal{T}] + \Pr_{V} [h_{\mathcal{T}j^*} \text{ is bad} \mid j^* \notin \mathcal{W}_\mathcal{T}] \\
            \leq &\Pr_{V} [h_{\mathcal{T}j^*} \text{ is bad} \mid j^* \in \mathcal{W}_\mathcal{T}] \cdot \Pr_{j^* \sim [t]}[j^* \in \mathcal{W}_\mathcal{T}] \\
            &+ \Pr_{V} [h_{\mathcal{T}j^*} \text{ is bad} \mid j^* \notin \mathcal{W}_\mathcal{T} \land \mathcal{Y} > (1-\delta/5)(t - \vert \mathcal{W}_r \vert)] \cdot \Pr_{V} [\mathcal{Y} > (1-\delta/5)(t - \vert \mathcal{W}_r \vert)] \\
            &+ \Pr_{V} [h_{\mathcal{T}j^*} \text{ is bad} \mid j^* \notin \mathcal{W}_\mathcal{T} \land \mathcal{Y} \leq (1-\delta/5)(t - \vert \mathcal{W}_r \vert)] \cdot \Pr_{V} [\mathcal{Y} \leq (1-\delta/5)(t - \vert \mathcal{W}_r \vert)] \\
            \leq &\delta/40 + \delta/5 + \delta/4 \\
            < &\delta/2
        \end{split}
    \end{equation*}
    The third inequality comes from the fact that if $\mathcal{Y} > (1-\delta/5)(t - \vert \mathcal{W}_\mathcal{T} \vert)$, then the number of bad hypotheses in $\{h_{\mathcal{T}j}\}_{j \in [t] \setminus \mathcal{W}_\mathcal{T}}$ is given by the random variable $t - \mathcal{W}_{\mathcal{T}} - \mathcal{Y} \leq (t - \mathcal{W}_{\mathcal{T}})\delta/5$, and that the probability of the event $\mathcal{Y} \leq (1-\delta/5)(t - \vert \mathcal{W}_r \vert)$ is at most $\delta/4$.
    
    Combined with the fact that $\beta_{\mathcal{T}} \leq 2\beta^*$,  with probability at least $1-\delta/2$, there exists a hypothesis in $H_{\mathsf{final}}$ that has error at most $2 \ell^4 \beta^*$. The rest of the analysis is identical to that of the completeness case, and we conclude that the verifier outputs a hypothesis that has error at most $4 \ell \beta^*$, with probability $1-\delta$ over its internal randomness.

    \paragraph*{Complexity of the interactive proof.} From Corollary \ref{cor:query_max_size}, we see that $q_{\max} = \poly(L) \leq 2^{O(\log^{4d+2} n)}$. Thus, the proof complexity and the query complexity of the PAC-verification protocol are both at most $n \cdot q_{\max} \cdot q_{\max} \cdot M \leq 2^{O(\log^{4d+2} n)}$. Similarly, the verifier running time is also given by $\poly(L) \cdot q_{\max} \cdot M + O(W_{\mathsf{test}} \cdot \poly(L)) \leq 2^{O(\log^{4d+2} n)}$. Finally, the honest prover just answers the queries to $f$ it is asked for, and thus runs in time at most $t \cdot q_{\max} \cdot M \leq 2^{O(\log^{4d+2} n)}$.
\end{itemize}
\end{proof}

\paragraph*{The case of $\cAC^0[p]$, for every prime $p > 2$} We construct an interactive proof for agnostic learning $\cAC^0[p]$, for any prime $p > 2$, follows an identical strategy as the case of $\cAC^0[2]$, by building off Theorem \ref{lem:cikk_ac0[p]_prime_learn}. We omit the proof details here.
\begin{theorem}
    \label{thm:aip_ac0_any_prime}
        For prime $p > 2$, let $f : \{0,1\}^n \rightarrow \{0,1\}$ be any function such that $\opt(f,\cAC^0[p]) > 1/n^{\frac{\log(n)}{2}}$. Then, $\cAC^0[p][\poly(n)]$ is $(O(\log^{16pd+8p}(n)),1/10)$-agnostic verifiable over the uniform distribution with the following parameters.
        \begin{equation*}
        \InteractiveParams{\random}{\exp{(O(\log^{4pd+2p} (n)))}}{\exp{(O(\log^{4pd+2p} (n)))}}{\proofonly}{\exp{(O(\log^{4pd+2p} (n)))}}{\exp{(O(\log^{4pd+2p} (n)))}}
        \end{equation*}
\end{theorem}

\paragraph*{Tolerant Natural Properties Imply PAC-verification}
We finally show that for any typical circuit classes $\sC$, $(1/2-1/\poly(n))$-tolerant natural properties against $\sC$ imply PAC-verification protocols for $\sC$, that use sub-exponentially many random examples. This builds off the learner from \cref{lem:cikk_general_circuits}.

\begin{theorem}
    \label{thm:tolerant_natural_implies_aip}
    Let $\sR$ be a $(1/2-1/\poly(n))$-tolerant natural property useful against $\sC[\poly(n)]$, where $\sC$ is any circuit class that contains $\cAC^0[2]$.

    Then, for any $a \geq 1$ and $\gamma > 0$, such that $\dist(f,\sC[n^a]) = \beta^* \geq \frac{1}
    {n^{2a\gamma}}$, $\sC[\poly(n)]$ is $(O(n^\gamma), 1/10)$-agnostic verifiable over the uniform distribution with the following parameters.
    \begin{equation*}
    \InteractiveParams{\random}{2^{O(n^\gamma)}}{2^{O(n^\gamma)}}{\proofonly}{2^{O(n^\gamma)}}{2^{O(n^\gamma)}}
    \end{equation*}
\end{theorem}

\begin{proof}[Proof sketch]
    This interactive proof also follows a similar strategy as above. Worth highlighting here is that the query construction still generates all strings in a set of subcubes based on the NW-set design. However, each of the $2^{\vert \sE_{u,v} \vert}$ vectors in a subcube $\sE_{u,v}$, where $1 \leq u \leq i-1$ and $1 \leq v \leq k$, are added to the same output of $[\mathsf{PI}(u,z_2)]_v$ (i.e., the $v^{\text{th}}$ block of $\mathsf{PI}(u,z_2)$ of length $n$) to generate the actual set of queries to $\mathsf{Amp}^f_k$.

    The completeness analysis remains unchanged. On the other hand, even if the malicious prover looks at each query and finds the output of $\mathsf{PI}(u,z_2)$ that is added to it, by \cref{lem:nw_query_marginal}, the query marginal of each query is still indistinguishable from the distribution of the embedded random example projected onto the indices of the subcube corresponding to the query. Thus, the analysis extends to this case and shows that the interactive proof satisfies the soundness requirement.
\end{proof}


\section{PAC-Verification for Juntas}
\label{sec:pac-verify-juntas-formal}

In this section we will use our general agnostic verification techniques from \cref{sec:general-techniques} to obtain IPs for the class of juntas. We start by presenting an interactive protocol in \cref{sec:juntas-membership} for learning juntas with membership queries. Following that, in \cref{sec:juntas-random}, we prove the formal version of \cref{thm:our-results-juntas-informal}:

\begin{theorem}[Formal version of \cref{thm:our-results-juntas-informal}]\label{thm:formal-statement-junta-random}
     The class of $k$-juntas $\Junta_{n,k}$ is $(\eps,\delta)$-PAC Verifiable with respect to the uniform distribution with random examples with the following parameters:
    \begin{equation*}
        \InteractiveParams
        {\random}
        {2^k \cdot \poly(k, 1/\eps) \cdot \log^2(1/\delta)}
        {n\cdot\poly(2^k,1/\eps, \log(1/\delta))}
        {3}
        {\poly(n, 2^{k^2}, 1/\eps^k, 1/\delta^k)}
        {\poly(n^k, 2^k, 1/\eps, \log(1/\delta))}
        \enspace.
    \end{equation*}
\end{theorem}

\subsection{PAC-Verification of Juntas with Membership Queries}\label{sec:juntas-membership}

Our basic plan for proving \cref{thm:our-results-junta-membership} will be to show the existence of a distance estimator for the class $\Junta_{n,k}$, and then apply our general IP from \cref{thm:distance-estimator-implies-IP}.

\begin{theorem}\label{thm:our-results-junta-membership}
     The class of $k$-juntas $\Junta_{n,k}$ is $(\eps,\delta)$-PAC Verifiable with respect to the uniform distribution with membership queries with the following parameters:

     \begin{equation*}
         \InteractiveParams
         {\membership}
         {2^k \cdot \poly(k, 1/\eps) \cdot \log(1/\delta)}
         {k \log(n) + 2^k}
         {\proofonly}
         {\poly(n, 2^{k^2}, 1/\eps^k, 1/\delta^k)}
         {\poly(n^k, 2^k, 1/\eps, \log(1/\delta))}
         \enspace.
     \end{equation*}
 \end{theorem}

\begin{proof}[\proofof{\cref{thm:our-results-junta-membership}}]

\cite{DeMN19} show the existence of $(c_u, c_\ell)$-noise tolerant test with query complexity $2^k\cdot \poly(k, 1/\abs{c_u-c_\ell})$ for every $1/2 > c_u > c_\ell \ge 0$. From \cref{clm:distance-estimator-from-tolerant-test} it immediately follows that the class of $k$-juntas $\Junta_{n,k}$ has an $(\eps, \delta)$-distance estimator with query complexity $q(\eps,\delta) = 2^k\cdot\poly(k,1/\eps)\cdot\log(1/\delta)$. Indeed the constraint $c_u < 1/2$ is not an issue for transforming this test into a distance estimator since $\Junta_{n,k}$ is closed under complement and any function $f$ has distance at most $1/2$ from this class.

Now that we have an $(\eps, \delta)$-distance estimator for $k$-juntas, applying \cref{thm:distance-estimator-implies-IP} will give an IP for the class $\Junta_{n,k}$ with the same query complexity $q(\eps/6, \delta/2)$, with additional $\Chernoff(\eps/6, \delta/2)$ random samples. The verifier's running time is inherited from the distance estimator of \cite{DeMN19}, and the honest prover's running time is determined by running a learning algorithm as discussed in \cite{HS23}.
\end{proof}

\subsection{PAC-Verification of Juntas with Random Examples}\label{sec:juntas-random}
Our IP construction for agnostically verifying class of juntas uses the tolerant test for this class introduced in \cite{DeMN19}. Our goal is to transform this protocol into a protocol where verifier uses only random examples. Thus we need to argue that the queries of the tolerant test admit the property of ``constructing query set around a random query point'' (in the sense of \cref{thm:membership-to-random-transform}).

In this section we define two types of query patterns for which it is possible to construct membership query sets, and we further argue that all the queries of the junta's tolerant test are of these form. 

\begin{definition}[Linear Query Pattern]\label{def:linear-query-pattern}
Let $A\sim \Field_2^{n\times k}$ be a uniformly random matrix, and let $B\in \Field_2^{n\times t}$. Then, we say query set $Q=\{q_1,\dots, q_m\}$ is \emph{linear} if for each $i\in [m]$ there exists $\alpha_i \neq 0 \in \Field_2^k$ and $\beta_i\in\Field_2^t$ such that $q_i = A\alpha_i + B\beta_i$.
\end{definition}

It is easy to see that individual queries in a linear query pattern are uniformly distributed. Next, we want to show that we can hide a known sample inside a linear query pattern while preserving the linear dependence between the queries. Note that union of such query patterns also maintains this hiding property.

\begin{lemma}\label{lem:linear-query-pattern}
    For any linear query pattern, it is possible to construct a query set around a random query.
\end{lemma}
\begin{proof}
    Let $Q = \{q_1,\dots,q_m\}$ be a linear query pattern such that $q_i = A\alpha_i + B\beta_i$ 
 for all $i$ where $\alpha_i\in\Field_2^k$, $\beta_i\in\Field_2^t$ are the corresponding values as in \cref{def:linear-query-pattern}. 
    We follow our general template (described in \cref{thm:membership-to-random-transform}) for constructing such a query set. Pick a uniformly random $j\sim[m]$ and a uniformly random vector $r\sim\Field_2^n$ which corresponds to the $j$-th query; that is $q_j = r$.
    
    Let $L \subseteq [k]$ be the coordinates of $\alpha_j$ that are non-zero. $L$ must be non-empty since $\alpha_j \neq 0$. Choose an arbitrary element $\ell \in L$. To construct $A$, sample all of its columns uniformly at random except the $\ell$-th column. Set the $\ell$-th column of $A$ as the following
    \[
    a_\ell = r + B\beta_j + \sum_{i\in L\setminus\{\ell\}} a_i.
    \]
    It is easy to check that $q_j = r = A \alpha_j + B\beta_j$. Notice that $a_\ell$ is a uniformly random vector. Hence $A$ is a uniformly random matrix and the query set $Q$ is valid.
\end{proof}

We define one more query pattern used for dictator testing (Chapter 7 of \cite{odonnell_analysis_2014}). The $i^{\text{th}}$ dictator is a Boolean function which simply outputs the $i^{\text{th}}$ bit and ignores the rest of the input. This is needed in our conversion from membership queries to random for the class of juntas. 

\begin{definition}[NAE Query Pattern]
    The query set consisting of three vectors $x, y, z \in \Bits^n$ is called NAE if for all $i\in[n]$ the bits $(x_i, y_i, z_i)$ are drawn uniformly at random from $\Bits^3\setminus\{(0,0,0), (1,1,1)\}$.
\end{definition}

\begin{lemma}
    It is possible to construct a membership query set for a NAE query pattern.
\end{lemma}

\begin{proof}
    First, notice that each of the three query points is uniformly distributed. To construct a membership query set, pick one of the $x, y$ or $z$ at random (w.l.o.g. assume it is $z$) and assign a uniformly random string to it; in our constructions this will be the point we will hide in the query set. Then assign a uniformly random string to $y$. 
    
    Then for each $i\in [n]$ assign $x_i = 1-y_i$ or $x_i = 1-z_i$ each with probability $1/2$. This set of values for $x, y, z$ will be a NAE query pattern where each query point is uniformly distributed.
\end{proof}

Next, we prove our theorem about agnostic verification of juntas with random examples only. 

\begin{proof}[\proofof{\cref{thm:formal-statement-junta-random}}]
Apply the general query-to-sample reduction of \cref{thm:membership-to-random-transform} to \cref{thm:our-results-junta-membership}. We show this is possible by showing the query pattern of the latter protocol admits the property of ``constructing membership query set''. In \cite{DeMN19}, the tolerant test for the class of juntas $\Junta_{n,k}$ is a straightforward use of \textsf{Maximum-correlation-junta} where, given $k\in\N$, $\eps$, and oracle access to $f:\{-1,1\}^n\to\{-1,1\}$, outputs a number $\widehat{\mathsf{Corr}}_{f,k}$ that is an estimation of the correlation of $f$ with the closest junta from $\Junta_{n, k}$. At some point in this algorithm, we need to use a dictator test on the function $g(x) = \operatorname{sgn}(\Hastad_\eta f|_\rho(x) - \E[f|_\rho]^3)$, where $\rho$ is some random restriction and
\[
\Hastad_\eta f(x) = \E_{y_1, y_2 \sim \{-1,1\}^n, y_3 \sim Z_\eta}[f(y_1)f(y_2)f(x \oplus y_1 \oplus y_2 \oplus y_3)].
\]
Here, $Z_\eta$ denotes the product distribution on $\{-1,1\}^n$ with the $i$-th bit having expectation $\eta_i$ for $\eta \in [-1,1]$. To evaluate $g(x)$ it is enough to estimate $\Hastad_\eta f|_\rho(x)$ and $\E[f|_\rho]^3$. The queries used for these estimations have linear query pattern (assuming we treat the domain as $\Field_2^n$). Moreover, the dictator test uses NAE query pattern. Other components of this algorithm also rely on queries chosen uniformly at random. Therefore, overall, the query pattern for this algorithm is the union of query patterns for which we know how to construct query sets and hide known samples in. Thus, applying the transformation from \cref{thm:membership-to-random-transform} to \cref{thm:our-results-junta-membership} we get a protocol that only uses the random example oracle. 
\end{proof}


\section{Degeneracy of PAC-Verification with Unbounded Provers}
\label{sec:aip_erm_unbounded_prover}
In this section we obtain a general framework for learning hypothesis classes with few random examples, given the honest prover is unbounded. First we show any hypothesis class $\H$ is PAC-verfiable with $O(\VC(\H))$ random examples in \cref{sec:fully-supervised-ip}. Then we show how to decrease the number of labeled samples if the verifier is allowed to draw unlabeled samples from the underlying distribution to obtain the following result.

\begin{theorem}\label{thm:aip-semi-supervised-erm}
    For any finite hypothesis class $\H_n$ of circuits, there exists an $(\eps,\delta)$-PAC-verification with verifier $V$ and prover $P$ with respect to an unknown distribution $\D$, where the verifier has random example access $\mathsf{Ex}(\D, f)$ as well as access to unlabeled samples from $\D$. Let $m = \UCsamples_\H(\frac{\eps}{2}, \frac{\delta}{2}) = O((\VC(\H_n) + \log(1/\delta))/\eps^2)$.  This IP has the following parameters
    \[
    \InteractiveParams{\random}
                      {\log(1/\delta)/\eps}
                      {\log(|\H|) + mn + \log(1/\delta)\cdot \poly(\log m, \log |\H|)}
                      {\log(1/\delta)\cdot \poly(\log m, \log |\H|)}
                      {\log(1/\delta)\cdot m\cdot n\cdot \poly(\log m, \log |\H|)}
                      {\log(1/\delta)\cdot \poly(m, |\H|)}.
    \]
    The interaction $\ip{P,V}$ outputs $h$, which is either $\bot$ (reject) or a hypothesis in $\H_n$ with the following guarantees
    \begin{itemize}
        \item \textbf{Completeness:} There exists an honest prover $P^*$ such that the output of interaction $\ip{P^*,V}=h$ satisfies $\Pr[\dist_\D(h, f) \le \dist_\D(\H, f) + \eps] \ge 1-\delta$.
        \item \textbf{Soundness:} For any (possibly unbounded) prover $P$ the output of interaction $\ip{P,V}=h$ satisfies $\Pr[(h\neq\bot) \land (\dist_\D(h, f) > \dist_\D(\H, f) + \eps)] \le \delta.$
    \end{itemize}
\end{theorem}

For perspective, this construction is very powerful; we get an interactive proof for \textit{proper} agnostic learners in the \textit{distribution-free} setting, using a polynomial time verifier and $O(1/\varepsilon)$ many labeled examples. In more detail, our interactive proof outputs a good hypothesis over any unknown, but fixed marginal distribution $\sD$ over the input examples, using a random example oracle over $\sD$, and further, the hypothesis is from the same class as the target. The implication here is that, if we allow the honest prover to have unbounded computational power, then any target class $\sC$ having a bounded description length for each input length has an efficient proof system, for a very general agnostic learning task. 

As applications of this general framework, we consider hypothesis classes widely explored in learning theory, like $k$-juntas or functions computable by polynomial-sized circuit families, i.e., $\Ppoly$, and construct interactive proofs that can learn such classes using $O(1/\varepsilon)$ many random labeled examples using polynomial-time verifiers.

Regarding the notation for this section, the hypothesis (or the target) class $\H$ is given by a family $\{\H_n\}_{n\in\N}$, where each $\H_n$ can be thought of as a class of circuits defined over $n$ input bits.  Thus, for most natural circuit classes $\H_n$ (like $\cAC^0[p], \mathsf{TC}^0, \mathsf{NC}^1$, or $\Ppoly$), the size and (for $\Ppoly$) the depth of the circuits that compute the functions are bounded by $\log(|\H_n|)$. Furthermore, we consider the class of $k$-juntas, which can be implemented as small circuits. 

Let $f:\Bits^n\to\Bits$ be the input function. The \emph{empirical error} of a hypothesis $h$ w.r.t. a labeled sample set $S=\{(x_i, y_i)\}_{i=1}^m$ is defined to be $\err_S(h) = \Pr_{(x, y)\sim S}[h(x) \neq y]$. Note that the labels in a set of samples need not come from the input function $f$ (useful when a possibly misleading prover labels the samples). The \emph{true error} of a hypothesis $h$ w.r.t. the input function $f$ is $\dist_\D(h, f)=\Pr_{x\sim\D}[h(x)\neq f(x)]$. For simplicity of notation we define $\opt_S(\H) = \min_{h\in\H}\err_S(h)$ and $\dist_\D(\H, f) = \min_{h\in\H}\dist_\D(h, f)$.

\paragraph*{Empirical Risk Minimization (ERM):} An ERM algorithm finds a hypothesis with minimal empirical error. The main idea of ERM is that minimizing empirical error leads to minimizing true error. The learner does not know the underlying distribution, but it can measure the performance of a given hypothesis on random examples. Thus, instead of minimizing the true error, the learner tries to minimize the empirical error. It is well-known that for any hypothesis class that is not too complicated (i.e. has finite VC dimension), an ERM learner  over large enough samples outputs a hypothesis with small true error as well (with high probability).  

\noindent We start with the following definitions on representative samples from $\sD$.

\begin{definition}[$\eps$-Representative Samples \cite{shalev2014understanding}]
A set of samples $S$ is \emph{$\eps$-representative} w.r.t. hypothesis class $\H$, distribution $\D$, and function $f$ if for all $h\in \H$ it holds that $|\err_S(h)-\dist_\D(h, f)|\le\eps$.
\end{definition}

Learning via uniform convergence says that if $|S|$ is sufficiently large then a hypothesis $h$ with the minimum empirical error also has a small true error with high probability. In fact, it states something stronger: the empirical error of all hypotheses is close to their true error with high probability. We recall the following well-known fact from learning theory.

\begin{fact}[ERM for Hypothesis Classes of Finite VC Dimension \cite{shalev2014understanding}]\label{fact:uniform-convergence}
    For any $n$, let $\H = \H_n$ be any hypothesis class of finite VC-dimension, and let $S$ be a set of $\UCsamples_\H(\eps, \delta) = \frac{C}{\eps^2}\cdot\left(\VC(\H) + \log\frac{1}{\delta} \right)$ i.i.d. labeled samples from some unknown distribution $\mathsf{Ex}(\D, f)$, where $C$ is a universal constant. Then $S$ is $\frac{\eps}{2}$-representative with probability at least $1-\delta$. Moreover, an ERM algorithm minimizing $\err_S$ achieves agnostic $(\eps, \delta)$-PAC learning. 
\end{fact}

We will show for an ERM algorithm to be a successful agnostic PAC learner, an \textit{almost} empirically optimal hypothesis suffices, hence the following definition.

\begin{definition}\label{def:almost-empirically-optimal}
For a fixed $n$, let $\H=\H_n$ be a hypothesis class of finite size, and $S$ a finite set of labeled samples.  We define the language $L^\eps_S \subseteq \H$ as the set of empirically optimal hypotheses with respect to $S$ up to $\eps$ error. That is $L_S^\eps = \{h \in \H : \err_{S}(h) \le \opt_{S}(\H) + \eps \}$. We informally refer to elements of $L_S^\eps$ as almost empirically optimal.
\end{definition}

In \cref{sec:fully-supervised-ip}, we present the main idea of using ERM for distribution free PAC-verification of general circuit classes, where the verifier has only random example access. In \cref{sec:semi-supervised-ip} we further improve the query complexity of the verifier and apply the resulting generic protocol to polynomial sized circuits and juntas. 

\subsection{PAC-Verification for $\H_n$ Using $O(\VC(\H_n))$ Random Examples}\label{sec:fully-supervised-ip}
The high level idea is that verifier and prover agree on a sample set $S$ with i.i.d. samples. The prover provides an almost empirically optimal hypothesis $h$ with respect to the samples (\cref{def:almost-empirically-optimal}). Then it only remains for the verifier to check that $h$ is, in fact, almost empirically optimal; i.e., $h\in L_S^\eps$ for some error parameter $\eps$. Next we show that this statement reduces to an instance of Circuit Value Problem (CVP), and that we can apply an existing interactive proof for this statement. Alas, the prover is not efficient, and the resulting protocol is not doubly-efficient.

We start by observing that an almost empirically optimal hypothesis can be verified by a low-depth parallelised computation, in the following lemma.

\begin{lemma} [Reducing $h\in L_S^\eps$ to an Instance of CVP]\label{rem:erm-circuit}
Fix a hypothesis class $\H = \H_n$, a sample set $S = \{(x_i, y_i)\}_{i=1}^m\subseteq\Bits^n\times\Bits$, and a parameter $\eps > 0$. Then, there exists a log-space uniform circuit $C_\mathsf{ERM}^h:\Bits^{mn+m}\to\Bits$ of size $\poly(m, |\H|)$ and depth $\poly(\log(m), \log(|\H|))$ such that $h\in L_S^\eps$ if and only if $C_\mathsf{ERM}^h(x_1,y_1,\dots,x_m,y_m)=1$.
\end{lemma}
\begin{proof}
The inputs to $C_\mathsf{ERM}^h$ are the samples $x_i\in\Bits^n$ along with $y_i\in\Bits$. Give as input each $x_i$ to each circuit in $\H_n$. In other words, test all samples on all hypotheses in parallel. For any $h'\in\H$, define its agreement with samples to be $\mathsf{agr}(h')=|\{i\in[m] : h'(x_i) = y_i\}|$. Then, take the outputs of the circuit $h^\star \in \H$ that maximizes the agreement and compare this number to the number of accepting samples from $h$ and output $1$ if and only if $\mathsf{agr}(h) + \eps m \ge \mathsf{agr}(h^\star)$. Observe that the circuit outputs $1$ if and only if $h \in L^\eps_S$. Therefore, the statement $h \in L^\eps_S$ reduces to an instance of Circuit Value Problem, namely the value of $C_\mathsf{ERM}^h$ on input $S$.

The circuit as described above needs $\poly(m, |\H|)$ size as we run $m$ samples on each hypothesis in $\H$. Note that the rest of the circuit is not any larger. The depth of the circuit is bounded by $\poly(\log(m),\log(|\H|))$; each circuit in $\H$ can be implemented in depth $\log(|\H|)$, and we need no more than $\poly(\log m, \log |\H|)$ to compute the agreement and find the maximum.

We can argue the circuit is log-space uniform since each level of the circuit is tasked with the same computation; in other words, this circuit has a highly regular wiring pattern.
\end{proof}

Following this, we state the following doubly-efficient interactive proof that verifies the result of the computation of any language computable by a low-depth circuit by an untrusted prover. In our case, the largeness of the circuit $C_\mathsf{ERM}$ leads to inefficient prover time.

\begin{lemma}[GKR Protocol \cite{GoldwasserKR08}]\label{lem:GKR08}
    Let $L$ be a language that can be computed by a family of $O(\log(S(n)))$-space uniform Boolean circuits of size $S(n)$ and depth $d(n)$. Then $L$ has an interactive proof with perfect completeness, and soundness $1/2$, where the prover's running time is $\poly(S(n))$ and the verifier runs in time $(n + d(n))\cdot \polylog(S(n))$. Moreover, the communication complexity of the interactive proof is $d(n) \cdot \polylog(S(n))$.
\end{lemma}
\begin{remark}
    The round complexity of the GKR protocol is $d(n)\cdot\polylog(S(n))$.
\end{remark}
\begin{remark}
    We can set the soundness of the GKR protocol to any $\delta > 0$ by repeating the interaction for $\log(1/\delta)$ times and rejecting if at least one of the iterations rejects.
\end{remark}
Denote by $\ip{P_\mathsf{GKR}, V_\mathsf{GKR}}[\delta]$ the interactive protocol of GKR with soundness error boosted to $\delta$. The output of this interaction on circuit $C^h_\mathsf{ERM}$ constructed based on \cref{rem:erm-circuit} for the statement $h\in L_S^\eps$ is denoted by $\ip{P_\mathsf{GKR}, V_\mathsf{GKR}}[\delta](C^h_\mathsf{ERM}, S)$; note that the hidden parameters like $\H$, $\eps$, etc. are captured by the construction of our circuit. We now describe the general PAC-verification proof system for any circuit class $\H$ using $O(\VC(\H))$ random examples. 

\begin{protocol}[H]
\setstretch{1.15}
\caption{A general $(\eps,\delta)$-PAC-Verification protocol for class $\sH$}\label{alg:supervised-protocol}

    \SetKwInOut{FunIn}{Function input}
    \SetKwInOut{ExpIn}{Explicit inputs}
    \FunIn{A function $f:\Bits^n \to \Bits$, where the verifier has access to the random example oracle $\mathsf{Ex}(\D, f)$}
    \ExpIn{Parameters $\eps, \delta > 0$}

    The verifier takes a set $S = \{(x_i, f(x_i))\}_{i=1}^m$ of $m = \UCsamples_\H(\frac{\eps}{2}, \delta) = O(( \VC(\H) + \log(1/\delta))/\eps^2)$ labeled samples from the random example oracle $\mathsf{Ex}(\sD,f)$ and sends them to prover.
    
    The prover sends $\tilde{h}\in\H$, which is purported to be empirically optimal up to error $\eps/2$. That is $\tilde{h} \in L_{S}^{\eps/2}$.

    Let $C^{\tilde{h}}_\mathsf{ERM}$ denote the circuit from \cref{rem:erm-circuit} for the statement $\tilde{h} \in L_{S}^{\eps/2}$. The verifier and the prover run the protocol $\ip{P_\mathsf{GKR}, V_\mathsf{GKR}}[\delta](C_\mathsf{ERM}^{\tilde{h}}, S)$.
    
    The verifier outputs $\tilde{h}$ iff the GKR protocol does not reject and outputs $1$.
\end{protocol}

\begin{lemma}\label{lem:LS-implies-PAC}
    Let $S = \{(x_i, f(x_i))\}_{i=1}^m$ be a set of samples of size $m = \mathsf{m}_\H(\frac{\eps}{2}, \delta) = O(( \VC(\H) + \log(1/\delta))/\eps^2)$ drawn i.i.d. from $\mathsf{Ex}(\D, f)$. Then for any $h\in L_{S}^{\eps/2}$, $\dist_\D(h, f) \le \dist_\D(\H, f) + \eps$ with probability at least $1-\delta$.
\end{lemma}
\begin{proof}
    Let $h_{\mathsf{emp}}$ be an empirically optimal hypothesis w.r.t. samples $S$, and let $h_{\mathsf{true}}$ be an optimal hypothesis w.r.t. the underlying distribution $\D$. By \cref{fact:uniform-convergence} with probability at least $1-\delta$, the sample set $S$ is $\frac{\eps}{4}$-representative, and thus for any $h\in L_{S}^{\eps/2}$
    \[
    \dist_\D(h, f) \le \err_S(h) + \frac{\eps}{4} \le \err_S(h_\mathsf{emp}) + \frac{\eps}{2} + \frac{\eps}{4} \le \err_S(h_\mathsf{true}) +  \frac{\eps}{2} + \frac{\eps}{4} \le \dist_\D(h_\mathsf{true}, f) + \eps.
    \]

    The first and last inequalities follow from $S$ being $\frac{\eps}{4}$-representative. The second one follows from $h\in L^{\eps/2}_S$. The third one follows from $h_\mathsf{emp}$ being optimal with respect to $\err_S$.
\end{proof}

\begin{lemma}\label{lem:guarantees-of-supervised-learning}
    Let $h\in\H\cup\{\bot\text{ (reject)}\}$ denote the output of Protocol~\ref{alg:supervised-protocol}. Then this protocol satisfies the following conditions.
    \begin{itemize}
        \item \textbf{Completeness:} There exists an honest prover $P^*$ such that the output of interaction $\ip{P^*,V}=h$ satisfies $\Pr[\dist_\D(h, f) \le \dist_\D(\H, f) + \eps] \ge 1-\delta$.
        \item \textbf{Soundness:} For any (possibly unbounded) prover $P$, the output of interaction $\ip{P, V}=h$ satisfies $\Pr[(h\neq\bot) \land (\dist_\D(h, f) > \dist_\D(\H, f) + \eps)] \le \delta$.
    \end{itemize}
\end{lemma}
\begin{proof}
    We start with completeness. Assuming the prover is honest, he sends a hypothesis $\tilde h$ that is in $L^{\eps/2}_S$. Then by \cref{lem:LS-implies-PAC}, with probability at least $1-\delta$, we have $\dist_\D(\tilde{h}, f) \le \dist_\D(\H, f) + \eps$. Since $\tilde{h}\in L^{\eps/2}_S$ and since $\ip{P_\mathsf{GKR}, V_\mathsf{GKR}}[\delta]$ has perfect completeness, the verifier accepts the interaction and outputs $h=\tilde{h}$ with probability $1$. Therefore, with probability at least $1-\delta$, $\dist_\D(h, f) \le \dist_\D(\H, f) + \eps$.
    
    Now to prove soundness there are two scenarios. First, suppose that the prover sends $\tilde{h}$ such that $\tilde{h}\notin L^{\eps/2}_S$. Then by guarantee of $\ip{P_\mathsf{GKR}, V_\mathsf{GKR}}[\delta]$ the verifier accepts with probability no more than $\delta$. In the second scenario, the sample set $S$ is not $(\eps/4)$-representative, and the prover sends $\tilde{h}\in L^{\eps/2}_S$. This will be accepted by the verifier, but the probability of $S$ not being representative is at most $\delta$.
\end{proof}

\subsection{Sample-Efficient PAC-Verification for $\H_n$}\label{sec:semi-supervised-ip}

Now that we have a protocol for PAC-verifying general circuit classes, we obtain a protocol from it, where the verifier uses much fewer labeled samples, alongside \textit{unlabeled} samples. The idea is to delegate the task of labeling unlabeled samples to the prover, where prover is forced to send an \textit{almost} correct labeling of samples for otherwise the verifier catches him with overwhelming probability. Following this, we argue that an almost correct labeling is as good as a correct labeling, as far as the ERM algorithm is concerned.

For additional perspective, this setting of learning using unlabeled samples is referred to as \emph{semi-supervised learning} in literature. This is motivated by real-world scenarios where we are required to train models using a limited number of labeled samples, while unlabeled samples are available for \emph{free} or very low cost. The protocol in this section thus gives us a general interactive proof for verifying agnostic learning of finite\footnote{While the size of $\H_n$ could grow with $n$, for any fixed $n$ it is finite.} target classes in the semi-supervised setting, with the help of a prover that is computationally unbounded. While \cite{GoldwasserRSY21} (Claim 5.2) show that any class of finite \textsf{VC}-dimension has $(\eps,\delta)$-PAC-verification where the verifier uses only $O(\log(1/\delta)/\eps)$ labeled samples, the verifier runs an ERM algorithm, potentially in exponential time, while the honest prover is efficient in the number of samples used by the ERM algorithm. On the other hand, we focus on the scenario of understanding the necessity of doubly-efficient proof systems for PAC-verification - and obtain a highly efficient verifier (polynomial in the logarithm of \textsf{VC}-dimension, as well as the size of the hypothesis class), which delegates the task to an inefficient prover.

The protocol is as follows.

\begin{protocol}[H]
\setstretch{1.15}
\caption{Sample-efficient $(\eps, \delta)$-PAC-verification Protocol for class $\sH$}\label{pcl:semi-supervised-protocol}

    \SetKwInOut{FunIn}{Function input}
    \SetKwInOut{ExpIn}{Explicit inputs}
    \FunIn{A function $f:\Bits^n \to \Bits$, where the verifier has access to the random example oracle $\mathsf{Ex}(\D, f)$ and a source for unlabeled samples $\D$}
    \ExpIn{Parameters $\eps, \delta > 0$}
    
    Let $m = \UCsamples_\H(\frac{\eps}{2}, \frac{\delta}{2}) = O((\VC(\H) + \log\frac{2}{\delta})/\epsilon^2)$ and $q = \frac{8}{\epsilon}\cdot \log(\frac{2}{\delta})$. The verifier takes $q$ labeled samples $\{(x_i, f(x_i))\}_{i=1}^q$ from $\mathsf{Ex}(\D, f)$ and $m-q$ unlabeled samples $\{x_i\}_{i=q+1}^m$ from $\D$. Then, the verifier sends $\{x_{\pi(i)}\}_{i=1}^m$ to the prover, where $\pi$ is a uniformly random permutation.

    The prover sends $\tilde{S}=\{(x_{\pi(i)}, \tilde{f}(x_{\pi(i)}))\}_{i=1}^m$ along with a hypothesis $\tilde{h}\in\H$, where $\tilde{f}$ is prover's labeling strategy.\footnote{The prover only needs to send labels $\tilde{f}(x_{\pi(i)})$. For simplicity, we say the prover sends $x_{\pi(i)}$'s as well.}
    
    The verifier checks if the labels of $q$ known samples are correct, if not she rejects. More precisely, she asserts $f(x_i) = \tilde{f}(x_i)$ for all $i\in[q]$.

    Let $C^{\tilde{h}}_\mathsf{ERM}$ denote the circuit from \cref{rem:erm-circuit} for the statement $\tilde{h} \in L_{\tilde{S}}^{\eps/4}$. The verifier and the prover run $\ip{P_\mathsf{GKR}, V_\mathsf{GKR}}[\frac{\delta}{2}](C_\mathsf{ERM}^{\tilde{h}}, \tilde{S})$. 

    The verifier outputs $\tilde{h}$ iff the GKR protocol does not reject and outputs $1$.
\end{protocol}

\begin{lemma}\label{lem:labels-soundness}
    Let $S=\{(x_i, f(x_i))\}_{i=1}^m$ be the set of correct labeling of samples in Protocol~\ref{pcl:semi-supervised-protocol}, and let $\tilde{S} = \{(x_i, \tilde{f}(x_i))\}_{i=1}^m$ be the set of samples labeled by the prover. If the verifier does not reject the samples $\tilde S$ labeled by the prover, it holds with probability at least $1-\frac{\delta}{2}$ that $L^{\eps/4}_{\tilde{S}}\subseteq L^{\eps/2}_S$.
\end{lemma}
\begin{proof}
First, observe that prover has no way to distinguish between the samples that are drawn from the labeled example oracle $\mathsf{Ex}(\sD,f)$ and the unlabeled samples that are drawn from $\sD$, as the samples are all independent and identically distributed. 

Let $p = \frac{1}{m}\cdot |\{i\in[m] : f(x_i) \neq \tilde{f}(x_i)\}|$ be the fraction of mislabeled samples. If $p > \frac{\epsilon}{8}$ then the verifier will catch at least one of them with probability $1 - (1 - p)^q \geq 1 - \exp({-\frac{\epsilon}{8} \cdot \frac{8}{\epsilon}\cdot \log(\frac{2}{\delta})}) = 1 - \frac{\delta}{2}$. Thus, if the verifier does not reject the labeling $\tilde{S}$ provided by the prover, with probability at least $1-\frac{\delta}{2}$, for every hypothesis $h$ we have $|\err_S(h) - \err_{\tilde S}(h)| \le \frac{\eps}{8}$. $(\ast)$

Let $h_\mathsf{\widetilde{emp}}$ and $h_\mathsf{emp}$ be empirically optimal hypotheses with respect to $\tilde{S}$ and $S$ respectively. That is, $\err_{\tilde S}(h_\mathsf{\widetilde{emp}}) = \opt_{\tilde S}(\H)$ and $\err_{S}(h_\mathsf{emp}) = \opt_{S}(\H)$.  Then, using a similar argument as Lemma \ref{lem:LS-implies-PAC}, for any $h\in L^{\eps/4}_{\tilde{S}}$, we have 

\[
\err_S(h) \le \err_{\tilde S}(h) + \frac{\eps}{8} \le \err_{\tilde S}(h_\mathsf{\widetilde{emp}}) + \frac{\eps}{4} + \frac{\eps}{8} \le \err_{\tilde S}(h_\mathsf{emp}) + \frac{\eps}{4} + \frac{\eps}{8} \le \err_S(h_\mathsf{emp}) + \frac{\eps}{2}.
\]

The first and last inequalities hold because of $(\ast)$. The second one follows from $h\in L^{\eps/4}_{\tilde{S}}$. The third one follows from the fact that $h_\mathsf{\widetilde{emp}}$ is empirically optimal with respect to $\tilde S$.
\end{proof}

Intuitively, the above lemma states that cheating on a small fraction of samples does not give any advantage to the dishonest prover. On the other hand, a dishonest prover that cheats on a large fraction of samples is rejected with high probability.

\begin{proof}[\proofof{\cref{thm:aip-semi-supervised-erm}}]
\cref{lem:labels-soundness} shows with probability at least $1-\delta/2$, $L_{\tilde S}^{\eps/4}\subseteq L_S^{\eps/2}$. Therefore, the argument from \cref{lem:guarantees-of-supervised-learning} can be used verbatim to prove similar guarantees.

The multiplicative factor of $\log(1/\delta)$ in communication complexity, round complexity and running time is a result of repeating the GKR protocol for $O(\log(1/\delta))$ times to improve the soundness. The verifier sends $m$ unlabeled samples following which the prover sends a hypothesis and $m$ labels. This needs $\log(|\H|) + m(n+1)$ bits of communication. To run $\ip{P_\mathsf{GKR}, V_\mathsf{GKR}}[\frac{\delta}{2}]$, we need $\log(2/\delta)\cdot\poly(\log m, \log |\H|)$ communication complexity. The verifer's and prover's running time are a direct result of running GKR as well. Note that an honest prover finds a correct hypothesis by running ERM, in a similar way to our $C_\mathsf{ERM}$ circuit.
\end{proof}

We next show applications of \cref{thm:aip-semi-supervised-erm} to get interactive proofs for PAC-verifying specific classes of interest.

\paragraph{Polynomial-sized circuits:} To instantiate \cref{thm:aip-semi-supervised-erm} for $\Ppoly$, fix any polynomial $s = s(n)$. We will consider learning the class $\mathsf{SIZE}[s]$; polynomial sized circuits over $n$ inputs of size at most $s$. Note that there are at most $s^{O(s)}$ circuits of this size.

\begin{corollary}[Formal version of \cref{thm:our-results-df-pac-verify-ppoly}]\label{thm:aip-ppoly-erm}
    For any $c\in\N$, the class $\mathsf{SIZE}[n^c]$ is \textit{distribution-free} $(\eps,\delta)$-PAC-verifiable with verifier $V$ and prover $P$, where for an unknown, but fixed underlying distribution $\sD$, the verifier has random example access via $\mathsf{Ex}(\D, f)$, as well as access to unlabeled samples from $\D$. This interactive proof has the following parameters.
    \[
    \InteractiveParams{\random}
                      {\log(1/\delta)/\eps}
                      {\poly(n, 1/\eps, \log(1/\delta))}
                      {\poly(n, 1/\eps, \log(1/\delta))}
                      {\poly(n, 1/\eps, \log(1/\delta))}
                      { n^{\poly(n)}\poly(1/\eps,\log(1/\delta))}.
    \]
    The interaction $\ip{P, V}$ outputs $h$, which is either $\bot$ (reject) or a circuit of size at most $s$ with the following guarantees
    \begin{itemize}
        \item \textbf{Completeness:} There exists an honest prover $P^*$ such that the output of interaction $\ip{P^*,V}=h$ satisfies $\Pr[\dist_\D(h, f) \le \dist_\D(\mathsf{SIZE}[n^c], f) + \eps] \ge 1-\delta$.
        \item \textbf{Soundness:} For any (possibly unbounded) prover $P$, the output of interaction $\ip{P,V}=h$ satisfies $\Pr[(h\neq\bot) \land (\dist_\D(h, f) > \dist_\D(\mathsf{SIZE}[n^c], f) + \eps)] \le \delta.$
    \end{itemize}
\end{corollary}
\begin{proof}
    Apply the general protocol discussed in \cref{thm:aip-semi-supervised-erm}. Observe that the VC dimension of any hypothesis class is bounded by the logarithm of its size. Thus, we have 
    $$\VC(\mathsf{SIZE}[n^c]) \le \log_2(|\mathsf{SIZE}[n^c]|) = O(n^c\log(n)).$$
    
    Here the number of samples is $m = \mathsf{m}_{\mathsf{SIZE}[n^c]}(\eps/2, \delta/2) = O((n^c\log(n) + \log(1/\delta))\cdot 1/\eps^2)$. The depth of the circuit $C_\mathsf{ERM}$ is bounded by $\poly(n^c, 1/\eps, \log(1/\delta))$ and its size is $n^{cn^c}\cdot\poly(m)$. Thus, by substituting these values in the complexity parameters from \cref{thm:aip-semi-supervised-erm}, we see that the verifier runs in $\poly(n,1/\eps,\log(1/\delta))$, and the honest prover runs in time $n^{\poly(n)}\poly(1/\eps,\log(1/\delta))$.
\end{proof}

\paragraph{Class of $k$-juntas:} Next, we instantiate \cref{thm:aip-semi-supervised-erm} for the class of $k$-juntas, $\Junta_{n,k}$. This \textsf{IP} unlike \cref{thm:our-results-junta-membership} and \cref{thm:formal-statement-junta-random} works for any underlying distribution. The verifier's running time is more efficient than known learning algorithms for $k$-juntas (e.g., \cite{MosselOS03}). 

\begin{corollary}\label{thm:aip-junta-erm}
    The class $\Junta_{n,k}$ is $(\eps,\delta)$-PAC-verifiable with verifier $V$ and prover $P$ in the distribution-free setting, where for an unknown, but fixed underlying distribution $\sD$, the verifier has random example access via $\mathsf{Ex}(\D, f)$ as well as access to unlabeled samples from $\D$. This protocol has the following parameters.
    \[
    \InteractiveParams{\random}
                      {\log(1/\delta)/\eps}
                      {n\cdot \poly(\log(n), 2^k, 1/\eps, \log(1/\delta))}
                      {\poly(\log(n), 2^k, \log(1/\eps), \log(1/\delta)))}
                      {n\cdot \poly(\log(n), 2^k, 1/\eps, \log(1/\delta))}
                      {\poly(n^k, 2^{2^k}, 1/\eps, \log(1/\delta))}.
    \]
    The interaction $\ip{P,V}$ outputs $h$, which is either $\bot$ (reject) or a $k$-junta with the following guarantees
    \begin{itemize}
        \item \textbf{Completeness:} There exists an honest prover $P^*$ such that the output of interaction $\ip{P^*,V}=h$ satisfies $\Pr[\dist_\D(h, f) \le \dist_\D(\Junta_{n,k}, f) + \eps] \ge 1-\delta$.
        \item \textbf{Soundness:} For any (possibly unbounded) prover $P$ the output of interaction $\ip{P,V}=h$ satisfies $\Pr[(h\neq\bot) \land (\dist_\D(h, f) > \dist_\D(\Junta_{n,k}, f) + \eps)] \le \delta.$
    \end{itemize}
\end{corollary}
\begin{proof}
Apply the general protocol discussed in \cref{thm:aip-semi-supervised-erm}. First, note that $|\Junta_{n,k}|=\binom{n}{k}2^{2^k} \le (ne/k)^k 2^{2^k}$. Therefore, $\VC(\Junta_{n,k}) \le k\log(ne) + 2^{k}$ and we need $m = \mathsf{m}_{\Junta_{n,k}}(\eps/2, \delta/2) = O((k\log(n) + 2^k + \log(1/\delta))/\eps^2) = \poly(\log(n), 2^k, 1/\eps, \log(1/\delta))$ samples. By parameters of \cref{thm:aip-semi-supervised-erm} we see that the prover's running time is $\poly(n^k, 2^{2^k}, 1/\eps, \log(1/\delta))$, and the verifier's running time is $n\cdot \poly(\log(n), 2^k, 1/\eps, \log(1/\delta))$. 
\end{proof}

\bibliographystyle{alpha}
\bibliography{references}

\appendix

\section{Embeddability}
\label{sec:app_embed}
In this section we present a general definition of the notion of \textit{embeddability} of a random query in a query sets obtained from membership samples, described in \cref{sec:query-to-sample-reduction}. Recall, when transforming a protocol with membership queries into a protocol that works with  random examples, we try to hide a \textit{single} random example in \textit{non-adaptively} generated query sets for the following definition. 

\begin{definition}[Embeddable non-adaptive query set]
    \label{def:embed2}
    Let $\D$ be a distribution over $\Bits^n$, and let $V_{\mathsf{query}}$ be an \textit{algorithm} that given randomness $r$ returns a sequence of $Q$ queries, $y_1,\dots,y_Q \in \Bits^n$.

    We say that a random query from $\D$ can be embedded into the distribution defined by $V_{\mathsf{query}}$, if there exists an algorithm $\Hat{V}_{\mathsf{query}}$ that takes as input $w \sim \D$ and randomness $r'$, and outputs a query set $\hat{y}_1, \dots, \hat{y}_Q \in \Bits^n$ such that the following hold.

    \begin{enumerate}
    \item There exists $i \in [Q]$ such that $\hat{y}_i = w$, i.e., $w$ is in the output of $\Hat{V}_{\mathsf{query}}$.
        
    \item For every $\hat{y}_1, \dots, \hat{y}_Q$ that is generated by $V_{\mathsf{query}}$ with non-zero probability it holds that
    \begin{equation*}
        \Pr_{r'} \left[\hat{y}_i = w \mid \Hat{V}_{\mathsf{query}}(w;r') \text{ outputs } \hat{y}_1, \dots, \hat{y}_Q\right] = \frac{1}{Q}.
    \end{equation*}
    In other words, the query index into which $w$ is embedded has a uniform distribution, where the probability is over the randomness $r'$ and the sample access to $\D$.
    
    \item For all $\hat{y}_1,\dots, \hat{y}_Q$ it holds that 
    \begin{equation*}
        \Pr_{r'}\left[\Hat{V}_{\mathsf{query}}(w;r')\text{ outputs } \hat{y}_1,\dots, \hat{y}_Q \right] = \Pr_{r}\left[V_{\mathsf{query}}(r) \text{ outputs } \hat{y}_1,\dots, \hat{y}_Q \right].
    \end{equation*}
    In other words, the joint distribution of the queries generated by $\Hat{V}_{\mathsf{query}}$ and $V_{\mathsf{query}}$ coincide.
    \end{enumerate}
    \end{definition}

    In particular, the third condition implies that any (possibly unbounded) prover $P$ can gain no information about the location of the embedded query $w$ by looking at the joint query distributions of the output of $\Hat{V}_{\mathsf{query}}$, and the second condition implies that $P$ can guess the location of $w$ with probability no more than $1/Q$.

    \begin{proposition}[Embeddability is closed under independent unions\footnote{This argument can be easily extended to any finite number of unions.}]\label{prop:embed-closed-under-unions}
        Let $V_{\mathsf{query}}$ and $U_{\mathsf{query}}$ be two Embeddable query generation algorithms. Then the query set generated by the algorithm $T_{\mathsf{query}}(r_1,r_2) = V_{\mathsf{query}}(r_1) \cup U_{\mathsf{query}}(r_2)$ that outputs the concatenation of the independent outputs of $V_\mathsf{query}$ and $U_\mathsf{query}$ is Embeddable.
    \end{proposition}
    \begin{proof}
       Let $Q_1$ and $Q_2$ be the number of queries outputted by $V_\mathsf{query}$ and $U_\mathsf{query}$ respectively. The algorithm $\hat{T}_\mathsf{query}$ works as follows: on input $(w; r_1', r_2')$ with probability $Q_1/(Q_1+Q_2)$ outputs $\hat{V}_{\mathsf{query}}(w;r'_1) \cup U_{\mathsf{query}}(r'_2)$ and with probability $Q_2/(Q_1+Q_2)$ outputs $V_{\mathsf{query}}(r'_1) \cup \hat{U}_{\mathsf{query}}(w;r'_2)$. It is easily observed that $\hat{T}_\mathsf{query}$ satisfies the three items in \cref{def:embed2}.
    \end{proof}
    
    For more context, we sketch how the transformations from our PAC-verification protocols fall into this framework.

    \begin{enumerate}
        \item \textbf{Heavy Fourier coefficients (\cref{thm:fourier-random})} The transformation is presented in \cref{thm:membership-to-random-transform}. Item 2 of the definition is trivially satisfied by picking $i \in [Q]$ uniformly at random and embedding a random example $w$ (drawn from the uniform distribution) in $\hat{y}_i$. For Item 3, first observe from \cref{clm:estimate-of-fourth-power} that all the ``small" query sets are independent of each other. By \cref{prop:embed-closed-under-unions} it suffices to embed $w$ into a small query set as the overall query set is an independent union of embeddable query sets. In turn, Item 3 holds for the small query set since it has a linear query pattern (see \cref{def:linear-query-pattern}), from \cref{lem:linear-query-pattern}, coupled with the fact that the marginal distribution of each query and underlying distribution over $w$ are both uniform distributions.
               
        \item \textbf{$\cAC^0[2]$ (\cref{thm:agnostic_ip_ac0})} The transformation is presented in the proof of \cref{thm:agnostic_ip_ac0}. Observe from \cref{algo:embed_query_iw01}, that a random example $w$ (drawn from the uniform distribution) is placed by sampling a subcube to embed in, with probability given by its density. Item 2 follows by coupling this with the fact that since $w$ is a uniformly random string, the probability of it being any particular element within a subcube query is also uniform. 
        
        Item 3 follows, since the restriction of $w$ (based on its embedded subcube), that is also a uniformly random substring, is used to fix the seed $z$ by the embedded NW-query construction algorithm from \cref{algo:embed_query_iw01} to generate all its subcube queries.
    \end{enumerate}

     We leave the study of potentially embedding multiple random examples in query sets that are adaptive (based on past prover messages and/or verifier queries), and their applicability in query to sample transformations for PAC-verification for future work.
     
\end{document}